\documentclass[runningheads]{llncs}
\usepackage[T1]{fontenc}

\usepackage{tikz}
\usepackage{amsmath}
\usepackage{graphicx} 
\usepackage{amsmath,amsfonts,mathtools}
\usepackage{xspace}
\usepackage{etoolbox}
\usepackage{mdframed}
\usepackage{subfiles}
\usepackage{xr}

\usepackage{enumitem}
\usepackage[square,sort,comma,numbers]{natbib}
\usepackage{multirow}
\usepackage{multicol}
\mdfsetup{frametitlealignment=\center}
\usepackage{breqn}
\usepackage{url}
\usepackage{subcaption}
\usepackage{wrapfig}

\usepackage{hyperref}
\usepackage{makecell}
\usepackage[operators,sets, primitives, oracles]{cryptocode}
\usepackage{tabto}
\usepackage{adjustbox}
\usepackage{booktabs}
\newif\iffullversion
\fullversionfalse

\newcommand{\codeLink}{https://anonymous.4open.science/r/RoadsideCashSystem-4322}
\newcommand{\fullVersionLink}{https://anonymous.4open.science/r/RoadsideCashSystem-4322/paper.pdf}

\newcommand{\notecolor}{blue}
\renewcommand{\thenote}{\thesection.\arabic{note}}
\newcommand{\ac}[1]{\refstepcounter{note}{\bf \textcolor{\notecolor}{$\ll$AC~\thenote: {\sf #1}$\gg$}}}

\newcommand{\sys}{\textit{StreetCred}\xspace}

\newtheorem{defn}{Definition}
\newcommand{\myparagraph}[1]{\medskip\noindent\textbf{#1}:\xspace}

\newenvironment{myquote}%
{\list{}{\leftmargin=0.2in\rightmargin=0.2in}\item[] \em}%
{\endlist}

\newcommand{\kilobytes}{\ensuremath{\mathrm{KB}}\xspace}

\newcommand{\megabytes}{\ensuremath{\mathrm{MB}}\xspace}

\newcommand{\secref}[1]{\mbox{Sec.~\ref{#1}}\xspace}

\newcommand{\secsref}[2]{\mbox{Secs.~\ref{#1}--\ref{#2}}\xspace}

\newcommand{\figref}[1]{\mbox{Fig.~\ref{#1}}\xspace}

\newcommand{\tblref}[1]{\mbox{Table~\ref{#1}}\xspace}

\newcommand{\appref}[1]{\mbox{App.~\ref{#1}}\xspace}

\newcommand{\stepref}[1]{\mbox{Step~\ref{#1}}\xspace}
\newcommand{\stepsref}[2]{\mbox{Steps~\ref{#1}--\ref{#2}}\xspace}

\NewDocumentCommand{\nats}{o o}%
{%
	\IfNoValueTF{#1}%
	{\ensuremath{\mathbb{N}}\xspace}
	{\IfNoValueTF{#2}%
		{\ensuremath{[{#1}]}\xspace}
		{\ensuremath{[{#1},{#2}]}\xspace}
	}
}

\newcommand{\valNotation}[1]{\ensuremath{\expandafter\MakeLowercase{#1}}\xspace}
\newcommand{\setNotation}[1]{\ensuremath{\expandafter\mathcal{#1}}\xspace}
\newcommand{\rvNotation}[1]{\ensuremath{\expandafter\varmathbb{#1}}\xspace}
\newcommand{\fnNotation}[1]{\ensuremath{\expandafter\MakeLowercase{#1}}\xspace}
\newcommand{\algNotation}[1]{\ensuremath{\expandafter\mathtt{\MakeUppercase{#1}}}\xspace}
\newcommand{\constNotation}[1]{\ensuremath{\expandafter\mathsf{#1}}\xspace}
\newcommand{\polyNotation}[1]{\ensuremath{\expandafter\MakeLowercase{#1}}\xspace}
\newcommand{\vecNotation}[1]{\ensuremath{\expandafter\vec{#1}}\xspace}
\newcommand{\protocolNotation}[1]{\ensuremath{\pi}_{#1}\xspace}

\newcommand{\compactProtocol}{\ensuremath{\protocolNotation{\mathsf{comp}}}\xspace}
\newcommand{\noncompactProtocol}{\ensuremath{\protocolNotation{\mathsf{nc}}}\xspace}

\newcommand{\intsMod}[1]{\ensuremath{\mathbb{Z}_{#1}}\xspace}
\newcommand{\intsModStar}[1]{\ensuremath{\mathbb{Z}_{#1}^{*}}\xspace}
\newcommand{\grpOrd}{\ensuremath{p}\xspace}
\newcommand{\secParam}{\ensuremath{\lambda}\xspace}
\newcommand{\getsr}{\ensuremath{\stackrel{\scriptsize\$}{\leftarrow}}\xspace}

\newcommand{\integersModArg}[1]{\ensuremath{\mathbb{Z}_{#1}}\xspace}

\NewDocumentCommand{\grp}{g}{\ensuremath{\GG\IfValueT{#1}{_{#1}}}\xspace}

\NewDocumentCommand{\grpGen}{o}%
{%
	\IfNoValueTF{#1}{\ensuremath{g}\xspace}%
	{\ensuremath{g_{#1}}\xspace}%
}

\newcommand{\nonce}{\ensuremath{\mathsf{nonce}}\xspace}
\newcommand{\intentMsg}{\ensuremath{\mathsf{I}}\xspace}
\newcommand{\invalidList}{\ensuremath{\mathsf{InvList}}\xspace}
\newcommand{\abortWithdrawal}{\ensuremath{\mathsf{AbortWithdrawal}}\xspace}

\newcommand{\sigScheme}{\ensuremath{\protocolNotation{\mathsf{sgn}}}\xspace}

\newcommand{\dsign}[1]{\ensuremath{\sign_{#1}}\xspace}
	\newcommand{\vrfy}[1]{\ensuremath{\mathsf{Vrfy}_{#1}}\xspace}	
\newcommand{\signKeyGen}{\ensuremath{\mathsf{SKGen}}\xspace}	
\newcommand{\sigSchemeDefn}{\ensuremath{\sigScheme = \langle\signKeyGen, \dsign{}, \vrfy{} \rangle}\xspace}

\newcommand{\blindSigScheme}{\ensuremath{\protocolNotation{\mathsf{bsgn}}}\xspace}
\newcommand{\bkeygen}{\ensuremath{\mathsf{BKGen}}\xspace}
\NewDocumentCommand{\bblind}{o}{\ensuremath{\mathsf{BBlind}\IfNoValueF{#1}{_{#1}}}\xspace}
\NewDocumentCommand{\bsign}{g}{\ensuremath{\mathsf{BSign}\IfNoValueF{#1}{_{#1}}}\xspace}
\NewDocumentCommand{\bunblind}{o}{\ensuremath{\mathsf{BUnblind}\IfNoValueF{#1}{_{#1}}}\xspace}
\newcommand{\bvrfy}[1]{\ensuremath{\mathsf{BVrfy}_{#1}}\xspace}
\newcommand{\blindSigSchemeDefn}{\ensuremath{\blindSigScheme = \langle \bkeygen, \bblind, \bsign{}, \bunblind, \bvrfy{} \rangle}\xspace}

\newcommand{\blindSigSchemeZKP}{\ensuremath{\protocolNotation{\mathsf{zsgn}}}\xspace}
\newcommand{\zkeygen}{\ensuremath{\mathsf{ZkGen}}\xspace}
\NewDocumentCommand{\zblind}{o}{%
	\ensuremath{\mathsf{ZBlind}\IfNoValueF{#1}{_{#1}}}\xspace}
\NewDocumentCommand{\zsign}{o}{%
	\ensuremath{\mathsf{ZSign}\IfNoValueF{#1}{_{#1}}}\xspace}
\NewDocumentCommand{\zunblind}{o}{%
	\ensuremath{\mathsf{Ublind}\IfNoValueF{#1}{_{#1}}}\xspace}
\newcommand{\zvrfy}[1]{\ensuremath{\mathsf{ZVrfy}_{#1}}\xspace}
\NewDocumentCommand{\zprove}{g}{\ensuremath{\mathsf{ZProve}\IfValueT{#1}{_{#1}}}\xspace}
\newcommand{\blindSigSchemeZKPDefn}{\ensuremath{\blindSigSchemeZKP = \langle \zkeygen, \zsign{}, \zprove, \zvrfy{} \rangle}\xspace}

\newcommand{\bank}{\ensuremath{\mathcal{B}}\xspace}

\newcommand{\serialNum}{\ensuremath{R}\xspace}
\NewDocumentCommand{\dyPRF}{o}
{
	\IfNoValueTF{#1}
	{\ensuremath{\prf}\xspace}
	{\ensuremath{\prf_{#1}}\xspace}
}
\newcommand{\abort}{\ensuremath{\mathsf{Abort}}\xspace}
\newcommand{\dyPRFKeySpace}{\ensuremath{\integersModArg{\grpOrd}}\xspace}
\newcommand{\dyPRFDefn}{\ensuremath{\dyPRF: \dyPRFKeySpace \times \integersModArg{\grpOrd} \rightarrow \grp{}}\xspace}
\newcommand{\commScheme}{\ensuremath{\algNotation{Comm}}\xspace}
\newcommand{\commSchemeDefn}{\ensuremath{\commScheme: \integersModArg{\grpOrd}^{\numkeys + 1} \rightarrow \grp}\xspace}
\newcommand{\KDF}{\ensuremath{\hashFn'}\xspace}

\newcommand{\rsuWithdrawCoin}{\ensuremath{\textsc{ReqCoin}}\xspace}
\newcommand{\serverInit}{\ensuremath{\textsc{InitBank}}\xspace}
\newcommand{\serverVerifyTx}{\ensuremath{\textsc{UpdateTx}}\xspace}
\newcommand{\serverDblSpend}{\ensuremath{\textsc{DblSpend}}\xspace}
\newcommand{\serverRegisterVehicle}{\ensuremath{\textsc{RegUser}}\xspace}
\newcommand{\serverRegisterATM}{\ensuremath{\textsc{RegATM}}\xspace}
 
 \newcommand{\serverUpdateWithdrawal}{\ensuremath{\textsc{UpdateBal}}\xspace}

 \newcommand{\serverIssueCoin}{\ensuremath{\textsc{IssueCoin}}\xspace}
\newcommand{\rsuInit}{\ensuremath{\textsc{InitATM}}\xspace}
\newcommand{\rsuAssignCoin}{\ensuremath{\textsc{IssueVoucher}}\xspace}
\newcommand{\rsuVerifyCoin}{\ensuremath{\textsc{VerifyTx}}\xspace}

\newcommand{\carInit}{\ensuremath{\textsc{InitUser}}\xspace}
\newcommand{\carDepositCoin}{\ensuremath{\textsc{SpendCoin}}\xspace}
\newcommand{\carGetCoin}{\ensuremath{\textsc{Withdraw}}\xspace}
\NewDocumentCommand{\tx}{o}{\ensuremath{\textsc{tx\IfValueT{#1}{_{#1}}}}\xspace}
\NewDocumentCommand{\txAlt}{o}{\ensuremath{\textsc{tx\IfValueT{#1}{_{#1}}}'}\xspace}

\newcommand{\txID}{\ensuremath{t_{id}}\xspace}
\newcommand{\txIDAlt}{\ensuremath{t_{id}'}\xspace}
\newcommand{\carIssueReceipt}{\ensuremath{\algNotation{IssueReceipt}}}
\newcommand{\genericCoin}{\ensuremath{\mathsf{C}}\xspace}
\newcommand{\genericVoucher}{\ensuremath{\mathsf{V}}\xspace}
\newcommand{\genericTx}{\ensuremath{\mathsf{T}}\xspace}
\newcommand{\genericTxAlt}{\ensuremath{\mathsf{T}'}\xspace}

\newcommand{\rsuKeyStore}{\ensuremath{\mathsf{KS}}\xspace}

\newcommand{\rsuGetKey}{\ensuremath{\mathsf{Get}}\xspace}

\newcommand{\txCoinLabel}{\ensuremath{\coinExt}\xspace}

\newcommand{\coinStore}{\ensuremath{\mathsf{CS}}\xspace}
\newcommand{\vehicleIDList}{\ensuremath{\mathsf{IDS}}\xspace}
\newcommand{\numHonestRSU}{\ensuremath{\numRSUs_{H}}\xspace}	
\newcommand{\advState}{\ensuremath{\phi}\xspace}

\NewDocumentCommand{\genericSig}{o}{\ensuremath{\sigma\IfValueT{#1}{_{#1}}}\xspace}

\newcommand{\hashedCoin}{\ensuremath{K}\xspace}
\newcommand{\signedCoin}{\genericSig[c]\xspace}

\NewDocumentCommand{\vPubKey}{g}{\ensuremath{pk\IfValueT{#1}{_{#1}}}\xspace}
\NewDocumentCommand{\vPrivKey}{g}{\ensuremath{sk\IfValueT{#1}{_{#1}}}\xspace}
\NewDocumentCommand{\vSigPubKey}{g}{\ensuremath{pk\IfValueT{#1}{_{#1}}'}\xspace}
\NewDocumentCommand{\vSigPrivKey}{g}{\ensuremath{sk\IfValueT{#1}{_{#1}}'}\xspace}

\NewDocumentCommand{\vPRFKey}{g}{\ensuremath{k\IfValueT{#1}{_{#1}}}\xspace}
\NewDocumentCommand{\vKeyComm}{g}{\ensuremath{p\IfValueT{#1}{_{#1}}}\xspace}
\NewDocumentCommand{\vPRFTx}{g}{\ensuremath{U\IfValueT{#1}{_{#1}}}\xspace}

\NewDocumentCommand{\rsuPubKey}{g}{\ensuremath{pk^{A}\IfValueT{#1}{_{#1}}}\xspace}
\NewDocumentCommand{\rsuPrivKey}{g}{\ensuremath{sk^{A}\IfValueT{#1}{_{#1}}}\xspace}
\NewDocumentCommand{\rsuSigPubKey}{g}{\ensuremath{pk^{A'}\IfValueT{#1}{_{#1}}}\xspace}
\NewDocumentCommand{\rsuSigPrivKey}{g}{\ensuremath{sk^{A'}\IfValueT{#1}{_{#1}}}\xspace}

\newcommand{\sPubKey}{\ensuremath{\mathit{pk}_{\mathsf{sig}}^{\server}}\xspace}

\newcommand{\sPrivKey}{\ensuremath{\mathit{sk}_{\mathsf{sig}}^{\server}}\xspace}

\newcommand{\sPubKeyBlind}{\ensuremath{\mathit{pk}_{b}^{\server}}\xspace}
\newcommand{\sPrivKeyBlind}{\ensuremath{\mathit{sk}_{b}^{\server}}\xspace}
\newcommand{\sPubKeyBlindZKP}{\ensuremath{\mathit{pk}_{z}^{\server}}\xspace}
\newcommand{\sPrivKeyBlindZKP}{\ensuremath{\mathit{sk}_{z}^{\server}}\xspace}
\newcommand{\sPubKeyBlindZKPAlt}{\ensuremath{\mathit{pk}_{z}^{\server'}}\xspace}
\newcommand{\sPrivKeyBlindZKPAlt}{\ensuremath{\mathit{sk}_{z}^{\server'}}\xspace}

\newcommand{\rsuCert}[1]{\ensuremath{\mathsf{cert}^{A}_{#1}}\xspace}

\NewDocumentCommand{\genCarID}{g}{
	\IfValueTF{#1}
	{\ensuremath{\crHash(\pubKey{i#1})\xspace}}
	{\ensuremath{\crHash(\pubKey{\vID})}}
	}
\newcommand{\genCarIDAlt}{\ensuremath{\crHash(\pubKey{\vID'})}\xspace}

\NewDocumentCommand{\vID}{o}{%
        \ensuremath{i\IfNoValueF{#1}{_{#1}}}\xspace}
\newcommand{\vIDAlt}{\ensuremath{i'}\xspace}
\newcommand{\vIDAltAlt}{\ensuremath{i''}\xspace}
\newcommand{\vIDDep}{\ensuremath{k}\xspace}

\NewDocumentCommand{\genCar}{g}{
\IfValueTF{#1}
{\ensuremath{\carLabel_{#1}}\xspace}
{\ensuremath{\carLabel_{\vID}}\xspace}
}

\NewDocumentCommand{\user}{g}{
\IfValueTF{#1}
{\ensuremath{\carLabel_{#1}}\xspace}
{\ensuremath{\carLabel_{\vID}}\xspace}
}

\newcommand{\auxInfo}{\ensuremath{\mathsf{info}}\xspace}

\newcommand{\rsuID}{\ensuremath{j}\xspace}
\newcommand{\rsuIDAlt}{\ensuremath{j'}\xspace}
\newcommand{\rsuIDAltAlt}{\ensuremath{j''}\xspace}

\NewDocumentCommand{\genRSU}{g}{\ensuremath{\mathit{\rsu}\IfValueTF{#1}{_{#1}}{_{\rsuID}}}\xspace}

\newcommand{\ecash}{e-cash\xspace}
\newcommand{\Ecash}{E-cash\xspace}
\newcommand{\cashSchemeName}{hierarchical cash\xspace}

\newcommand{\genRSUIDAlt}{\ensuremath{j'}}

\NewDocumentCommand{\genRSUID}{o}{%
        \ensuremath{j\IfNoValueF{#1}{_{#1}}}\xspace}

\newcommand{\genDepRSUID}{\ensuremath{m}}
\newcommand{\genDepRSU}{\ensuremath{\carLabel_{\genDepRSUID}}\xspace}
\newcommand{\coinExt}{\ensuremath{C}\xspace}
\newcommand{\carSig}{\ensuremath{\phi}\xspace}

\newcommand{\repCoins}{coins\xspace}

\newcommand{\returnVar}{\ensuremath{rr}\xspace}
\newcommand{\returnVarAlt}{\ensuremath{rr'}\xspace}

\NewDocumentCommand{\repCoinAbbrv}{o}{\ensuremath{\mathsf{RCoin}\IfValueT{#1}{_{#1}}}\xspace}
\NewDocumentCommand{\repCoinsAbbrv}{o}{\ensuremath{\mathsf{Coin}\IfValueT{#1}{_{#1}}}\xspace}
\newcommand{\repCoinScheme}{\ensuremath{\Pi_{\mathsf{coin}}}\xspace}

\NewDocumentCommand{\signedCoinComm}{o}
{
\IfNoValueTF{#1}
{\ensuremath{\genericSig[\coinPRFKey]}\xspace}
{\ensuremath{\genericSig[\coinPRFKey, #1]}\xspace}
}

\newcommand{\setOfValidReceipts}{\ensuremath{\mathcal{R}}\xspace}

\newcommand{\frameAdversaryState}{\ensuremath{\phi}\xspace}

\newcommand{\setOfWithdrawnCoins}{\ensuremath{\setNotation{W}}\xspace}

\newcommand{\bankKeyPair}{\ensuremath{\langle \sPubKey, \sPrivKey \rangle}\xspace}
\newcommand{\genericRSUIDKeyPair}{\ensuremath{\langle \rsuPubKey{\rsuID}, \rsuPrivKey{\rsuID} \rangle}\xspace}
\newcommand{\genericRSUSigKeyPair}{\ensuremath{\langle \rsuSigPubKey{\rsuID}, \rsuSigPrivKey{\rsuID} \rangle}\xspace}
\newcommand{\genericUserIDKeyPair}{\ensuremath{\langle \vPubKey{\vID}, \vPrivKey{\vID} \rangle}\xspace}
\newcommand{\genericUserSigKeyPair}{\ensuremath{\langle \vPubKey{\vID}', \vPrivKey{\vID}' \rangle}\xspace}

\newcommand{\adversaryInputKeyPairs}{\ensuremath{\begin{array}{@{}l@{}} \langle \vPubKey{\vID} \rangle_{\vID \in \nats[1, \numHonestVehicles]},  \langle \vPrivKey{\vID} \rangle_{\vID \in \nats[\numHonestVehicles + 1 , \numVehicles]}, \\ \langle \rsuPubKey{\rsuID} \rangle_{\rsuID \in \nats[1, \numHonestRSU]}, \langle \rsuPrivKey{\rsuID} \rangle_{\rsuID \in \nats[\numHonestRSU + 1, \numRSUs]}\end{array}}\xspace}

\NewDocumentCommand{\signedComm}{o o o}
{
	\ifstrequal{#1}{1}
	{\IfNoValueTF{#2} {\ensuremath{\sigma_{s}}\xspace} {\IfNoValueTF{#3} {\ensuremath{\sigma_{s,#2}}\xspace} {\ensuremath{\sigma_{s,#2,#3}}\xspace}}} 
	{\IfNoValueTF{#2} {\ensuremath{\sigma_{t}}\xspace} {\IfNoValueTF{#3} {\ensuremath{\sigma_{t,#2}}\xspace} {\ensuremath{\sigma_{t,#2,#3}}\xspace}}} 
}

\NewDocumentCommand{\genericVoucherDefn}{o}{
	\IfNoValueTF{#1}
	{\ensuremath{\langle}
        \ensuremath{\prfEval[1][\coinID],}
         \ensuremath{\prfEval[2][\coinID],} \ensuremath{\coinID{},}
         \ensuremath{\coinRandVehicle,} \ensuremath{\coinRand,}
         \ensuremath{\zkpPRF[\coinID]}, \ensuremath{\auxInfo \rangle} \xspace}		
	{\ensuremath{\langle}
	 \ensuremath{\prfEval[1][\coinID],}
	\ensuremath{\prfEval[2][\coinID],} \ensuremath{\coinID{},}
	\ensuremath{\coinRandVehicle,} \ensuremath{\coinRand,}
	\ensuremath{\zkpPRF[\coinID]}, \ensuremath{\auxInfo \rangle} \xspace}
}

\newcommand{\genericVoucherAlt}{\genericVoucher'\xspace}

\NewDocumentCommand{\genericCoinDefn}{o}
{
	\IfNoValueTF{#1}
	{
		\ensuremath{\langle \keyCommitment[1], \signedComm[1] \rangle\xspace}
	}
	{
		\ensuremath{\langle \keyCommitment[1][#1], \signedComm[1][#1]\rangle\xspace}
	}
}

\NewDocumentCommand{\genericTxDefn}{o}
{
	\IfNoValueTF{#1}
	{
		\ensuremath{\langle \txRand, \dblSpendMark{\txID},  \dblSpendCT{\txID}, \txID \rangle\xspace}
	}
	{
		\ensuremath{\langle \langle \vPRFTx{\txID}, \zkpPRF[\txID], \vKeyComm{#1}, \zkpSigVehicle[#1] \rangle \xspace}
	}
}

\NewDocumentCommand{\genericTxAltDefn}{o}
{
	\IfNoValueTF{#1}
	{
		\ensuremath{\langle \txRandAlt, \dblSpendMark{\txIDAlt},  \dblSpendCT{\txIDAlt}, 
		\txIDAlt 
		\rangle\xspace}
	}
	{
		\ensuremath{\langle \langle \vPRFTx{\txID}, \zkpPRF[\txID], \vKeyComm{#1}, 
		\zkpSigVehicle[#1] \rangle \xspace}
	}
}

\NewDocumentCommand{\genericTxSigDefn}{o}
{
	\IfNoValueTF{#1}
	{
		\ensuremath{\langle \vPRFTx{\txID}, \zkpPRF[\txID], \vKeyComm{\vID}, \zkpSigVehicle[\vID] \rangle\xspace}
	}
	{
		\ensuremath{\langle \vPRFTx{\txID}, \zkpPRF[\txID], \vKeyComm{#1}, \zkpSigVehicle[#1] \ranglexspace}
	}
}

\newcommand{\vrfPubKey}{\ensuremath{pk_{v}}}
\newcommand{\vrfPrivKey}{\ensuremath{sk_{v}}}

\newcommand{\vrfSeed}{\ensuremath{\genericRand{}}\xspace}

\NewDocumentCommand{\repCoinDefn}{o}{
	\IfNoValueTF{#1}
	{\ensuremath{\repCoinAbbrv \gets \langle \keyCommitment[1],}
         \ensuremath{\keyCommitment[2],} \ensuremath{\prfEval[1],}
         \ensuremath{\prfEval[2],} \ensuremath{\coinID{},}
         \ensuremath{\coinRandVehicle,} \ensuremath{\coinRand,}
         \ensuremath{\zkpSig[1],} \ensuremath{\zkpSig[2],}
         \ensuremath{\zkpPRF,} \ensuremath{\auxInfo \rangle} \xspace}		
	{\ensuremath{\repCoinAbbrv \gets \langle \keyCommitment[1][#1],}
         \ensuremath{\keyCommitment[2][#1],} \ensuremath{\coinID{},}
         \ensuremath{\prfEval[1][\coinID{}],}
         \ensuremath{\prfEval[2][\coinID{}],}
         \ensuremath{\coinRandVehicle,} \ensuremath{\coinRand,}
         \ensuremath{\zkpSig[1][#1],} \ensuremath{\zkpSig[2][#1],}
         \ensuremath{\zkpPRF[\coinID{}],} \ensuremath{\auxInfo \rangle}\xspace}
}

\newcommand{\symmKey}{\ensuremath{\kappa}\xspace}

\newcommand{\dblSpendMark}[1]{\ensuremath{X_{#1}}\xspace}
\newcommand{\dblSpendCT}[1]{\ensuremath{E_{#1}}\xspace}

\newcommand{\vCert}[1]{\ensuremath{\mathsf{cert}_{#1}}\xspace}
\NewDocumentCommand{\hashFn}{o}{\ensuremath{\fnNotation{H}\IfValueT{#1}{_{#1}}}\xspace}
\newcommand{\crHash}{\ensuremath{\hashFn}\xspace}
\newcommand{\roHash}{\ensuremath{\ro}\xspace}
\newcommand{\genericRand}[1]{\ensuremath{r_{#1}}\xspace}

\newcommand{\coinRand}{\ensuremath{r_c}\xspace}

\newcommand{\coinRandVehicle}{\ensuremath{r_v}\xspace}
\newcommand{\txRand}{\ensuremath{r_{t}}\xspace}
\newcommand{\coinRandAlt}{\ensuremath{r_c'}\xspace}

\newcommand{\txRandAlt}{\ensuremath{r_{t}'}\xspace}

\newcommand{\keyGen}{\ensuremath{\algNotation{K}}\xspace}
\NewDocumentCommand{\encrypt}{g}{%
	\ensuremath{\algNotation{E}\IfNoValueF{#1}{_{#1}}}\xspace}
\NewDocumentCommand{\decrypt}{g}{%
	\ensuremath{\algNotation{D}\IfNoValueF{#1}{_{#1}}}\xspace}

\NewDocumentCommand{\pubKey}{g}{%
	\ensuremath{\mathit{pk}\IfNoValueF{#1}{_{#1}}}\xspace}
\NewDocumentCommand{\privKey}{g}{%
	\ensuremath{\mathit{sk}\IfNoValueF{#1}{_{#1}}}\xspace}
\NewDocumentCommand{\pubSignKey}{g}{%
	\ensuremath{\mathit{pk}^{U}\IfNoValueF{#1}{_{#1}}}\xspace}
\NewDocumentCommand{\privSignKey}{g}{%
	\ensuremath{\mathit{sk}^{U}\IfNoValueF{#1}{_{#1}}}\xspace}

\NewDocumentCommand{\plaintextSpace}{o}{%
	\ensuremath{\setNotation{P}\IfNoValueF{#1}{_{#1}}}\xspace}
\NewDocumentCommand{\ciphertextSpace}{o}{%
	\ensuremath{\setNotation{C}\IfNoValueF{#1}{_{#1}}}\xspace}
\NewDocumentCommand{\plaintext}{o}{%
	\ensuremath{\valNotation{m}\IfNoValueF{#1}{_{#1}}}\xspace}
\NewDocumentCommand{\ciphertext}{o}{%
	\ensuremath{\valNotation{c}\IfNoValueF{#1}{_{#1}}}\xspace}
\NewDocumentCommand{\ciphertextAlt}{o}{%
	\ensuremath{\valNotation{c}'\IfNoValueF{#1}{_{#1}}}\xspace}
\NewDocumentCommand{\msg}{o}{%
	\ensuremath{\valNotation{m}\IfNoValueF{#1}{_{#1}}}\xspace}

\newcommand{\Experiment}[2]{\ensuremath{\textsc{Expt}^{#1}_{#2}}\xspace}
\newcommand{\Advantage}[2]{\ensuremath{\mathbf{Adv}^{#1}_{#2}}\xspace}
\newcommand{\hidingBit}{\ensuremath{\valNotation{b}}\xspace}
\newcommand{\hidingGuess}{\ensuremath{\hat{\valNotation{b}}}\xspace}
\newcommand{\hidingAdversary}{\ensuremath{\algNotation{A}_h}\xspace}
\renewcommand{\oracle}{\ensuremath{\algNotation{O}}\xspace}
\NewDocumentCommand{\hidingOracle}{g g g}{%
	\ensuremath{\algNotation{O}\IfNoValueF{#1}{({#1},{#2},{#3})}}\xspace}
\newcommand{\bindingAdversary}{\ensuremath{\algNotation{A}}\xspace}
\newcommand{\ufcmaAdversary}{\ensuremath{\algNotation{A}_{\mathsf{ufcma}}}\xspace}
\newcommand{\ufcmaBlindAdversary}{\ensuremath{\algNotation{A}_{\mathsf{ufcmab}}}\xspace}

\NewDocumentCommand{\fairexchangeAdversary}{o}{\ensuremath{\algNotation{A}\IfNoValueF{#1}{_{#1}}}\xspace}
\newcommand{\fairexchangeLabel}{\scriptsize\textsf{FairExchange}\xspace}
\newcommand{\fairexchangeComb}{\ensuremath{\fairexchangeAdversary[1], \fairexchangeAdversary[2]}\xspace}
\newcommand{\fairexchangeAdv}{\ensuremath{\Advantage{\fairexchangeLabel}{\repCoinScheme}(\timeBound, \oracleQueries)}\xspace}
\newcommand{\frameAdv}{\ensuremath{\Advantage{\frameLabel}{\repCoinScheme}(\timeBound, \oracleQueries)}\xspace}

\newcommand{\prfAdversary}{\ensuremath{\algNotation{A}_{\mathsf{prf}}}\xspace}
\NewDocumentCommand{\blindSigAdversary}{o}{%
        \ensuremath{\algNotation{A}\IfNoValueF{#1}{_{#1}}}\xspace}
\NewDocumentCommand{\locPrivacyAdversary}{o}{%
        \ensuremath{\algNotation{A}\IfNoValueF{#1}{_{#1}}}\xspace}
\NewDocumentCommand{\anonAdversary}{o}{%
        \ensuremath{\algNotation{A}\IfNoValueF{#1}{_{#1}}}\xspace}
\NewDocumentCommand{\frameAdversary}{o}{%
        \ensuremath{\algNotation{A}\IfNoValueF{#1}{_{#1}}}\xspace}
\newcommand{\blindSigBit}{\ensuremath{\valNotation{b}}\xspace}
\newcommand{\blindSigGuess}{\ensuremath{\hat{\valNotation{b}}}\xspace}
\NewDocumentCommand{\blindSigOracle}{g g g g}{%
	\ensuremath{\algNotation{O}\IfNoValueF{#1}{_{#1}({#2},{#3},{#4})}}\xspace}
\newcommand{\indcpaBit}{\ensuremath{\valNotation{b}}\xspace}
\newcommand{\indcpaGuess}{\ensuremath{\hat{\valNotation{b}}}\xspace}
\newcommand{\indcpaAdversary}{\ensuremath{\algNotation{A}_{\mathsf{cpa}}}\xspace}
\NewDocumentCommand{\unforgeAdversary}{o}{\ensuremath{\algNotation{A}\IfValueT{#1}{_{#1}}}\xspace}
\NewDocumentCommand{\indcpaOracle}{g g g g}{%
	\ensuremath{\algNotation{O}\IfNoValueF{#1}{_{#1}({#2},{#3},{#4})}}\xspace}
\NewDocumentCommand{\dblSpendAdversary}{o}{\ensuremath{\algNotation{A}\IfValueT{#1}{_{#1}}}\xspace}
\newcommand{\dblSpendAdversaryState}{\ensuremath{\phi}\xspace}
\newcommand{\anonBit}{\ensuremath{\valNotation{b}}\xspace}
\newcommand{\anonGuess}{\ensuremath{\hat{\valNotation{b}}}\xspace}
\NewDocumentCommand{\anonLabel}{o}{%
	{\scriptsize\textsf{anon\IfNoValueF{#1}{-{#1}}}}\xspace}
\NewDocumentCommand{\prfLabel}{o}{%
	{\scriptsize\textsf{prf\IfNoValueF{#1}{-{#1}}}}\xspace}
\NewDocumentCommand{\indcpaLabel}{o}{%
	{\scriptsize\textsf{ind-cpa\IfNoValueF{#1}{-{#1}}}}\xspace}
\NewDocumentCommand{\hidingLabel}{o}{%
	{\scriptsize\textsf{hiding\IfNoValueF{#1}{-{#1}}}}\xspace}
\newcommand{\bindingLabel}{\scriptsize\textsf{binding}\xspace}
\newcommand{\ufcmaLabel}{\scriptsize\textsf{uf-cma}\xspace}
\newcommand{\ufcmaBlindLabel}{\scriptsize\textsf{uf-cma-blind}\xspace}
\NewDocumentCommand{\soundnessLabel}{o}{%
        {\scriptsize\textsf{ext\IfNoValueF{#1}{-{#1}}}}\xspace}
\NewDocumentCommand{\blindSigLabel}{o}{%
	{\scriptsize\textsf{b\IfNoValueF{#1}{-{#1}}}}\xspace}
\NewDocumentCommand{\locPrivacyLabel}{o}{%
	{\scriptsize\textsf{loc-priv\IfNoValueF{#1}{-{#1}}}}\xspace}
\newcommand{\unforgeLabel}{\scriptsize\textsf{mint}\xspace}
\newcommand{\unforgeVoucherLabel}{\scriptsize\textsf{ufvoucher}\xspace}
\newcommand{\frameLabel}{\scriptsize\textsf{frame}\xspace}
\newcommand{\dblSpendLabel}{\scriptsize\textsf{dblspnd}\xspace}
\newcommand{\locPrivacyBit}{\ensuremath{b}\xspace}
\newcommand{\locPrivacyGuess}{\ensuremath{\hat{b}}\xspace}
\newcommand{\locPrivacyAdvState}{\ensuremath{\phi}\xspace}
\newcommand{\anonAdversaryState}{\ensuremath{\phi}\xspace}
\newcommand{\dlogLabel}{\scriptsize\textsf{dl}\xspace}
\newcommand{\dlogAversary}{\ensuremath{\algNotation{A}_{\mathsf{dl}}}\xspace}

\newcommand{\timeBound}{\ensuremath{t}\xspace}
\newcommand{\oracleQueries}{\ensuremath{q}\xspace}
\newcommand{\neglProb}{\ensuremath{\fnNotation{negl}(\secParam)}\xspace}
\newcommand{\prob}[1]{\ensuremath{\mathbb{P}\mathopen{}\left({#1}\right)\mathclose{}}\xspace}

\newcommand{\commFn}{\ensuremath{\algNotation{comm}}\xspace}
\newcommand{\genericCommitment}{\ensuremath{com}\xspace}
\NewDocumentCommand{\genericFnFamily}{ o }{\ensuremath{\fnNotation{G}\IfNoValueF{#1}{_{#1}}}\xspace}
\newcommand{\genericFnKeyspace}{\ensuremath{\setNotation{K}}\xspace}
\newcommand{\genericFnKey}{\ensuremath{k}\xspace}
\newcommand{\genericFnDomain}{\ensuremath{\setNotation{D}}\xspace}

\newcommand{\genericFnRange}{\ensuremath{\setNotation{R}}\xspace}
\newcommand{\genericFn}{\ensuremath{\fnNotation{H}}\xspace}

\newcommand{\allFuncs}[1]{\ensuremath{\setNotation{F}({#1})}\xspace}
\newcommand{\prfGuess}{\ensuremath{\valNotation{\hat{b}}}\xspace}
\newcommand{\Adv}[1]{\ensuremath{\algNotation{A}_{#1}}\xspace}

\newcommand{\forwardOracleQueries}{\ensuremath{q}\xspace}

\NewDocumentCommand{\zkp}{ g  g }{\ensuremath{\Psi\IfNoValueF{#1}{_{#1 \IfNoValueF{#2}{, #2}}}}\xspace}
\NewDocumentCommand{\zkpGen}{ o }{\ensuremath{\mathsf{pGen}\IfNoValueF{#1}{_{#1}}}\xspace}
\newcommand{\zkpVerify}{\ensuremath{\mathsf{pVerify}}\xspace}
\newcommand{\zkpSim}{\ensuremath{\mathsf{pSim}}\xspace}
\newcommand{\zkpSimHash}{\ensuremath{\zkpSim.\mathsf{hash}}\xspace}
\newcommand{\zkpSimProve}{\ensuremath{\zkpSim.\mathsf{go}}\xspace}
\newcommand{\zkpExtract}{\ensuremath{\mathsf{pEx}}\xspace}
\newcommand{\zkpExtractHash}{\ensuremath{\zkpExtract.\mathsf{hash}}\xspace}
\newcommand{\zkpExtractGen}{\ensuremath{\zkpExtract.\mathsf{go}}\xspace}
\NewDocumentCommand{\zkpProverAdversary}{ o }{%
        \ensuremath{\algNotation{A}\IfNoValueF{#1}{_{#1}}}\xspace}

\NewDocumentCommand{\zkLabel}{o}{%
	\scriptsize\textsf{zk\IfNoValueF{#1}{-{#1}}}\xspace}
\newcommand{\lang}{\ensuremath{\setNotation{L}}\xspace}
\newcommand{\relation}{\ensuremath{\mathcal{R}_{\lang}}\xspace}
\newcommand{\statement}{\ensuremath{\valNotation{x}}\xspace}
\NewDocumentCommand{\witness}{ g }{\ensuremath{\valNotation{w}\IfNoValueF{#1}{_{#1}}}\xspace}
\newcommand{\zkpName}{\ensuremath{\Pi}\xspace}
\newcommand{\zkpVerifierAdversary}{\ensuremath{\algNotation{D}}\xspace}
\newcommand{\zkpVerifierAdversaryBit}{\ensuremath{\valNotation{b}'}\xspace}
\newcommand{\zkpVerifierAdversaryGuess}{\ensuremath{\hat{\valNotation{b}}'}\xspace}

\newcommand{\randomOracle}{\ensuremath{\fnNotation{h}}\xspace}

\newcommand{\randomOracles}{\ensuremath{\mathcal{H}}\xspace}
\newcommand{\coinSN}{\ensuremath{\mathsf{ctr}}\xspace}
\newcommand{\coinSNCommitment}{\ensuremath{J}\xspace}

\newcommand{\server}{\ensuremath{\bank}\xspace}

\NewDocumentCommand{\locPrivacySecDef}{g}{\textsf{priv}\IfValueT{#1}{[{#1}]}\xspace}

\NewDocumentCommand{\anonSecDef}{g}{\textsf{anon}\IfValueT{#1}{[{#1}]}\xspace}

\newcommand{\numVehicles}{\ensuremath{\eta}\xspace}
\newcommand{\numHonestVehicles}{\ensuremath{\numVehicles_{H}}\xspace}

\NewDocumentCommand{\unforgetxSecDef}{g}{\scriptsize \textsf{Unforge-Tx}\IfValueT{#1}{[{#1}]}\xspace}

\newcommand{\coinVoucherLink}{\ensuremath{\rightleftharpoons}\xspace}

\newcommand{\numRSUs}{\ensuremath{\gamma}\xspace}

\NewDocumentCommand{\genericPRFKey}{o}{\IfNoValueTF{#1}{\genericFnKey}{\ensuremath{\genericFnKey_{#1}}}\xspace}

\newcommand{\assign}{\ensuremath{\gets}\xspace}

\NewDocumentCommand{\rsuPRFKey}{o}
{
	\IfNoValueTF{#1}
	{\ensuremath{s}\xspace}
	{\ensuremath{s_{#1}}\xspace}
}

\NewDocumentCommand{\coinPRFKey}{o}
{
	\IfNoValueTF{#1}
	{\ensuremath{c}\xspace}
	{\ensuremath{c_{#1}}\xspace}
}

\NewDocumentCommand{\userPRFKey}{o}
{
	\IfNoValueTF{#1}
	{\ensuremath{s_{\vID}}\xspace}
	{\ensuremath{s_{#1}}\xspace}
}

\NewDocumentCommand{\prfKey}{o o o}
{
	\ifstrequal{#1}{1}
	{\IfNoValueTF{#2} {\ensuremath{a}\xspace} {\IfNoValueTF{#3} {\ensuremath{a_{#2}}\xspace} {\ensuremath{a_{#2,#3}}\xspace}}} 
	{\IfNoValueTF{#2} {\ensuremath{b}\xspace} {\IfNoValueTF{#3} {\ensuremath{b_{#2}}\xspace} {\ensuremath{b_{#2,#3}}\xspace}}} 
}

\NewDocumentCommand{\blindedPRFKey}{o o o}
{
	\ifstrequal{#1}{1}
	{\IfNoValueTF{#2} {\ensuremath{\hat{c}}\xspace} {\IfNoValueTF{#3} {\ensuremath{\hat{c}_{#2}}\xspace} {\ensuremath{\hat{c}_{#2,#3}}\xspace}}} 
	{\IfNoValueTF{#2} {\ensuremath{\hat{s}}\xspace} {\IfNoValueTF{#3} {\ensuremath{\hat{s}_{#2}}\xspace} {\ensuremath{\hat{s}_{#2,#3}}\xspace}}} 
}

\NewDocumentCommand{\blindedVPRFKey}{o}
{
{\ensuremath{\hat{\vPRFKey}_{#1}}\xspace}
}

\NewDocumentCommand{\prfEval}{o o}
{
	\ifstrequal{#1}{1}{\ensuremath{X\IfNoValueF{#2}{_{#2}}}\xspace}{\ensuremath{Z\IfNoValueF{#2}{_{#2}}}\xspace}
}
\NewDocumentCommand{\keyCommitment}{o o}
{
	\ifstrequal{#1}{1}{\ensuremath{A\IfNoValueF{#2}{_{#2}}}\xspace}{\ensuremath{B\IfNoValueF{#2}{_{#2}}}\xspace}
}

\newcommand{\numkeys}{\ensuremath{n}\xspace}

\NewDocumentCommand{\coinID}{g}{\ensuremath{R\IfValueT{#1}{_{#1}}}\xspace}

\NewDocumentCommand{\zkpSig}{o o}
{
	\ifstrequal{#1}{1}{{\ensuremath{\Psi_{s\IfValueT{#2}{, #2}}}\xspace}}{{\ensuremath{\Psi_{t\IfValueT{#2}{, #2}}}\xspace}}
}

\NewDocumentCommand{\zkpSigVehicle}{o}
{
	\ensuremath{\Psi_{p\IfValueT{#1}{_{#1}}}\xspace}
}
\NewDocumentCommand{\signedCommVehicleKey}{o}
{
	\ensuremath{\sigma_{\vPRFKey\IfValueT{#1}{_{#1}}}\xspace}
}

\NewDocumentCommand{\zkpPRF}{o}{\ensuremath{\Psi_{\IfValueT{#1}{#1}}}\xspace}

\newcommand{\keyIdx}{\ensuremath{\ell}\xspace}

\newcommand{\simOp}[1]{\ensuremath{\overline{#1}}\xspace}
\newcommand{\hybridArg}[1]{\ensuremath{\mathsf{Hybrid}_{#1}}\xspace}
\newcounter{ctr}
\newcommand{\thisHyb}{\ensuremath{\hybridArg{\thectr}}\xspace}
\newcommand{\prevHyb}{\ensuremath{\hybridArg{\addtocounter{ctr}{-1}\thectr\stepcounter{ctr}}}\xspace}
\newcommand{\nextHyb}{\noindent\ensuremath{\hybridArg{\stepcounter{ctr}\thectr}}\xspace}
\newcommand{\probDiff}[2]{\ensuremath{\left|\prob{\Experiment{#1}{\repCoinScheme}(\Adv{}) = 1} 
		- \prob{\Experiment{#2}{\repCoinScheme}(\Adv{}) = 1}\right|}\xspace}
\newcommand{\grpIdentity}{\ensuremath{\mathbf{1}}\xspace}

\newcommand{\CRLabel}{\ensuremath{CR}\xspace}
\newcommand{\hidingExp}{\ensuremath{\Experiment{\hidingLabel[\hidingBit]}{\commFn}(\hidingAdversary)}\xspace} 
\newcommand{\hidingAdv}{\ensuremath{\Advantage{\hidingLabel}{\commFn}(\timeBound,
 \oracleQueries)}\xspace}
\newcommand{\bindingAdv}{	\Advantage{\bindingLabel}{\commFn}(\timeBound, 
\oracleQueries)}
\newcommand{\unforgeAdv}{\ensuremath{\Advantage{\unforgeLabel}{\repCoinScheme}(\timeBound,
 \oracleQueries)}\xspace}
\NewDocumentCommand{\zkpVerifierExp}{o}{\ensuremath{\Experiment{\zkLabel\IfValueT{#1}{[#1]}}{\zkpName}(\zkpVerifierAdversary)}\xspace}
\newcommand{\zkpHidingAdv}{\ensuremath{\Advantage{\zkLabel}{\zkpName}(\timeBound, \oracleQueries)}\xspace}

\newcommand{\zkpSoundnessAdv}{\ensuremath{\Advantage{\soundnessLabel}{\zkpName}(\timeBound, \oracleQueries)}\xspace}

\newcommand{\ufcmaBlindAdv}{\ensuremath{\Advantage{\ufcmaBlindLabel}{\blindSigScheme}(\timeBound, \oracleQueries})\xspace}
\newcommand{\crAdversary}{\ensuremath{\Adv{CR}}\xspace}

\newcommand{\crAdv}{\ensuremath{\Advantage{\CRLabel}{\hashFn}(\timeBound, \oracleQueries)}\xspace}
\NewDocumentCommand{\PRFExp}{o}{\ensuremath{\Experiment{\prfLabel\IfValueT{#1}{-#1}}{\dyPRF}(\prfAdversary)}\xspace}
\newcommand{\PRFAdv}{\ensuremath{\Advantage{\prfLabel}{\dyPRF}(\timeBound, \oracleQueries)}\xspace}
\newcommand{\untraceAdv}{\ensuremath{\Advantage{\locPrivacyLabel}{\repCoinScheme}(\timeBound, \oracleQueries)}\xspace}

\NewDocumentCommand{\anonExp}{o}{\ensuremath{\Experiment{\anonLabel[#1]}{\repCoinScheme}(\anonAdversary[1], \anonAdversary[2])}\xspace}

\newcommand{\ufcmaAdv}{\ensuremath{\Advantage{\ufcmaLabel}{\sigScheme }(\timeBound, \oracleQueries)}\xspace}
\newcommand{\dlogExp}{\ensuremath{\Experiment{\dlogLabel}{\langle \grp, \grpGen \rangle}(\dlogAversary)}\xspace}
\newcommand{\dlogAdv}{\ensuremath{\Advantage{\dlogLabel}{\langle \grp,\grpGen \rangle}(\timeBound)}\xspace}

\newcommand{\dblSpendAdv}{\ensuremath{\Advantage{\dblSpendLabel}{\repCoinScheme}(\timeBound, \oracleQueries) }\xspace}

\NewDocumentCommand{\timeVar}{o}{\ensuremath{t\IfValueT{#1}{_{#1}}}\xspace}
\NewDocumentCommand{\timeVarAlt}{o}{\ensuremath{t'\IfValueT{#1}{_{#1}}}\xspace}
\NewDocumentCommand{\timeDiff}{o}{\ensuremath{\Delta\IfValueT{#1}{_{#1}}}\xspace}
\NewDocumentCommand{\numCoinDeps}{o}{\ensuremath{N\IfValueT{#1}{_{#1}}}\xspace}
\NewDocumentCommand{\numCoinDepsDelta}{o}{\ensuremath{N\IfValueT{#1}{_{#1}}[\timeVar, 
\timeVar + \timeDiff]}\xspace}
\NewDocumentCommand{\numCoinVar}{o}{\ensuremath{n\IfValueT{#1}{_{#1}}}\xspace}
\NewDocumentCommand{\poissonRate}{o o}{\ensuremath{\lambda_{#1}^{(#2)}}\xspace}

\newcommand{\mixProperty}{time-unlinkable deposits\xspace}

\newcommand{\privKeyCommitment}{\ensuremath{P}\xspace}
\newcommand{\rsuPrivKeyCommitment}{\ensuremath{Q}\xspace}
\newcommand{\userPRFKeyCommitment}{\ensuremath{S}\xspace}
\NewDocumentCommand{\genericFnInput}{g}{\ensuremath{x\IfValueT{#1}{_#1}}\xspace}

\createpseudocodeblock{pcb}{center , boxed, space=0.2cm}{}{}{}

\newcommand{\voucherPRFEval}{\ensuremath{X}\xspace}
\newcommand{\voucherDblSpendToken}{\ensuremath{Y}\xspace}
\newcommand{\txPRFEval}{\ensuremath{Z}\xspace}
\newcommand{\userCLSig}[1]{\ensuremath{\sigma_{#1}}\xspace}
\newcommand{\userReceipt}{\ensuremath{\mathsf{receipt}}\xspace}

\newcommand{\carLabel}{\ensuremath{\mathcal{U}}\xspace}
\newcommand{\rsu}{ATM\xspace}
\newcommand{\rsus}{ATMs\xspace}

\usepackage{enumitem}
\setlist[itemize]{nosep, itemsep=0pt, topsep=3pt, leftmargin=1em, labelwidth=*, align=left}
\setlist[enumerate]{nosep, itemsep=0pt, topsep=3pt, leftmargin=1em, labelwidth=*, align=left}

\usepackage{tikz}

\usepackage{cellspace}
\usepackage{makecell} 
\setcellgapes{1pt}
\usepackage{booktabs}

\newcommand{\mysection}[1]{\section{#1}}
\newcommand{\mysubsection}[1]{\subsection{#1}}

\newcommand\notsotiny{\@setfontsize\notsotiny{6}{7}}

\begin{document}

\date{}

\title{Automatic Teller Machines for Offline \Ecash}

\author{Anrin Chakraborti\inst{1} \and
Qingzhao Zhang\inst{2} \and
Jingjia Peng\inst{3} \and
Morley Mao\inst{3} \and
Michael K. Reiter\inst{4}}
\authorrunning{Chakraborti et al.}
%
\institute{University of Illinois Chicago, Chicago IL 60607, USA \and
University of Arizona, Tucson AZ 85721, USA \and
University of Michigan, Ann Arbor MI 48109, USA \and 
Duke University, Durham NC 27708, USA
}

   \maketitle

\begin{abstract}
Electronic cash (\ecash) is a digital alternative to physical currency that allows anonymous transactions between users and merchants. Typically, coins in an \ecash scheme are only dispensed through a central bank. A drawback of this approach is that the bank is always on the critical path during withdrawals, and if a reliable connection to the bank is temporarily unavailable, users may be unable to withdraw coins in a timely fashion. As with physical currency, there are benefits to supporting a decentralized infrastructure where withdrawals can be performed without involving the bank in the critical path. 

We propose the design of a new cryptographic bearer token that can be  dispensed by  \textit{automatic teller machines} (ATM) in a fully offline \ecash scheme. Such bearer tokens provide 
 anonymity, unforgeability and untraceability,  i.e., users cannot be tracked 
 by their spending activities or the locations of withdrawal. We formalize the requirements of an \ecash scheme with multiple issuers and propose an efficient design building on top of the compact \ecash protocol of \citet{camenisch2005compact}. Our construction leverages an unforgeable and \emph{doubly-anonymous} voucher that allows a one-time transfer of coins between an ATM and a user, while hiding their identities from parties not involved in the transaction. 

\end{abstract}

\mysection{Introduction}
\label{sec:intro}

Electronic cash (\ecash)~\cite{chaum1990untraceable} entitles a \textit{user} holding a \textit{bearer token} (coin) to certain privileges, 
which she exercises by presenting the token to
a \textit{gatekeeper} (merchant).  As such, to be useful, a party should be unable to forge bearer tokens for herself; instead, a user must get a
bearer token from a sanctioned \textit{issuer} (bank). Furthermore, coins should also ensure \textit{anonymity} during transactions, that is, 
they should not divulge any information regarding the user spending a coin, either to the bank or to the party that receives it in a transaction. Typically, there is only one issuer, namely the
bank, and most \ecash schemes cannot be trivially extended (unless the bank shares secrets) to support multiple issuers.   


Nonetheless, even in the context of \ecash, there are benefits to supporting multiple issuers and using them in a role similar to \emph{automatic teller machines} (ATMs) for physical currency. For example, to withdraw coins in a single-issuer system, a user must always connect to the bank, which makes the process vulnerable to network disruptions. As a viable alternative, we can envision a setup where ATMs act as physical proxies of the bank and facilitate fully \emph{offline} operations. The user-to-ATM connections are supported through fast short-range communication channels, e.g., bluetooth, while the ATM-to–bank communication is over a standard wireless/wired network.

Introducing multiple issuers, which we will refer to as ATMs henceforth, pose new technical challenges for \ecash designs. First, if ATMs are physical entities, then a coin bearing the ATM's identity, e.g., in the form of a signature, will violate the user's anonymity guarantees. Therefore, coins should be \textit{untraceable} to their issuers, as in case of real physical currency. Second, some of the ATMs can behave maliciously and the necessary safeguards need to be incorporated in the design; indeed instances of ATM hijacking are known in reality~\cite{atm_malware}. Third, the bank must be able to enforce rate-limiting policies on the ATMs and identify malicious behavior. Finally, it is desirable to support offline withdrawals if connectivity to the bank is temporarily disrupted.

Thus, this papers asks the following question: 

\begin{myquote}
Can we design \ecash schemes with multiple (potentially malicious) issuers under the control of an honest-but-curious bank
that guarantees unforgeability, untraceability and anonymity?
\end{myquote}

To answer this question affirmatively, we formalize the requirements of an \ecash scheme with multiple issuers and propose an efficient design. In our model, ATMs are stocked with coins \textit{independent of user requests}, and subsequently a coin(s) is transferred between an authorized ATM and a verified user during a withdrawal from the ATM. To support such (limited) coin transfers, we introduce the notion of an unforgeable \textit{voucher} that links a coin to the identities of the issuing ATM and the user who receives the coin. The voucher is included as part of the coin transaction and allows the bank to identify misbehaving users and ATMs. Crucially, the vouchers are \emph{doubly anonymous} in the sense that they do not reveal the user's and issuing ATM's identities to the bank unless the coin has been double-spent (or double-issued). This general framework can be instantiated with several existing \ecash schemes by integrating vouchers into their withdrawal and spending mechanisms. We provide an instantiation building on top of the compact \ecash protocol of \citet{camenisch2005compact}. 

\begin{figure}[t]
	\centering \includegraphics[width=\textwidth]{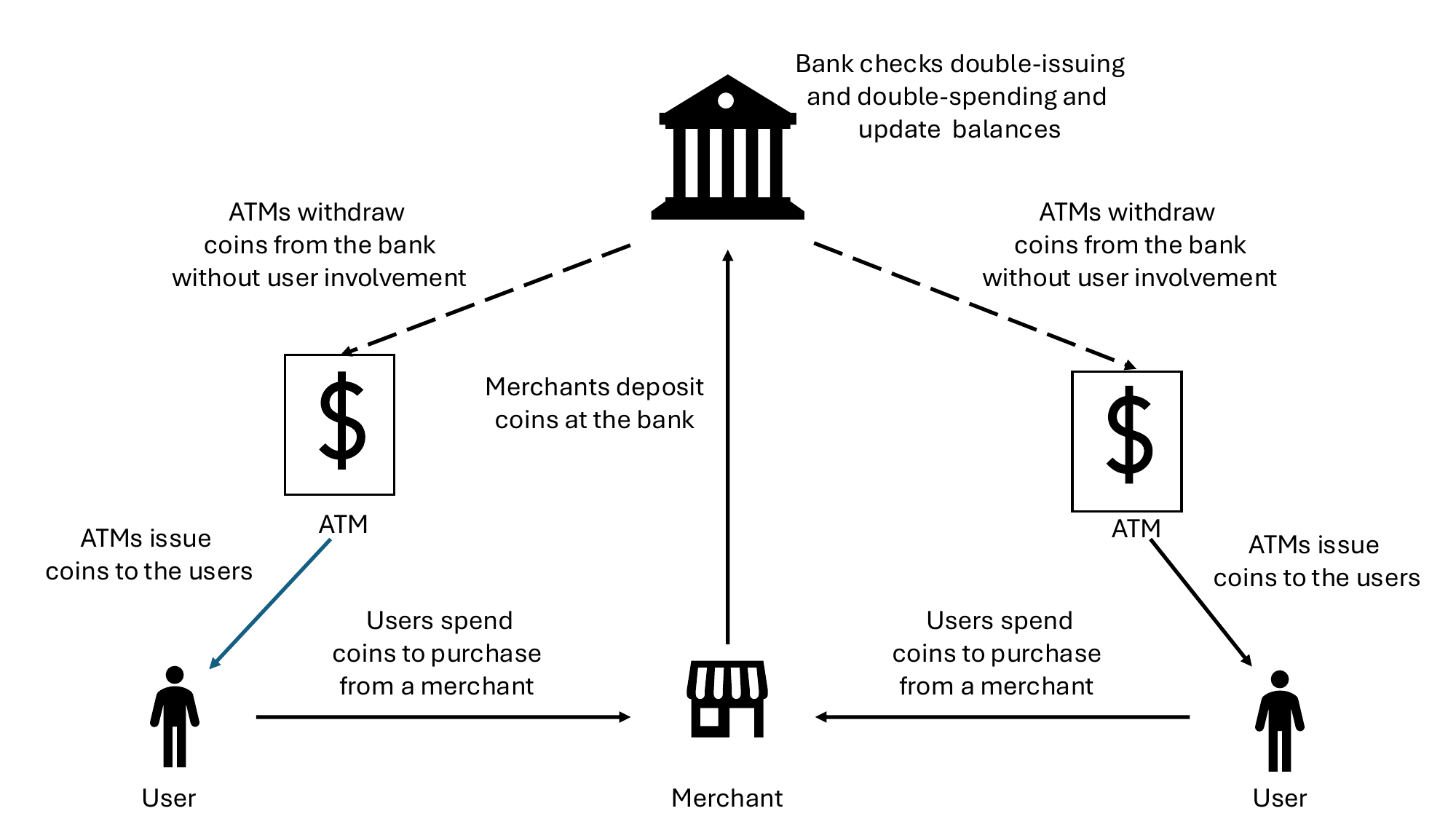} \caption{A description of the parties and operations in a multi-issuer \ecash scheme. The ATMs requests coins from the bank (the dashed line indicates that this process is performed in the background and periodically) which are then
	  issued to users. Coins are spent at
	  merchants using transactions. Once the merchant verifies a
	  transaction locally, it can deposit it at the bank. }  \label{fig:system_overview} \end{figure}

\subsection{Problem Definition}
\label{sec:problem_defn}

\myparagraph{Functionality}
An \ecash scheme with multiple issuers (ATMs) must support three types of operations: i) ATMs can \emph{withdraw} coins from the bank, ii) ATMs can \emph{issue} previously withdrawn coins to users on demand, and iii) users can \emph{spend} the coin acquired previously at a merchant (\figref{fig:system_overview}). Merchants deposit coins obtained in transactions to the bank for account credits. In our model, only the bank is authorized to mint coins; the ATMs are intermediaries who facilitate fast withdrawals. However, ATMs cannot accept coin deposits because they do not maintain account balances and/or have the ability to check for double-spending. Thus, deposits must always be to the bank. 

For security, informally speaking, we assume that some of the ATMs can behave maliciously and potentially \emph{double-issue} coins to honest users, and/or attempt to debit an honest user's account without issuing a commensurate number of coins. 
Similarly, malicious users can attempt to \emph{double-spend} coins. Thus, all operations between users and ATMs need verifiability; e.g., the merchant can verify that the coin spent in a transaction was minted by the bank and issued to the user by an authorized ATM. 
The bank is honest-but-curious in our model and detects both double-spending by users and double-issuing by ATMs.

\myparagraph{Existing definitions: transferable \ecash}
Before proceeding further to formally define a multi-issuer \ecash system, we remark that all of the functions we require are already provided by transferable \ecash schemes. Introduced first by \citet{okamoto1989disposable}, transferable \ecash allows an arbitrary number of anonymous coin transfers between users and merchants. Indeed any 
transferable \ecash scheme will support ATMs almost trivially, e.g., ATMs can blindly withdraw coins from the bank, and then transfer them to the users.

However, transferable \ecash is a strictly stronger primitive than multi-issuer \ecash. Specifically, as noted by \citet{canard2008anonymity}, any scheme supporting arbitrary coin transfers between users must satisfy a stronger notion of anonymity called \emph{coin transparency}: once a user transfers a coin to another user, she can no longer recognize this coin even if she receives it in a subsequent transfer from another user. This is necessary because without this property, a merchant can trace users by ```marking'' a coin obtained in a transaction and then transferring the coin to another user. Subsequently, if the same coin is used again in a transaction, it will allow the merchant to easily correlate the two transactions, violating anonymous-spending guarantees. 

Unfortunately, coin transparency comes with steep performance penalties: each transfer essentially randomizes the coin and the entire transcript pertaining to any previous transfer(s), thus ensuring that even if a coin reaches a party that has previously seen the same coin, he cannot recognize it. The transcript size is typically linear in the number of transfers, and it proves authenticity of the transfers while also allowing double-spending detection by the bank.  To support this randomization, transferable schemes either use randomizable non-interactive zero-knowledge proofs~\cite{bauer2021transferable} to generate the transcript, and/or malleable signatures~\cite{baldimtsi2015anonymous}. Both of these tools are computationally expensive, making the state-of-the-art transferable \ecash scheme~\cite{bauer2021transferable} concretely less efficient than any standard \ecash scheme, which can be typically constructed from cheaper primitives. Another major drawback of supporting arbitrary transfers is an inherently large coin size. For example, in the scheme of  \cite{bauer2021transferable}, a single coin comprises of more than 272 pairing group elements, and each transfer adds 104 new elements to the coin (see Table 1 of \cite{bauer2021transferable}). When the pairing group is instantiated with the BLS12-381 curve~\cite{barreto2005pairing}, a coin is roughly 19\kilobytes in size. 

In contrast to transferable \ecash, \emph{coin transparency is not necessary} for multi-issuer e-cash since, unlike in transferable schemes, after a user spends a coin, the merchant can only deposit the coin to the bank. This alleviates fears that the merchant will ``mark'' the coin and correlate any subsequent transactions. As such, we do need to randomize the coin (and any related transcript), and more efficient constructions matching the performance of standard \ecash designs are possible. Concretely speaking, the size of a coin in our scheme is 1\kilobytes, more than 18$\times$ smaller than the coin size in \citet{bauer2021transferable}, and \emph{all} operations are highly efficient, taking less than 100ms even on resource-constrained setups. While the performance of \citet{bauer2021transferable} is not reported, we estimate higher computational costs, since the scheme relies on the Groth-Sahai proof framework~\cite{groth2008efficient} for the randomizable non-interactive zero-knowledge proof system. We have no such requirement and only rely on standard Sigma proofs.

\myparagraph{New requirements} In addition to the model relaxation, a multi-issuer scheme also enforces new requirements that are not directly addressed by a transferable \ecash scheme. We details these requirements below.

\begin{enumerate}
\item \textbf{Fair exchange:} Since some of the ATMs and users are malicious, we need to ensure that a \emph{fair exchange} of coins takes place before an account is debited at the bank. That is, we must ensure that an ATM cannot falsely effectuate an update to a user's balance at the bank, unless it has issued a commensurate number of coins to the user, and conversely, a user cannot falsely accuse an ATM of reneging on a promise to issue coins, despite receiving the coin from the ATM. It is known that fair exchange with malicious parties is impossible without a trusted third party~\cite{pagnia1999impossibility,sandholm2002possibility}. Thus, a multi-issuer scheme must incorporate new mechanisms for such fair exchange by involving the bank as a TTP in case of disputes, without violating anonymity.
 
\item \textbf{Compactness:} For storage efficiency, we want to ensure that coin withdrawals by the ATM can be performed offline. This will make the issuing process more efficient, as the bank needs to be consulted only for checking the user's balance. However, to support this feature, the ATMs would need to store a sufficient stock of coins to serve user requests as they arrive. Thus, optimizing coin storage becomes critical if the ATMs run on resource-constrained devices. One way to optimize storage is by making the scheme \emph{compact}~\cite{camenisch2005compact}, i.e., the cost of storing $\numkeys$ coins scales sub-linearly in $\numkeys$. This feature is missing from the state-of-the-art transferable \ecash scheme.  

\end{enumerate}

\myparagraph{New definitions} 
In light of the discussion above, we define the security of a multi-issuer e-cash scheme by extending the formal definitions of unforgeability and anonymity of coins -- used in standard \ecash schemes -- to tolerate the presence of multiple issuers using a game-based model~\cite{bellare2004code}. Since some of the ATMs and users can be malicious, our security definitions account for both double issuing and double spending. If the bank detects either double spending or double issuing, then there is a mechanism to identify the offending party and enforce penalties. In addition, we define two new properties:

\begin{itemize}
\item \textbf{Untraceability:} When a user obtains a coin from one of the ATMs, the corresponding transcript of issue does not reveal the identity of the issuer to either the merchant where the coin is spent, or to the bank. A single-issuer \ecash scheme does not need such untraceability guarantees because it is implicit that \emph{all} coins must have been directly withdrawn from the bank. 
\item \textbf{Fair exchange:} Informally, this property requires that if either a user or an ATM reneges on its promise during a coin issue --- as we will discuss later, a user provides a receipt upon receiving a coin from an ATM --- the coin cannot be later spent without violating anonymity of the guilty party. 
\end{itemize}


\subsection{Construction Overview}
\label{sec:construction_overview}
Towards a solution for a multi-issuer \ecash scheme, we can start with a standard \ecash scheme such that the ATMs withdraw coins from the bank and store them locally. When a user requests a coin from an ATM, the ATM invokes the spending procedure of the \ecash scheme, with the user acting as the merchant, to issue a coin(s). Then, the user may spend the coin at a merchant and provide the transcript pertaining to the coin issue, i.e., the proof of ATM-to-user transaction, to demonstrate that it has acquired the coin from a valid ATM. The coins will be unforgeable and preserve anonymity during spending. It is worth noting that ensuring the validity of the coin is not enough because the merchant must also verify that the user presenting the coin is also the owner of the coin. Otherwise, a malicious ATM may be able to \emph{double-issue} coins.

The problem with this construction is that if issuing a coin simply follows the procedure of a standard transaction, and the users provides the corresponding transcript to the merchant/bank for verification, then the user does not remain anonymous. Indeed, in \ecash, the identity of a merchant (the user in case of a ATM-to-user transaction) reporting a transaction is part of the transcript \emph{in the clear}, so as to enable the bank to credit the merchant's account. 

\myparagraph{Doubly anonymous vouchers}
We address this problem by introducing an unforgeable and doubly anonymous \emph{voucher} that records a transaction of a valid coin between a spending party (an ATM) and receiving party (a user) while hiding both of their identities from any party that is not involved in the transaction, including the bank. More specifically, the voucher + coin allows any party to verify that the coin was minted by the bank, and was issued to the user claiming ownership, without revealing identities of either the user or the ATM that issued the coin.   In an \ecash scheme, such a voucher serves no purpose since at least the identity of the receiving party must be known to credit an account. However, a doubly anonymous voucher will allow for a \emph{single} transfer of a coin between an ATM and a user, while also enabling the user to later spend the coin at a merchant. We remark that such a voucher can be computed \emph{only} by the party who withdrew the coin --- in our construction, only an ATM can compute a voucher --- and so vouchers cannot be used for arbitrary number of transfers, as in a transferable \ecash scheme. Therefore, a doubly-anonymous voucher is a strictly weaker primitive than a randomizable \emph{tag} proposed in previous work on transferable \ecash~\cite{baldimtsi2015anonymous}. 

\myparagraph{Instantiation with compact \ecash}
Introducing a doubly anonymous and unforgeable voucher provides a blueprint to turn any single-issuer \ecash scheme into an \ecash scheme with multiple issuers. We demonstrate feasibility of this idea using the compact \ecash scheme of \citet{camenisch2005compact}.  Informally, their protocol can be described \emph{without compactness} as follows. 

\begin{itemize}

\item \textbf{Withdrawal:} Each coin is a key $\genericFnKey$ for a pseudorandom function $\dyPRF$ uniformly-sampled by the user, and blindly signed by the bank using a CL signature~\cite{camenisch2003signature}.

\item \textbf{Spending:}
\begin{enumerate}
\item The user spends the coin by computing a \emph{verification token}  $\voucherPRFEval \assign \dyPRF[\genericFnKey](\coinID)$, where $\coinID$ is a unique coin-specific identifier (defined later). 

\item  Double spending is detected using a \emph{double-spending token} of the form $\txPRFEval \assign \pubKey{\vID} \cdot \dyPRF[\genericFnKey](\coinID)^{\txRand}$ where $\txRand \gets \hashFn(\pubKey{\genDepRSUID}, \dots)$ is derived from the merchant's public key $\pubKey{\genDepRSUID}$ and other transaction-specific information. Here, $\pubKey{\vID}$ is the spending user's public key.

\item The user proves in zero-knowledge that it knows a signature on $\genericFnKey$ under the bank's private key. The user also proves that the verification token and the double-spending token have been correctly computed. Thus, each transaction reported to the bank includes $\voucherPRFEval$, $\txPRFEval$, $\txRand$ and the zero-knowledge proofs. The coin is not revealed in the clear as part of the transcript.
\end{enumerate}

\item \textbf{Deposit \& double-spending:}
If there are two transactions spending the same coin, which means that both transactions have matching values of the verification token \voucherPRFEval but different double-spending tokens $\txPRFEval$ and $\txPRFEval'$, then the bank obtains i) $\txPRFEval \gets \pubKey{\vID} \cdot \dyPRF[\genericFnKey](\coinID)^{\txRand}$, ii) $\txPRFEval' \gets \pubKey{\vID} \cdot \dyPRF[\genericFnKey](\coinID)^{\txRandAlt}$, iii) $\txRand$ and iv) $\txRandAlt$. Using these values the bank derives the misbehaving user's identity 

\begin{align*}
& \dyPRF[\genericFnKey](\coinID) = \left(\frac{\txPRFEval}{\txPRFEval'}\right)^{\txRandAlt - \txRand}; ~\pubKey{\vID} =  \frac{\txPRFEval}{\dyPRF[\genericFnKey](\coinID)^{\txRand}}
\end{align*}

\end{itemize}

\myparagraph{Design overview}
We adapt this design to a multi-issuer system. At a high-level, our scheme can be viewed as a  two-layered \ecash scheme where each spent coin is part of two transactions: the first transaction is when the coin is issued to a user by an ATM, and the second transaction is when it is spent by the user at a merchant. The caveat is that both these transaction must be cryptographically linked to each other; i.e., the receiver in the first transaction is the spender in the second transaction. Since we cannot reveal the receiver's (and spender's) identity, we use the unforgeable and doubly anonymous voucher to tie the two transactions together. In the following, we informally the describe this idea.

We begin by describing the keys in the possession of the users and the ATMs. Each user has: i) a \emph{identity key} pair $\langle \pubKey{\vID}, \privKey{\vID} \rangle$, ii) a \emph{spending} key that is a fixed, uniformly sampled PRF key $\userPRFKey$, and iii) a key pair for a digital signature scheme. The private identity key and the PRF key are blindly signed by the bank using a CL-signature. 
Similarly, each ATM has: i) a identity key pair $\langle \rsuPubKey{\rsuID}, \rsuPrivKey{\rsuID} \rangle$, and ii) a key pair for a digital signature scheme. The private identity key is signed blindly by the bank using a CL-signature. ATMs withdraw coins from the bank following  Step 1 of the compact \ecash scheme described above; i.e., each coin is a PRF key $\genericFnKey$ that is uniformly sampled by the ATM and blindly signed by the bank. When an ATM issues the coin corresponding to \genericFnKey to a user, the ATM additionally creates a voucher and provides it to the user as proof of ownership.

\begin{enumerate}
\item The voucher includes a unique identifier $\voucherPRFEval \assign \dyPRF[\genericFnKey](\privKeyCommitment)$ where $\privKeyCommitment$ is a commitment to the user's private identity key $\privKey{\vID}$. The voucher also includes \privKeyCommitment. 

\item The voucher includes a \emph{double-issuing token} that embeds the identity of the ATM $\voucherDblSpendToken \assign \rsuPubKey{\rsuID} \cdot \dyPRF[\genericFnKey](\dots)^{\privKeyCommitment}$. Here, $\rsuPubKey{\rsuID}$ is the issuing ATM's public key. (The inputs to the PRF computation is a constant and will be defined later.) 

\item The voucher includes non-interactive zero-knowledge proofs to show that \voucherPRFEval and \voucherDblSpendToken have been correctly computed under a PRF key signed by the bank. 
\end{enumerate} 

The voucher is doubly anonymous because it does not include the issuing ATM's or user's identity in the clear. Unforgeability is ensured because the issuing party must prove knowledge of the bank's signature on the key $\genericFnKey$. Similarly, the voucher is bound to the user because of the computationally binding commitment on the user's private key that ties all the components of the voucher together.
Moreover, if an ATM issues the same coin twice, then the bank will be able to recover $\rsuPubKey{\rsuID}$ from the double-issuing token, assuming that with overwhelming probability, the commitments in two separate withdrawals will not be the same.  As we will see, in order to spend the coin, the user must prove that it knows the opening of the commitment.

\subsection{Design Optimizations}
\label{sec:design_opt}

\myparagraph{Fair exchange}
We ensure an \emph{optimistic} fair exchange of coins between users and ATMs by 
using the bank as a trusted third party~\cite{asokan1997optimistic,asokan1998optimistic,micali1999verifiable,burk1990value}. Optimistic as defined by previous work~\cite{burk1990value,asokan1997optimistic} in this context means that the bank is contacted only when there is a dispute during an exchange

The idea is that after a user is issued a coin, it provides a signed receipt to the ATM that indicates that it has withdrawn a coin(s) from that ATM; importantly, the receipt does not include any information that identifies the coins. The receipt is subsequently provided by the ATM to the bank to debit the user's account. However, during this exchange, it is possible that one of the parties reneges on its promise: the user may not provide a receipt after receiving a coin, or conversely, the ATM may not provide the coin to the user after receiving the receipt. Either way, this is problematic: if the user does not provide a receipt, then it essentially gets a coin for ``free'' because the ATM cannot prove to the bank that the user has withdrawn the coin. And, if the ATM misbehaves, then it may issue the coin to another (malicious) user, but debit the honest user's account instead.

To mitigate this problem, we facilitate an optimistic fair exchange as follows. 

\begin{enumerate}
\item The ATM first provides a signed (under its private key) \emph{promise} to the user, which includes a i) computationally hiding and binding commitment to the coin, ii) the voucher linked to the coin in the clear, and iii) a randomly chosen nonce that uniquely identifies the withdrawal. 
\item After verifying the signature, the user signs a receipt (under its private key) which includes i) the public key of the user, ii) the public key of the ATM, and iii) the nonce, and provides it to the ATM. 
\item After verifying the receipt, the ATM opens the commitment to the coin. The user verifies the opening and ensures that the coin is consistent with all the information signed by the ATM in Step 1. 
\item If the user does not receive the coin from the ATM after sending the receipt, or the coin is inconsistent with all the other information, the user informs the bank about the ATM's misbehavior using an \abortWithdrawal message. In that message, the user includes the coin commitment, the voucher, the nonce and the ATM's signature on all of the information obtained in Step 1.
\item If the bank receives an \abortWithdrawal message from a user, it stores all the information in a log for future investigation. 
Moreover, the bank does \emph{not} debit the user's account, \emph{even if a contradictory receipt is received from the ATM signifying that the user's balance should be debited}. 
\end{enumerate}

This mechanism ensures fair exchange as follows. Suppose first that a misbehaving user does not provide the receipt after receiving the first message from the ATM. In that case, the user does not get the coin, unless it is able to derive the coin from the commitment, which happens with negligible probability. Of course, a user may try to misuse the mechanism by sending an \abortWithdrawal message even after it receives the coin. But this is not a problem, because since the message contains the signed commitment to the coin, the bank can later identify if the coin is spent by the same user, and catch the user's lie.

Conversely, suppose that an ATM misbehaves by not opening the commitment after receiving the receipt, and the user sends the \abortWithdrawal message along with the voucher and the coin commitment. The idea is that since the voucher ties the coin to the user, and includes the double-issuing token, if the ATM tries to double-issue the coin to a different user later, then upon receiving the coin after it is spent, the bank will be able to identify the ATM for double-issuing. Thus, a (committed) coin that is reported as part of \abortWithdrawal effectively becomes void because it cannot be spent by a misbehaving user or issued by a misbehaving ATM. We remark that since both the user's and ATM's identities are available to the bank, this mechanism cannot be misused for a denial of service since : i) the bank will identify if too many reports are submitted by the same user, and ii) the ATM will exhaust all its coins as a reported coin becomes unusable, and the bank can deny further withdrawals by the ATM. 

\myparagraph{Spending \& deposits}
A user spends a coin by providing both the coin and the voucher to the merchant. Recall that the commitment to the user's private key $\privKeyCommitment$ binds the coin to the voucher, and the merchant verifies that the coin was indeed issued to the user using the commitment. Specifically, to spend a coin, 

\begin{enumerate}
\item The user with public key $\pubKey{\vID}$ computes the double-spending token $\txPRFEval \assign \pubKey{\vID} \cdot \dyPRF[\userPRFKey](\privKeyCommitment)^{\txRand}$ using its signed spending key $\userPRFKey$ (where \txRand is derived from the merchant's public key $\pubKey{\genDepRSUID}$ and other information specific to the transaction).
\item The user proves in zero-knowledge the bank's signature on the opening of $\privKeyCommitment$, and that the public key $\pubKey{\vID}$ in the double-spending token  is linked to the private key in the opening of the commitment. The user also proves in zero-knowledge the bank's signature on $\userPRFKey$. The transaction with the merchant includes the coin, the voucher, \txPRFEval, and the zero-knowledge proofs.   
\end{enumerate}

As we will describe later, the PRF and the commitment scheme can be instantiated such that all the zero-knowledge proofs can be implemented with standard Sigma protocols. Once the merchant verifies the coin and the voucher, it deposits all the information to the bank. The bank checks for double spending and double issuing, following a procedure similar to the compact \ecash protocol. The bank also check if there is a matching commitment to the coin in any of the \abortWithdrawal messages it has received in the past. If so, the bank identifies the misbehaving user/ATM.

\myparagraph{Compactness}
The scheme can be made compact by leveraging a technique due to \citet{camenisch2005compact} whereby the same PRF key $\genericFnKey$ is used multiple times to issue new coins. Each coin has a unique serial number $\coinSN \in \nats[1][\numkeys]$, and both the verification token $\voucherPRFEval \assign \dyPRF[\genericFnKey](\coinSN)$ and double-issuing token $\voucherDblSpendToken \assign \rsuPubKey{\rsuID} \cdot \dyPRF[\genericFnKey](\coinSN)^{\privKeyCommitment}$ are tied together using $\coinSN$. The issuing ATM computes a non-interactive zero-knowledge range proof showing that the same value of $\coinSN \in \nats[1][\numkeys]$ was used for both evaluations. We describe the scheme in detail in \secref{}. 


\myparagraph{Fully offline withdrawals:}
Our scheme supports fully offline withdrawals such that an ATM does not need to check a user's balance before issuing coins. This is particularly helpful if the connection between the ATMs and the bank is prone to disruptions, or the round-trip latency is high. To support this feature, we share with the ATMs a set of accounts that do not have enough balance to support the maximum amount that can be withdrawn from an ATM; typically this amount is much smaller than an average user's balance. The set is shared in a concise fashion by periodically broadcasting to the ATMs a Bloom filter with these account numbers.  When a user requests a withdrawal, the ATM checks from the Bloom filter whether the user has sufficient balance. We note that while the Bloom filter reveals the accounts with low balance to the ATMs, in practice this information (or parts of it) may be  available to the ATMs anyway during withdrawals. Sharing the Bloom filter does not violate any of the security guarantees of an \ecash scheme (anonymity during spending and unforgeability). 
The ATM updates the bank about the withdrawal using the users' receipts as proof, either periodically or under stable network connection.  

A potential risk of fully offline withdrawals is that the user may attempt to overdraw her account by visiting multiple ATMs in between the times that the user withdraws from an ATM and her account is debited by the bank. However, in this case, the user will be eventually identified by the signed receipts, and the bank can enforce appropriate penalties. Nonetheless, we also discuss mitigation strategies to reduce the risk of such behavior in \secref{sec:offline_withdraw}.  

\myparagraph{Collusion \& enforcing penalties}
As with most \emph{offline} \ecash schemes, our scheme also detects user/ATM misbehavior \emph{after the fact}, and thus we expect that the bank will enforce appropriate penalties. We do not stipulate how such penalties are decided, but we note that in real-world deployments, economic penalties are typical, e.g., overdraft charges when a user overdraws her account. Another point of concern is that if an ATM misbehaves, it may collude with a malicious user and issue a large number of coins. Since every coin that is issued by an ATM must be accounted for in a corresponding receipt, the bank will detect if there is a discrepancy in the amounts. One obvious precaution that the bank can implement is ensuring that ATMs can only stock a small amount of coins at a time. However, beyond these measures, purely cryptographic solutions are not enough to handle collusions, because if the ATM ceases all communication with the bank after issuing an arbitrary number of coins to a malicious user and deletes all local state, then the bank cannot learn any new information to identify the malicious users. Therefore, non-cryptographic safeguards are necessary. For example, it is common practice to have cameras inside ATMs to track user misbehavior; we envision similar measures if the ATMs are implemented by physical entities, e.g., edge devices, in our scheme. Of course, physical attacks against such mechanisms are possible but we consider them out of scope.

\subsection{Related Work}
\label{sec:background}

Electronic cash (\ecash) is a digital alternative to physical currency. In typical schemes, users {\it withdraw} digitally signed coins from a bank and {\it spend} them at merchants, who verify the signature and report valid, unspent coins. Anonymity is preserved by preventing the bank from linking withdrawals to spending. E-cash is \textit{online} when the bank participates in transactions~\cite{chaum1983blind}, and \textit{offline} when it only detects (but cannot prevent) double spending~\cite{chaum1990untraceable}. Numerous variants address specific needs~\cite{camenisch2007endorsed,canard2015divisible,camenisch2005compact,shi2007conditional,baldimtsi2015anonymous,brands1994untraceable,sander1999auditable}. We refer readers to broader surveys for more details, and focus on schemes with multiple issuers (ATMs).


One class of scheme is delegated \ecash, which uses proxy blind signatures~\cite{liu2013proxy,liu2015delegated,tan2011off,zhang2003efficient} to let a central bank authorize branches to issue coins. However, these schemes assume full trust in branches and lack oversight mechanisms such as \textit{rate-limiting} to prevent over-issuance. Another option is transferable \ecash~\cite{baldimtsi2015anonymous,canard2008anonymity,bauer2021transferable}, where coins can be passed between users arbitrarily. While meeting our requirements, such schemes are inefficient in practice and support more general transferability than needed as discussed before.

\mysection{Model}
\label{sec:model}


\mysubsection{Threat Model}
\label{sec:threat_model}
%

Our \ecash system has three types of participants.

\myparagraph{Bank} 
The bank issues currency and maintains user accounts. The bank also certifies users  
that are registered to the system, e.g., by signing their public keys. The bank is {\it honest but curious}.


	
\myparagraph{ATM} 
Similar to the real-world analogue of physical currency,  \rsus serve as cash dispensing 
points. Users can withdraw coins at the \rsus, \emph{but deposits must be directly at the bank}. This 
is because an account credit can occur only after a coin(s) has been checked for double-spending, which cannot be performed by an \rsu. Some \rsus are malicious, however, we assume that adversarial behaviors such as denial-of-service (DoS) attacks or reneging on a promise to dispense cash are considered beyond the scope of this study. 



\myparagraph{Users}
Users withdraw and spend coins either by interacting directly with the bank or via \rsus. A subset of 
the users are also merchants who provide services in exchange for coins. Users can behave maliciously, 
i.e., they may violate the protocol to overdraw their balance or double-spend coins.

\mysubsection{Syntax}
We define the syntax of our \ecash scheme (see \figref{fig:overview}), starting from the items that are exchanged 
in the scheme.

\begin{itemize}

\item {\bf Coin:} A coin is the currency exchanged in the scheme. A coin can be only issued by a bank, directly to a user or to an \rsu. Each such coin is unique (with overwhelming probability). To spend a coin, it needs to be first included in a voucher (described next), which ties the coin to a user identity (necessary for double spending detection). When a user directly withdraws a coin from a bank, the voucher is created by the bank. Otherwise, it is created by the \rsu from where the user withdraws the coin.

\item {\bf Voucher}: A voucher ties a specific coin to a user's identity, e.g., public key, when the user withdraws that coin. The rationale behind having a voucher (in addition to coins) is that unlike traditional online  \ecash, where the user must directly interact with the bank to withdraw a coin, 
in our scheme the user can be issued coins that have been acquired \textit{a priori} by an \rsu without the user's identity being known. Consequently, mechanisms that require the user's identity, e.g., double-spending detection, cannot be directly applied to the coin. A voucher enables these additional mechanisms. Vouchers are not transferable between users. We also assume for simplicity that a voucher is tied to a single coin, although this can be relaxed with additional mechanisms.


\item {\bf Transaction:} A transaction records the spending 
of a coin. It uniquely identifies the coin, and the
merchant that received the coin for a service. Each transaction has a
unique, random transaction ID,and may also include other
fields that are application-specific. Additionally, the identity of the user
spending the coin is embedded in a transaction, {\it but is only
available in the clear to the bank when the coin in the transaction is
double spent in another transaction}. Once the merchant verifies the
coin and the voucher, it deposits the transaction to the bank, which updates balances
and records the coin in a log.

\item {\bf Receipt:} A receipt records an interaction between a user and an \rsu. In particular, this is a message signed by user's private key indicating the amount that has been withdrawn and other relevant information. Receipts act as a safeguard against misbehavior by malicious \rsus.

\end{itemize}

\begin{figure}[t]
	\centering \includegraphics[width=\textwidth]{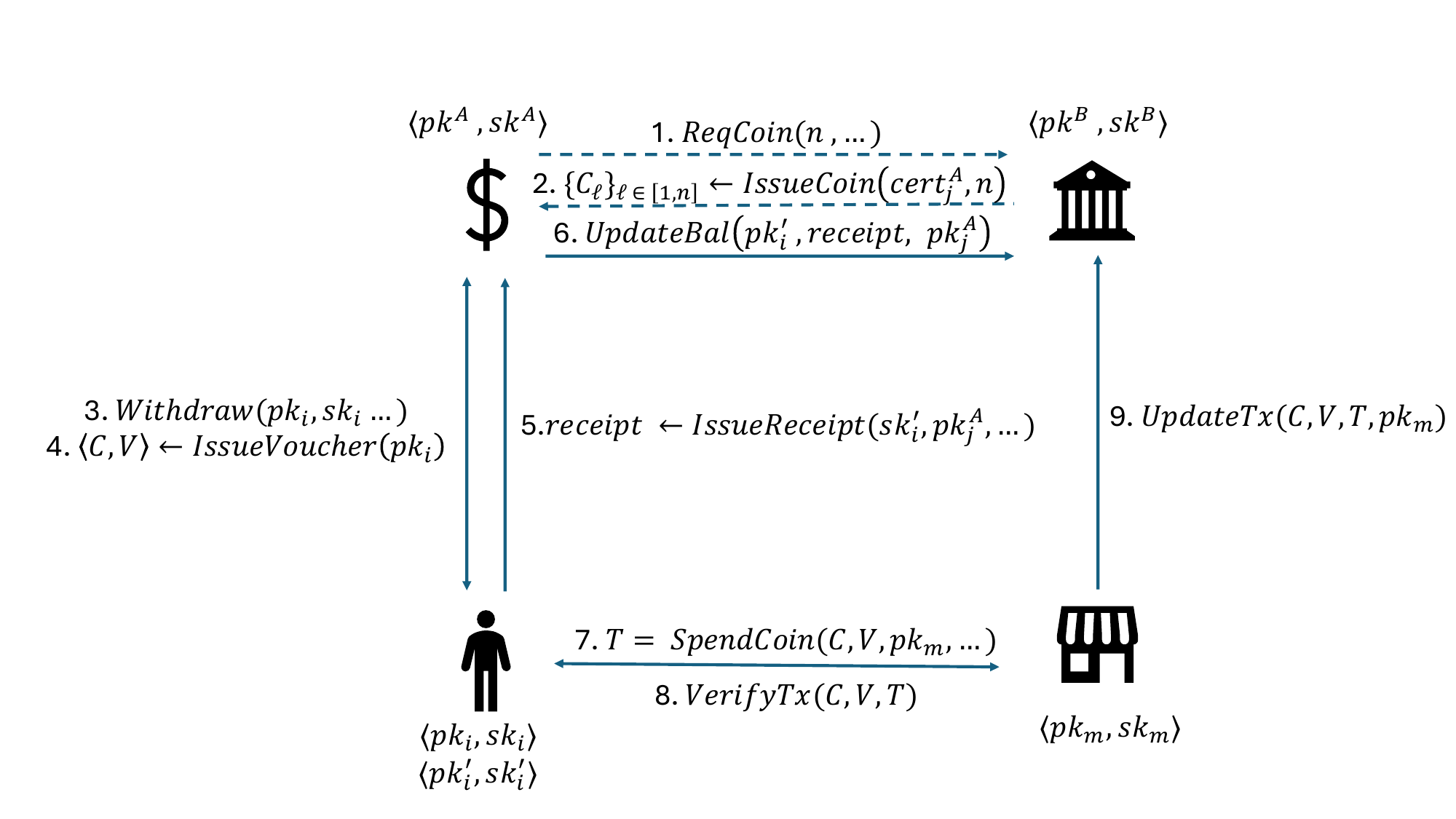} \caption{Formal interfaces in a multi-issuer \ecash scheme. Each withdrawal from an ATM produces a coin and a linked voucher. After withdrawal, the user provides a signed receipt to the ATM as proof of withdrawal that is used to debit the user's account. }  \label{fig:overview} \end{figure}

\subsection{Algorithms}
\medskip
Our scheme has the following algorithms (\figref{fig:overview}).

\subsubsection{Bank:} The bank implements the following algorithms. 

\begin{itemize}
  
	\item $\bankKeyPair \assign  \serverInit(1^{\secParam})$: When this interface is invoked, the bank
  generates the secrets necessary to issue coins, and authenticate
  users/\rsus. This includes a key pair for a blind signature
  scheme. It also initializes a key-value store \vehicleIDList where
  the values are user accounts (or derived information) that will be
  used in the scheme to track double spending, as well as a list of
  already deposited coins \coinStore, initially empty. The bank tracks 
  account balances for the users, which is updated during withdrawals and deposits. 


 \item $\{ \genericCoin_\keyIdx \}_{\keyIdx \in \nats[\numkeys]}/ \bot \assign   \serverIssueCoin(\rsuCert{\rsuID}, \numkeys)$: This interface takes as input the identity of an \rsu, e.g., in the form of a certification \rsuCert{\rsuID}, the number of coins requested \numkeys (and other auxiliary information), and either issues the specified number of coins, or aborts. If the request is directly from a user, the bank also issues a voucher tying the coin to the user's identity.

\item $0/1 \assign \serverUpdateWithdrawal(\vPubKey{\vID}', \userReceipt, \rsuPubKey{\rsuID})$: This interface takes as input the a receipt signed by the user who has requested
a withdrawal from \genRSU.  This invocation is made by \genRSU to update the balance of the indicated user after a withdrawal. The interface returns a 0 to update the balance was updated correctly, or returns a 1 if the invocation fails, e.g., the account is overdrawn.

\item $\vPubKey{\vID}/\rsuPubKey{\rsuID}/0 \assign \serverVerifyTx(\genericCoin, \genericVoucher, \genericTx, \vPubKey{\genDepRSUID})$: This interface takes as input the identity of the merchant where a transaction has taken place, and all the artifacts related to the transaction, which includes the coin \genericCoin, the corresponding voucher, \genericVoucher, and the transaction record, \genericTx, linking the coin to the merchant. The algorithm checks whether the coin is valid; the voucher records the assignment of this coin to a
  user; and that the transaction record pertains to the spending of
  this coin (by that user).  If the coin is not valid, then the interface aborts.
  Otherwise, the invocation first checks whether the coin is already
  recorded in the log, in which case it triggers a double-spending
  detection algorithm and \emph{outputs the identity (a public key) of the user or the ATM who double-spent/double-issued the coin}.  If there is no double-spending,
  the invocation adds the coin to the log, updates the balance of
  the merchant and outputs $0$ to signify that the transaction has
  been recorded.
  
\end{itemize}

\subsubsection{\rsu:} Each \rsu implements the following algorithms.
\begin{itemize}[nosep,leftmargin=1em,labelwidth=*,align=left]

  \item $\genericRSUIDKeyPair, \genericRSUSigKeyPair \assign \rsuInit(1^{\secParam})$: This interface sets up the secrets that will be used
by the \rsu to dispense coins to users. In our scheme, this corresponds to an identity key pair, \genericRSUIDKeyPair, where the private key is blindly signed by the bank, and a key pair for a digital signature \genericRSUSigKeyPair. 

\item $\{\genericCoin_\keyIdx \}_{\keyIdx \in \nats[\numkeys]}/ \bot \assign \rsuWithdrawCoin(\numkeys, \dots)$: This interface
requests $\numkeys$ coins from
the bank. The bank may decline the request if, e.g., the
bank would like to \textit{rate-limit} the \rsu.

\item $\genericCoin , \genericVoucher \assign \rsuAssignCoin(\vPubKey{\vID})$: This interface takes as input a user $\genCar$'s information, and outputs a voucher that is (unforgeably) linked to a coin; we indicate that a coin is linked to a voucher with the notation $\genericCoin \coinVoucherLink \genericVoucher$. The voucher has a serial number that tied to the user's identity; we do not allow transferring of vouchers between users. The voucher also includes relevant information to verify authenticity.


\end{itemize}
	
\subsubsection{Users} Each user implements the following
algorithms:
\begin{itemize}[nosep,leftmargin=1em,labelwidth=*,align=left]

 \item $\genericUserIDKeyPair, \genericUserSigKeyPair \assign \carInit(1^{\secParam})$: This interface generates a key pair to establish \genericUserIDKeyPair and a signing key pair \genericUserSigKeyPair for a digital signature scheme. The private key \vPrivKey{\vID} is blindly signed by the bank. The interface may also generate other secrets which will be described later. 
 
\item $\genericCoin, \genericVoucher \assign \carGetCoin(\pubKey{\vID}, \privKey{\vID}, \dots)$: This interface takes as input the user's identity key pair, $\langle \pubKey{\vID},\privKey{\vID} \rangle$, and 
other auxiliary information, and requests a coin(s) withdrawal from an \rsu or from the bank.  

\item $\genericTx \assign \carDepositCoin(\genericCoin, \genericVoucher, \vPubKey{\genDepRSUID} , \dots)$: This interface takes as input a coin $\genericCoin$, the
linked voucher $\genericVoucher$, the identity of a merchant $\genDepRSU$ and other auxiliary information, and produces a transaction signifying that the coin
has been spent at that merchant.

\item $0/\bot \assign \rsuVerifyCoin(\genericCoin, \genericVoucher, \genericTx)$: This interface takes as input a transaction
that has taken place between a user and a merchant, and outputs 0 if
the transaction is valid (defined later), otherwise it aborts. The
interface is invoked by a merchant to verify the coin and the voucher, 
and other relevant information, before depositing it to the bank.  


\item $\userReceipt \assign \carIssueReceipt(\privKey{\vID}', \rsuPubKey{\rsuID} , \dots)$: This interface takes as input a signing key $\privKey{\vID}'$, the identity of an ATM $\genRSU$ from where a withdrawal has taken place, and other auxiliary information, and produces a signed (using the user's private key) message indicating that the user withdrew a coin(s) from the \rsu at a specified time. 

\end{itemize}

\mysubsection{Security}
\label{sec:model:security_main}
We will describe the security guarantees provided by
our \ecash scheme. We note that the following properties of unforgeability, anonymity, authenticity 
of balance and identification of double-spenders is provided by all \ecash schemes. Our definitions adapt the model to multiple issuers.


\myparagraph{Unforgeability}
In our scheme, only the bank \bank
can mint new coins. Formally, we define this property with a security 
experiment $\Experiment{\unforgeLabel}{\repCoinScheme}$, 
where the adversary represents the malicious users and \rsus, and 
is able to invoke the bank, the honest \rsus and the honest users as oracles.
It eventually produces a coin \genericCoin, a voucher \genericVoucher, possibly linked to the coin, and a transaction \genericTx that spends the coin.  The adversary wins if the bank accepts the transaction, and the coin was not issued by the bank.

\myparagraph{Anonymity}
We will ensure that the coins spent at a merchant and the transactions
submitted to the bank do not reveal any information about the spender,
unless the coin has been double spent. We define this property with an
indistinguishability experiment $\Experiment{\anonLabel[\anonBit]}{\repCoinScheme}(\anonAdversary[1], 
    \anonAdversary[2])$ in \appref{sec:model:security}, 
where a distinguisher -- signifying the
bank, a merchant or an ATM -- receives transactions deposited by two
honest users and attempts to distinguish between them.

\myparagraph{Authenticity of balance}
We will ensure that an adversary, representing dishonest users and ATMs, cannot dishonestly withdraw from an honest user's account. They also cannot forge unauthorized transactions to credit their own accounts using coins that have been withdrawn by honest users. (This also ensures that an honest user cannot be falsely implicated for double spending). We formally define this property with the experiment $\Experiment{\frameLabel}{\repCoinScheme}(\frameAdversary)$ in \appref{sec:model:security}.

\myparagraph{Fair exchange}
We will ensure that a user that receives a coin during a withdrawal must provide a receipt to the ATM otherwise it is unable to spend the coin without being detected by the bank. Combined with authenticity of balance where a user does not provide a receipt until it receives a coin from the ATM, this property implies a fair exchange enabled through the bank acting as the trusted third party\footnote{It is known that a fair exchange without a TTP is impossible~\cite{pagnia1999impossibility,sandholm2002possibility}}. We define this property using the experiment $\Experiment{\fairexchangeLabel}{\repCoinScheme}(\fairexchangeComb)$ in \appref{sec:model:security}. In our protocol, the bank will play the role of the trusted third party.

\myparagraph{Identification of double spenders}
We will ensure that two valid transactions double spending the same coin will identify the user responsible. Alternatively, if the same coin is issued to two different users by a malicious ATM, then the identity of the offending ATM will be disclosed to the bank. We formally define this with the experiment $\Experiment{\dblSpendLabel}{\repCoinScheme}(\dblSpendAdversary[1], \dblSpendAdversary[2])$ in \appref{sec:model:security}.

\myparagraph{Untraceability}
We will introduce a new security property that is important for schemes such as ours that can dispense cash through intermediaries. If a voucher includes information about the \rsu that issued it, e.g., in the form of a signature, then the voucher will divulge information regarding the whereabouts of the user to the merchant and the bank. \ecash systems typically do not need to protect this information since they only allow withdrawals directly from the bank. However, as we are assuming that \rsus can be physical entities, protecting the user's locations from untrusted merchants becomes critical.

We will ensure that information regarding the honest \rsu where an honest 
user withdrew a coin is not available to the merchant where it deposits coins and to the bank. We 
formalize this using an indistinguishability experiment  $\Experiment{\locPrivacyLabel[\locPrivacyBit]}{\repCoinScheme}(\locPrivacyAdversary[1],
     \locPrivacyAdversary[2])$ in \appref{sec:model:security} as an inability to distinguish
a coin + voucher dispensed by one of two honest \rsus to an 
honest user of the adversary's choice

\myparagraph{Unforgeable receipts} We will ensure that the receipts can be used by 
\rsus and users if there are conflicts regarding account updates. Thus, receipts should 
be \textit{unforgeable}. We omit a formal definition since receipts are essentially a 
signature on all the relevant information (i.e., the identities of the users, \rsus, 
timestamp, etc.) 
which can be implemented using a standard digital signature scheme ensuring unforgeability under chosen message attacks.

\subsubsection{Non-goals}
The system does not address denial-of-service attacks where a malicious 
\rsu does not provide well-formed vouchers to honest users, or overloads the bank by 
flooding withdrawal/deposit requests.

\subsection{Cryptographic Tools}
\label{sec:defn}

\subsubsection{Pseudorandom function}
We will use a pseudorandom function (PRF) family. A
uniformly sampled PRF produces outputs that are
indistinguishable from a random function.
We recall the definition of a PRF using the formal experiment 
in \appref{}. In our constructions, we particularly use the
function family $\dyPRFDefn$ over the group \grp with order $\grpOrd$, due to \citet{dodis2005verifiable}. Given a key $\prfKey[1] \sample \integersModArg{\grpOrd}$ 
and an input $\genericFnInput \in \integersModArg{\grpOrd}$, the PRF computes $\dyPRF[\prfKey[1]](\genericFnInput) = \grpGen^{\frac{1}{1 + \prfKey[1] + \genericFnInput}}$, where $\grpGen$ is a generator of \grp.

\subsubsection{Commitment scheme}
We will use a cryptographic commitment scheme $\commSchemeDefn$, where \grp is a group with order \grpOrd.  The algorithm
$\commFn(\plaintext; \genericRand{})$ is used to commit to a set of messages
$\plaintext[1], \dots, \plaintext[\numkeys] \in \integersModArg{\grpOrd}$ with a random coin $\genericRand{} \sample \integersModArg{\grpOrd}$. Wherever it is not necessary, we will drop the random
coin in the notation for brevity. We will require
the scheme provides computational hiding and binding
properties. Intuitively, the hiding
property implies that a commitment to a message \plaintext hides the
message against an adversary. The binding property
ensures that an adversary
cannot produce two messages that produce the same commitment. We recall these
properties formally in \appref{}. 

We will use generalized Pedersen commitments~\cite{pedersen1991non}. The public parameters include \grp, and $\numkeys+1$ generators $\grpGen[1], \dots, \grpGen[\numkeys], \grpGen'$. In order to commit to a set of values $\plaintext[1], \dots, \plaintext[\numkeys]\in \integersModArg{\grpOrd}$, $\commFn(\plaintext[1], \dots, \plaintext[\numkeys]; \genericRand{}) = \grpGen[1]^{\plaintext[1] } \dots \grpGen[\numkeys]^{\plaintext[\numkeys]} (\grpGen')^{\genericRand{}}$, where $\genericRand{} \sample \integersModArg{\grpOrd}$. 

\subsubsection{Digital signature}

We will use a digital signature scheme $\sigScheme = \langle\signKeyGen, \sign{},
\vrfy{}\rangle$ where \signKeyGen generates a public-private key pair
$\langle \pubKey, \privKey\rangle$.  The signing algorithm \dsign{} on
input \plaintext produces a signature $\genericSig \assign
\dsign{\privKey}(\msg)$.  The verification algorithm
$\vrfy{\pubKey}(\msg, \genericSig)$ outputs 1 if the input is a valid
signature under $\privKey$ and outputs 0 otherwise. For correctness we
require that $\vrfy{\pubKey}(\dsign{\privKey}(\msg), \msg) = 1$ for any
$\langle \pubKey, \privKey\rangle$ output from \signKeyGen.  For security, we
require unforgeability, which implies that an adversary 
cannot produce a signature on a message of its own choice under
\privKey, for which it has not queried a signing oracle before.
We recall the definition in \appref{}

\subsubsection{Blind signature}
A blind signature scheme is a signature scheme in which the signer can create signatures on unknown messages~\cite{chaum1983blind}. The signature scheme is defined by a tuple of algorithms 
$\blindSigSchemeDefn$. Here, $\bkeygen(1^\secParam)$ generates a public-private key pair, $\bblind(\plaintext; \genericRand{})$ blinds the message \plaintext to be signed using a random blind \genericRand, $\bsign{\privKey}(\bblind(\plaintext; \genericRand{}))$ computes a signature $\genericSig$ on the blinded message, $\bunblind(\genericSig; \genericRand{})$ unblinds the signature to obtain a signature on the original message, and $\bvrfy{\pubKey}(\plaintext, \genericSig)$ verifies a signature.  We will use a standard RSA blind signature~\cite{pointcheval1996provably}.

\subsubsection{Blind signature with proofs of knowledge}
We will use a blind signature scheme where the holder can prove in zero-knowledge that it has a signature on a committed value. Specifically, such a signature scheme is defined by the tuple of algorithms \blindSigSchemeZKPDefn. The algorithm \zkeygen creates a public-private key pair. The algorithm $\genericSig \assign \zsign[\privKey](\genericCommitment)$ takes as input a commitment to a value, and produces a signature on the value that has been committed to. The algorithm $\zkp \assign \zprove{\pubKey}(\genericCommitment, \genericSig)$ takes as input a commitment to a value and a signature on the committed value, and produces a NIZK proof of knowledge of a signature on the opening of the commitment. Finally, $\zvrfy{\pubKey}(\genericCommitment, \zkp)$ takes as input a commitment and a NIZK proof produced by \zprove, and returns 1 if the proof shows that the party knows a signature on the committed value. 
Several schemes provide these properties, such as CL signatures~\cite{camenisch2003signature} and BBS
signatures~\cite{camenisch2004signature}.

\subsubsection{Non-interactive zero-knowledge proofs of knowledge}
\label{sec:background:zkp}
Our protocol leverages a noninteractive zero-knowledge (NIZK) proof of
knowledge $\zkpName = \langle \zkpGen$, \zkpVerify, \zkpSim,
$\zkpExtract\rangle$ in the random oracle
model~\cite{bellare1993oracles}.  Let \relation be the witness
relation for language $\lang \subseteq \{0, 1\}^{\ast}$, and let
\randomOracles denote the set of all functions from $\{0, 1\}^{\ast}$
to $\{0,1\}^{\secParam}$.  On input $\langle \statement, \witness\rangle
\in \relation$ and with access to a random oracle $\randomOracle
\sample \randomOracles$,
$\zkpGen[\witness]^{\randomOracle}(\statement)$ produces a proof \zkp
(using the witness \witness) that $\statement \in \lang$, so that if
$\zkp \assign \zkpGen[\witness]^{\randomOracle}(\statement)$ then
$\zkpVerify^{\randomOracle}(\statement, \zkp)$ returns true.  As is
typical, we require two properties from the proofs of knowledge: i)
the existence of a knowledge extractor \zkpExtract, and ii) a
proof simulator \zkpSim. We define the properties formally in \appref{}.

\mysection{Protocol Design}
\label{sec:scheme}

\begin{table}[th!]
\small
\caption{Table of notations\label{tab:notation}}
\begin{tabular}{|l|l|}
	\hline
     \multicolumn{2}{c}{\textbf{Cryptographic schemes}} \\
    \hline
    $\sigScheme$ & standard digital signature scheme \\
    $\blindSigScheme$ & blind signature scheme \\
    $\blindSigSchemeZKP$ & blind signature scheme with proofs of knowledge (e.g., CL-signature~\cite{camenisch2003signature})\\
     $\dyPRF$ & Dodis-Yampolinsky PRF scheme over group \grp with generator $\grpGen$ \\
	$\commFn$ & generalized Pedersen commitment over group \grp\\
    \hline 
    \multicolumn{2}{c}{\textbf{Keys}} \\
    \hline
       $\langle \sPubKey, \sPrivKey\rangle$ & bank's key pair for a digital signature $\sigScheme$  \\
	$\langle \sPubKeyBlind, \sPrivKeyBlind \rangle$ & bank's key pair for a blind signature $\blindSigScheme$\\
	$\langle \sPubKeyBlindZKP, \sPrivKeyBlindZKP \rangle$ & bank's key pair for a blind signature with proofs of knowledge \blindSigSchemeZKP\\
	$\langle \grpGen^{\rsuPrivKey{\rsuID}}, \rsuPrivKey{\rsuID} \rangle$ & \genRSU's identity key pair; $\rsuPrivKey{\rsuID} \sample \integersModArg{\grpOrd}$ \\
	$\langle \grpGen^{\privKey{\vID}}, \privKey{\vID} \rangle$ & user \genCar's identity key pair; $\privKey{\vID} \sample \integersModArg{\grpOrd}$ \\
   $\userPRFKey$ & User $\genCar$'s spending PRF key; $\userPRFKey \sample \integersModArg{\grpOrd}$\\
   $\genericSig[\privKey{\vID}]$ & signature under $\sPrivKeyBlindZKP$ on $\langle \privKey{\vID}, \userPRFKey\rangle$ \\
$\genericSig[\rsuPrivKey{\rsuID}]$ & signature under $\sPrivKeyBlindZKP$ on $\rsuPrivKey{\rsuID}$ \\
   \hline 
   \multicolumn{2}{c}{\textbf{Components of a coin}}
   \\
   \hline
	$\keyCommitment[1]$ & commitment to PRF key $\prfKey[1]$ sampled by ATM  \\
	$\keyCommitment[2]$ & commitment to PRF key $\prfKey[2]$ sampled by ATM  \\
	$\rsuPrivKeyCommitment$ & commitment to ATM's private identity key $\rsuPrivKey{\rsuID}$ \\
	$\langle \crHash(\keyCommitment[1], \keyCommitment[2], \rsuPrivKeyCommitment), \signedCoin \rangle$ & coin format; $\signedCoin$ is a signature on  
$\crHash(\keyCommitment[1], \keyCommitment[2], \rsuPrivKeyCommitment)$ under $\sPrivKeyBlind$\\
\hline
 \multicolumn{2}{c}{\textbf{Components of vouchers and transactions (except NIZK proofs)}}\\
\hline
$\privKeyCommitment$ & commitment to user $\genCar$'s secrets $\langle \privKey{\vID}, \userPRFKey \rangle$\\
$\voucherPRFEval$ & unique unforgeable identifier tying a specific coin to user's identity\\
$\voucherDblSpendToken$ & double-issuing token to identity ATMs who double-issue coins\\
 $\txPRFEval $ & double-spending token to identify users who double spend coins\\

\hline

\end{tabular}
\end{table}

In this section, we will describe two versions of our \ecash protocol. The first protocol \noncompactProtocol sacrifices compactness for faster withdrawals and deposits. The protocol \compactProtocol introduces compactness but at the cost of higher withdrawal and spending latencies. In the following, \grp is a prime-order group with order \grpOrd, \grpGen is a random generator of \grp, $\dyPRFDefn$ is the Dodis-Yampolinsky PRF, and $\commFn$ is the generalized Pedersen commitment scheme. For better readability, we summarize the key notations in \tblref{tab:notation}.

\mysubsection{Description of Non-Compact Protocol}

\subsubsection{Initialization} The bank initializes by selecting signing keys for three types of signature schemes: 
a key pair for \sigScheme, a key pair for a blind signature scheme \blindSigScheme and a key pair for a blind signature scheme with proofs of knowledge \blindSigSchemeZKP (see \tblref{tab:notation}). ATMs initialize
by selecting a public-private key for establishing identity $\langle \rsuPrivKey{\rsuID}, \grpGen^{\rsuPrivKey{\rsuID}} \rangle$. The bank signs the private identity key using \blindSigSchemeZKP. The bank stores $\grpGen^{\rsuPrivKey{\rsuID}}$ and other information about the ATM.


Each user initializes by registering with the bank. For this, the user samples a random PRF key $\userPRFKey$, and a key pair $\langle \vPrivKey{\vID}, \grpGen^{\vPrivKey{\vID}} \rangle$ as its identity keys. The bank signs both $\userPRFKey$ and $\vPrivKey{\vID}$ as a pair using $\blindSigSchemeZKP$. During signing, the user also proves in zero-knowledge that the public key $\grpGen^{\vPrivKey{\vID}}$ is correctly derived from the private key. The bank stores $\grpGen^{\vPrivKey{\vID}}$ and other information about the user.

\begin{figure}[t]
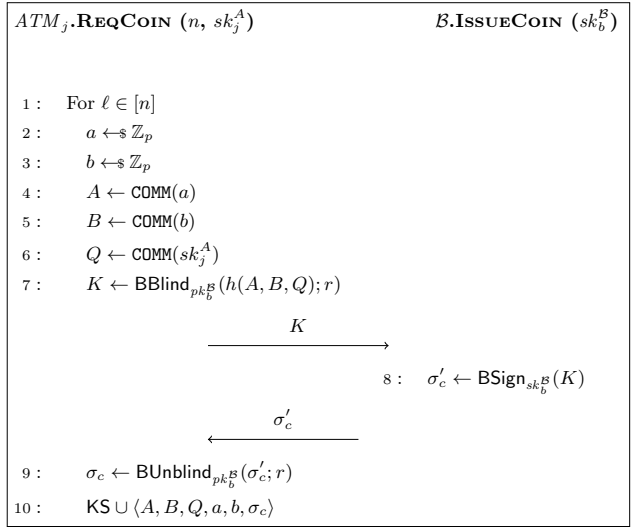

\centering
\resizebox{0.8\textwidth}{!}{\vbox{%
    \pcb{
    \textbf{\genRSU.\rsuWithdrawCoin(\numkeys, \rsuPrivKey{\rsuID})} \hspace{3cm} \textbf{\server.\serverIssueCoin(\sPrivKeyBlind)} \\
    \\[0.05\baselineskip] \\ [-0.3\baselineskip]
	\pcln \text{For} ~\keyIdx \in \nats[\numkeys] \\
	\pcln\label{withdraw:selectKeys1} \t \prfKey[1] \sample \integersModArg{\grpOrd} \\
	\pcln\label{withdraw:selectKeys2} \t \prfKey[2] \sample \integersModArg{\grpOrd} \\
	\pcln\label{withdraw:commitKeys1} \t \keyCommitment[1] \assign  \commFn(\prfKey[1])\\
	\pcln\label{withdraw:commitKeys2} \t \keyCommitment[2] \assign  \commFn(\prfKey[2])\\
	\pcln \t \rsuPrivKeyCommitment \assign \commFn(\rsuPrivKey{\rsuID})\\
    \pcln\label{withdraw:blind} \t \hashedCoin \assign \bblind[\sPubKeyBlind](\hashFn(\keyCommitment[1],\keyCommitment[2], \rsuPrivKeyCommitment); \genericRand{})\\		
    \hspace{3cm} \sendmessageright*[3cm]{\hashedCoin} \\
    \hspace{6cm} \pcln\label{withdraw:bsign} \signedCoin' \assign \bsign{\sPrivKeyBlind}(\hashedCoin)\\
    \hspace{3cm}\sendmessageleft*[2.5cm]{\signedCoin'}\\
    \pcln\label{withdraw:unblind} \t \signedCoin \assign \bunblind[\sPubKeyBlind](\signedCoin'; \genericRand{})\\
    \pcln\label{} \t \rsuKeyStore \cup \langle \keyCommitment[1], \keyCommitment[2], \rsuPrivKeyCommitment, \prfKey[1], \prfKey[2], \signedCoin \rangle 
    }
}}
\caption{\noncompactProtocol: Coin requests by ATM\label{fig:coin_req}}
\end{figure}

\myparagraph{\rsu coin requests}
An \rsu periodically requests new coins from the bank using the interface $	\rsuWithdrawCoin(\numkeys)$ interface (\figref{fig:coin_req}). This interface requires as input the number of coins $\numkeys$. (We implicitly assume that the bank can verify that the request is from an authorized ATM, e.g, using certificates).

The ATM randomly selects $\numkeys$ keys $\langle \prfKey[1], \prfKey[2] \rangle_{\keyIdx \in \nats[\numkeys]}$ for $\dyPRF$ and commits to them (\stepref{withdraw:selectKeys1}). The ATM also commits to its private identity key. The bank blindly signs the hash of the committed keys (\stepsref{withdraw:blind}{withdraw:unblind}), and \emph{each pair of signed commitments + blind signature constitutes a coin}. The ATM stores the coins and the corresponding PRF keys in a local key store \coinStore. These coins are later dispensed to the users. The bank may refuse the \rsu's request, e.g., if the bank wants to rate-limit the ATM from issuing more coins.


\begin{figure}[t!]
\resizebox{0.79\textwidth}{!}{\vbox{%
    \pcb{
   \underline{\genCar.\carGetCoin(\privKey{\vID}, \privKey{\vID}', \userPRFKey,  \userCLSig{\privKey{\vID}})} \hspace{6cm} \underline{\genRSU.\rsuAssignCoin    (\rsuPrivKey{\rsuID}, \prfKey[1],  \prfKey[2], \userCLSig{\rsuPrivKey{\rsuID} })} \\
    \\[0.05\baselineskip] \\ [-0.3\baselineskip]
    \pcln\label{getcoin:comm} \privKeyCommitment \assign \commFn(\vPrivKey{\vID}, \userPRFKey)\\
    \pcln\label{getcoin:zkp}\zkp{\privKey{\vID}} \assign \zprove{\sPubKeyBlindZKP}(\privKeyCommitment, \userCLSig{\privKey{\vID}}) \\
    \hspace{2cm} \sendmessageright*[2cm]{\vPubKey{\vID}, \privKeyCommitment, \zkp{\privKey{\vID}}} \\
    \hspace{4cm} \pclinecomment{Check user balance before \stepsref{assign:vrfyCLSig}{assign:computeZKP}} \\
    \hspace{6cm} \pcln\label{assign:vrfyCLSig} \text{If}~\zvrfy{\sPubKeyBlindZKP}(\privKeyCommitment, \zkp{\privKey{\vID}}) \neq 0, \text{then} ~\abort\\
	\hspace{6cm}\pcln \coinRand \assign \crHash(\privKeyCommitment)\\
    \hspace{6cm}\pcln\label{assign:computeID} \coinID \assign \coinRand + 1 \bmod \grpOrd \\
    \hspace{6cm}\pcln\label{assign:getKey} \langle \keyCommitment[1], \keyCommitment[2], \rsuPrivKeyCommitment, \prfKey[1], \prfKey[2], \signedCoin \rangle \assign \rsuKeyStore.\rsuGetKey()\\
 \hspace{6cm}\pcln\label{assign:computeVoucherPRF} \voucherPRFEval \assign \dyPRF[\prfKey[1]](\coinID)  \quad \pclinecomment{tie coin to user identity} \\
\hspace{6cm}\pcln\label{assign:computeVoucherDblToken} \voucherDblSpendToken \assign \rsuPubKey{\rsuID} \cdot \dyPRF[\prfKey[2]](0)^{\coinRand} \quad \pclinecomment{double issuing token}\\
\hspace{6cm} \pcln \zkp{\rsuPrivKey{\rsuID}} \assign \zprove{\sPubKeyBlindZKPAlt}(\rsuPrivKeyCommitment, \genericSig[\rsuPrivKey{\rsuID}])\\
\hspace{6cm} \pclinecomment{tie the secrets in the coin to a user's identity using a NIZK proof}\\
    \hspace{6cm}\pcln\label{assign:computeZKP} \zkp{\coinID} \assign \zkpGen[\alpha_1, \beta_1, \alpha_2, \beta_2, \gamma, \delta](\keyCommitment[1], \keyCommitment[2], \voucherPRFEval, \voucherDblSpendToken, \rsuPrivKeyCommitment, \coinID, \coinRand) \\
    \hspace{6cm}~\quad\{\keyCommitment[1] = \commFn(\alpha_1; \beta_1) ~\land~ \keyCommitment[2] = \commFn(\alpha_2; \beta_2) \\ \hspace{6cm}~\quad \t ~\land~\rsuPrivKeyCommitment = \commFn(\gamma, \delta) \land~ \voucherPRFEval  = \dyPRF[\alpha_1](\coinID) \\
    \hspace{6cm} ~\quad \t \t \land \voucherDblSpendToken = \grpGen^{\gamma}\cdot\dyPRF[\alpha_2](0)^{\coinRand} \} \\
	\hspace{6cm} \pcln \genericCoin \gets \langle \keyCommitment[1],  \keyCommitment[2], \rsuPrivKeyCommitment, \signedCoin\rangle\\
	\hspace{6cm} \pcln \genericVoucher \gets  \langle \privKeyCommitment, \zkp{\privKey{\vID}}, \voucherPRFEval, \voucherDblSpendToken, \zkp{\coinID}, \zkp{\rsuPrivKey{\rsuID}}, \coinRand \rangle\\
	\hspace{6cm} \pcln \nonce \sample \bin^{\secParam} \\
	\hspace{6cm} \pcln \intentMsg \assign \roHash(\genericCoin, \vPubKey{\vID}, \rsuPubKey{\rsuID}) \\
	\hspace{6cm}\label{assign:intent_msg} \pcln \Sigma \assign \dsign{\rsuSigPrivKey{\rsuID}}(\intentMsg, \genericVoucher, \nonce)\\
    \hspace{6cm}\sendmessageleft*[2cm]{\Sigma,\intentMsg, \genericVoucher, \nonce}\\
	\underline{\genCar.\carIssueReceipt(\privKey{\vID}')} \\
	\pcln \text{If} ~\vrfy{\rsuSigPubKey{\rsuID}}(\Sigma, \intentMsg, \genericVoucher, \nonce) \neq 0, ~\text{then} ~\abort\\
	\pcln\label{assign:user_receipt} \userReceipt \assign \dsign{\vSigPrivKey{\vID}}(\langle \mathsf{withdrawal}, \vPubKey{\vID}, \rsuPubKey{\rsuID}, \nonce \rangle)   \\
    \hspace{6cm}\sendmessageright*[1.5cm]{\userReceipt}\\
	\hspace{6cm} \pcln \text{If} ~\vrfy{\vSigPubKey{\vID}}(\userReceipt, \langle \mathsf{withdrawal}, \vPubKey{\vID}, \rsuPubKey{\rsuID}, \nonce \rangle) \neq 0, ~\text{then} ~\abort \\ 
    \hspace{6cm}\sendmessageleft*[1cm]{\genericCoin}\\
	\pcln \label{assign:complain_to_bank} \text{If \genRSU does not send \genericCoin} \vee ~\intentMsg \neq \roHash(\genericCoin, \vPubKey{\vID}, \rsuPubKey{\rsuID}) \\
	\pcln\t \text{Inform $\bank$ by sending message $\langle \abortWithdrawal$ $\Sigma$, $\intentMsg$, $\genericVoucher$, $\nonce$, $\vPubKey{\vID}$, $\rsuPubKey{\rsuID} \rangle$ and \abort} \\  
    \pcln\text{store} ~\genericCoin, \genericVoucher  \\
 \hspace{6cm} \pcln\label{withdrawal:update_bank} \server.\serverUpdateWithdrawal(\vPubKey{\vID}, \userReceipt, \rsuPubKey{\rsuID}, \nonce) \\ 
\\
\pclinecomment{If $\bank$ receives $\langle \abortWithdrawal$ $\Sigma$, $\intentMsg$,$\genericVoucher$, $\nonce$, $\vPubKey{\vID}$, $\rsuPubKey{\rsuID} \rangle$, then it stores all the values in a invalid coin list $\invalidList$}
  }
}}
\caption{\noncompactProtocol: Withdrawing a coin\label{fig:withdrawal}}
\end{figure}

%
\subsubsection{Withdrawing from \rsu}
When a user request a withdrawal from an \rsu, it is issued a coin(s) and a voucher. This is processed using the $\rsuAssignCoin$ interface (\figref{fig:withdrawal}). At a high level, the algorithm spends a coin that has been requested by the ATM before using the \rsuWithdrawCoin interface, where the user acts as the merchant. The transaction record becomes the voucher. 

In more details, the user commits to its identity private key in $\privKeyCommitment$, and proves in zero-knowledge that is has previously obtained a signature from the bank on the key (\stepsref{getcoin:comm}{getcoin:zkp}). After verifying the signature and ensuring that the 
user has enough balance for a withdrawal, the \rsu maps $\privKeyCommitment$ to an element in $\integersModArg{\grpOrd}$ by using a collision-resistant hash function, and computes a voucher ID $\coinID$ derived from the the commitment. 
Then, the \rsu retrieves metadata from \coinStore which includes i) the commitments to two PRF keys and the \rsu's private identity key, $\keyCommitment[1], \keyCommitment[2], \rsuPrivKeyCommitment$ respectively, ii) the PRF keys committed to $\prfKey[1], \prfKey[2]$, and iii) the bank's signature $\signedCoin$. The commitments to the PRF keys and the \rsu's private key along with the bank's signature constitutes a coin that will issued to the user.
The \rsu computes a voucher that is linked to the coin, and primarily includes two tokens \voucherPRFEval, \voucherDblSpendToken that i) tie the user's private key to the coin using $\prfKey[1]$ (\stepref{assign:computeVoucherPRF}), and ii) tie the issuing \rsu's identity to the coin + voucher in a double-issuing token computed using $\prfKey[2]$ (\stepref{assign:computeVoucherDblToken}). Finally, the \rsu computes two NIZK proofs showing that $\rsuPrivKeyCommitment$ is in fact a commitment to the \rsu's private key that has been previously (blindly) signed by the bank, and that the tokens $\voucherPRFEval, \voucherDblSpendToken$ have been correctly computed (\stepref{assign:computeZKP})

\begin{figure}[t!]
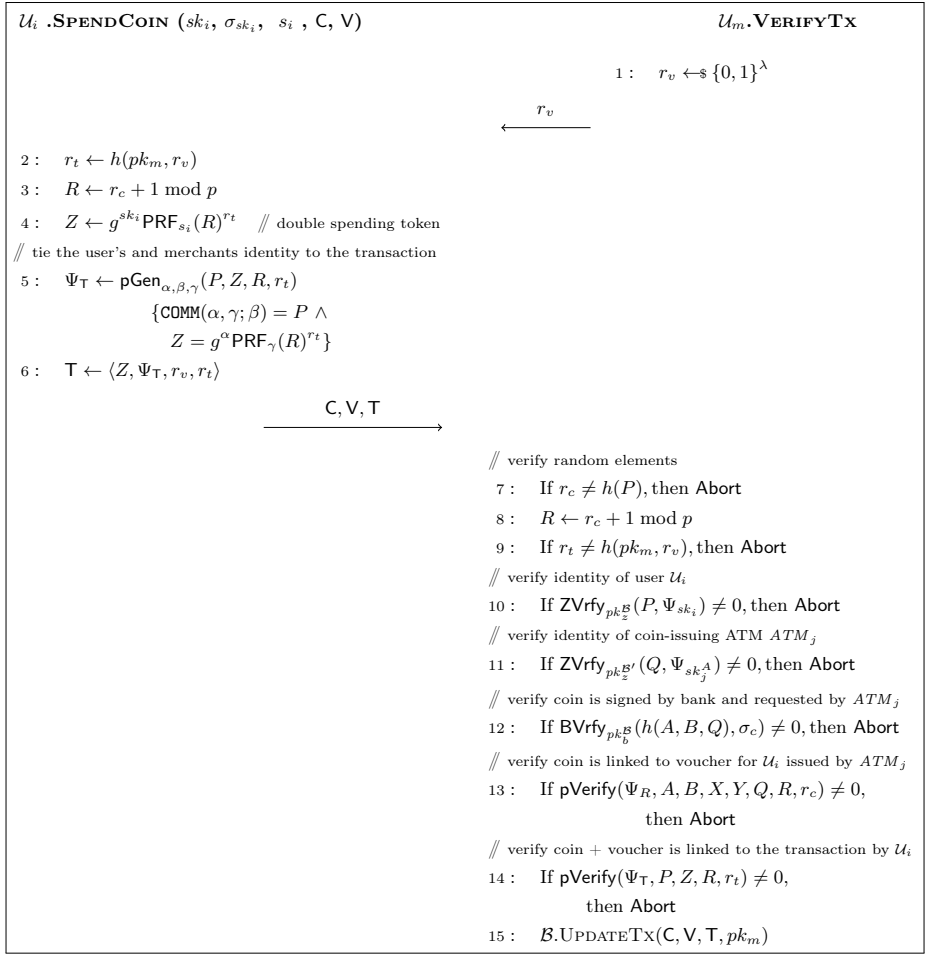

\resizebox{0.79\textwidth}{!}{\vbox{%
    \pcb{
  \textbf{\genCar.\carDepositCoin(\privKey{\vID}, \userCLSig{\privKey{\vID}}, \userPRFKey, \genericCoin, \genericVoucher)} \hspace{6cm} \textbf{\genDepRSU.\rsuVerifyCoin}
\\[0.05\baselineskip] \\ [-0.3\baselineskip]
    \hspace{10cm}\pcln\label{} \coinRandVehicle \sample \bin^{\secParam} \\
\hspace{8cm} \sendmessageleft*[1.5cm]{\coinRandVehicle}\\
    \pcln\label{} \txRand \assign \crHash(\vPubKey{\genDepRSUID}, \coinRandVehicle) \\
	\pcln\label{} \coinID \assign \coinRand + 1 \bmod \grpOrd \\
    \pcln\label{deposit:dblSpendToken} \txPRFEval \assign \grpGen^{\vPrivKey{\vID}} \dyPRF[\userPRFKey](\coinID)^{\txRand}  \quad \pclinecomment{double spending token} \\
\pclinecomment{tie the user's and merchants identity to the transaction}\\
    \pcln\label{deposit:zkpTx} \zkp{\genericTx} \assign \zkpGen[\alpha, \beta, \gamma](\privKeyCommitment, \txPRFEval, \coinID, \txRand) \\
  \hspace{2cm} \t \{ \commFn(\alpha, \gamma; \beta) = \privKeyCommitment  ~\land \\
   \hspace{2cm} \t \t  \txPRFEval = \grpGen^{\alpha} \dyPRF[\gamma](\coinID)^{\txRand} \}\\
  \pcln  \genericTx \assign \langle \txPRFEval, \zkp{\genericTx}, \coinRandVehicle, \txRand \rangle \\
    \hspace{4cm} \sendmessageright*[3cm]{\genericCoin, \genericVoucher, \genericTx} \\
\hspace{8cm}\pclinecomment{verify random elements} \\
\hspace{8cm}\pcln\label{verifyTx:coinRand} \text{If} ~\coinRand \neq \crHash(\privKeyCommitment), \text{then} ~\abort \\
\hspace{8cm}\pcln \coinID \assign \coinRand + 1 \bmod \grpOrd \\
\hspace{8cm}\pcln\label{verifyTx:txID}  \text{If} ~\txRand \neq \crHash(\vPubKey{\genDepRSUID}, 	\coinRandVehicle), \text{then} ~\abort \\
\hspace{8cm}\pclinecomment{verify identity of user $\genCar$} \\
\hspace{8cm} \pcln\label{verifyTx:clSigCommK} \text{If} ~\zvrfy{\sPubKeyBlindZKP}(\privKeyCommitment, \zkp{\privKey{\vID}}) \neq 0, \text{then} ~\abort\\
\hspace{8cm}\pclinecomment{verify identity of coin-issuing ATM $\genRSU$} \\
\hspace{8cm}\pcln\label{verifyTx:clSigCommPrivKeyRSU} \text{If} ~\zvrfy{\sPubKeyBlindZKPAlt}(\rsuPrivKeyCommitment, \zkp{\rsuPrivKey{\rsuID}}) \neq 0, \text{then} ~\abort\\
\hspace{8cm}\pclinecomment{verify coin is signed by bank and requested by $\genRSU$} \\
\hspace{8cm}\pcln\label{verifyTx:coin} \text{If} ~\bvrfy{\sPubKeyBlind}(\crHash(\keyCommitment[1], \keyCommitment[2], \rsuPrivKeyCommitment), \signedCoin) \neq 0, \text{then} ~\abort  \\
\hspace{8cm}\pclinecomment{verify coin is linked to voucher for $\genCar$ issued by $\genRSU$ } \\
\hspace{8cm}\pcln\label{verifyTx:coinID} \text{If} ~ \zkpVerify(\zkp{\coinID}, \keyCommitment[1], \keyCommitment[2], \voucherPRFEval, \voucherDblSpendToken, \rsuPrivKeyCommitment, \coinID, \coinRand) \neq 0,\\
\hspace{9cm}\hspace{1cm} \t \t \text{then} ~\abort 	 \\
\hspace{8cm}\pclinecomment{verify coin + voucher is linked to the transaction by $\genCar$} \\
\hspace{8cm}\pcln\label{verifyTx:txProof} \text{If} ~ \zkpVerify(\zkp{\genericTx}, \privKeyCommitment, \txPRFEval, \coinID, \txRand) \neq 0, \\
\hspace{8cm}\hspace{1cm} \t \t \text{then} ~\abort \\
\hspace{8cm}  \pcln \server.\serverVerifyTx(\genericCoin, \genericVoucher, \genericTx, \vPubKey{\genDepRSUID})   
    }
}}
\caption{\noncompactProtocol: Spending a coin\label{fig:spend}}
\end{figure}

Note that the bank's signature on the (hash of) signed commitments in the coin shows that the voucher has been computed by an ATM that has been authorized to issue the coin by the bank. By including the double-issuing token $\voucherPRFEval$ in the voucher, we ensure that if the same coin is issued by the ATM twice, then with overwhelming probability the bank will learn the identity of the offending ATM. Evaluating the PRF under the committed keys on the user's private key commitment ensures that the coin + voucher is cryptographically linked to the user's identity. This is violated only if the user can open the commitment to multiple value, or modify the tokens in the voucher and forge a NIZK proof without the witnesses, i.e., the PRF keys committed to in $\keyCommitment[1], \keyCommitment[2]$ by the ATM. Both these cases should happen with only negligible probability.

\subsubsection{Fair exchange of coins}
The $\carIssueReceipt$ interface is invoked by the user when receiving a coin from \genRSU to produce a receipt. However, to ensure an optimistic fair exchange, we follow a three-step process. First, the ATM commits to the coin by computing a hash of the coin, \intentMsg, using a random oracle $\ro$, and public keys of the user and the ATM. The ATM signs the hash, the voucher and a random nonce which acts as a unique identifier for the withdrawal and sends the signature to the user (\stepref{assign:intent_msg}). The user verifies the signature and then sends to the ATM a receipt which is essentially a signature on the nonce and the public keys of the user and the ATM (\stepref{assign:user_receipt}). After verifying the receipt, the ATM sends the coin to the user. If the coin/voucher is not well-formed, or the user does not receive the coin from the ATM after sending the receipt, then the user sends all the information received from the ATM to the bank in an \abortWithdrawal message. 

Upon receiving the \abortWithdrawal message, the bank stores the commitment and identities of the user and the ATM in an invalid coin list \invalidList. If a coin is added to the invalid list, then it should not be spent in any transaction later. Thus, for every transaction submitted, the bank checks whether the coin $\genericCoin'$ that is part of the transaction is in the invalid list by checking whether $\intentMsg = \ro(\genericCoin', \vPubKey{\vID}, \rsuPubKey{\genRSUID})$.  If there is a match, the bank can distinguish whether the coin was issued by the ATM to a different user by checking if the transaction has a different voucher than the one in \invalidList. If so, it can verify that both of the vouchers were issued by the same ATM by using the double-issuing token (as will be discussed later). Otherwise, the bank identifies that the user sent a false \abortWithdrawal message, and takes appropriate action.

\subsubsection{Spending a coin}

User $\genCar$ invokes $\carDepositCoin$ to spend a coin at merchant. Intuitively, this process ties the coin to the merchant's identity using a transaction. 
For this, the user first computes a ``double-spending token'' using the PRF key signed by the bank during initialization. In particular, the token ties the voucher ID \coinID 
to the private key of the user and the public key of the merchant (\stepref{deposit:dblSpendToken}). The user also generates proofs of knowledge to show that the private key embedded in the token corresponds to the private key 
in the commitment \privKeyCommitment that is included as part of the voucher. The merchant is provided the coin, the voucher, the double spending token and the proof of knowledge. The merchant verifies the transaction, and if successful, reports the transaction to the bank.

\subsubsection{Identifying double spenders}

\begin{figure}[t!]
\centering
\resizebox{0.8\textwidth}{!}{\vbox{%
    \pcb{
       \textbf{\bank.\serverVerifyTx(\genericCoin, \genericVoucher, \genericTx, \genDepRSU)}\\
	\pclinecomment{follow \stepsref{verifyTx:coinRand}{verifyTx:txProof} of \rsuVerifyCoin. \abort if verification fails} \\
	\pcln \text{If} ~\exists \langle \Sigma, \intentMsg, \genericVoucherAlt, \nonce, \vPubKey{\vID}, \rsuPubKey{\rsuID} \rangle \in \invalidList. ~\land ~\intentMsg = \ro(\genericCoin,  \vPubKey{\vID}, \rsuPubKey{\rsuID}) \\
\pcln \t \text{If} ~\genericVoucher = \genericVoucherAlt \\
\pcln \t \t \pcreturn \vPubKey{\vID} \\
\pcln \t \text{If} ~\genericVoucher = \genericVoucherAlt\\
\pcln \t \t \pclinecomment{check for double-issuing by following steps \stepsref{dbl:atm_double_issue}{dbl:atm_double_issue_output}} \\
\pcln\label{} \text{If} ~\genericCoin \notin \coinStore ~\land ~\genericCoin \notin \invalidList, \\ 
\pcln \t \coinStore \assign \coinStore \cup \langle \genericCoin, \coinID, \voucherDblSpendToken, \txPRFEval, \coinRand, \txRand \rangle; ~\pclinecomment{Update account of $\genDepRSU$}  \\
\pcln \t \pcreturn 0\\
\pclinecomment{Coin has been double-spent with same voucher}\\
\pcln\label{dbl:vehicle_double_spend} \text{If} ~\langle \genericCoin, \coinID, \voucherDblSpendToken, \txPRFEval', \coinRand, \txRandAlt \rangle \in \coinStore\\
    \pcln\label{} \t \text{If} ~\txRand = \txRandAlt, \text{then} ~\abort\\
    \pcln\label{dbl:computeTxPRFEval} \t \dyPRF[\userPRFKey](\coinID) \assign \left(\frac{\txPRFEval}{\txPRFEval'}\right)^{1/(\txRand - \txRandAlt)} \\
    \pcln\label{} \t \vPubKey{\vID} \assign \txPRFEval/\dyPRF[\userPRFKey](\coinID)^{\txRand}\\
    \pcln\label{} \t \vPubKey{\vIDAlt} \assign \txPRFEval'/\dyPRF[\userPRFKey](\coinID)^{\txRandAlt}\\
\pcln \t \text{If} ~\vPubKey{\vIDAlt} \neq \vPubKey{\vID}, \text{then} ~\abort \\
\pcln\label{dbl:vehicle_double_spend_output} \t \pcreturn \pubKey{\vID} ~\qquad\pclinecomment{coin has been double-spent by \genCar}\\
\pclinecomment{Coin has been double-spent with different vouchers or double-issued}\\
\pcln \label{dbl:atm_double_issue} \text{If} ~\langle \genericCoin, \coinID', \voucherDblSpendToken', \txPRFEval', \coinRandAlt, \txRandAlt \rangle \in \coinStore ~\land~ \voucherDblSpendToken \neq \voucherDblSpendToken' \\
\pcln\label{} \t \text{If} ~\coinRand = \coinRandAlt, \text{then follow \stepsref{dbl:vehicle_double_spend}{dbl:vehicle_double_spend_output}}\\
    \pcln\label{dbl:computeVoucherPRFEval} \t \dyPRF[\prfKey[2]](0) \assign \left(\frac{\voucherDblSpendToken}{\voucherDblSpendToken'}\right)^{1/(\coinRand - \coinRandAlt)} \\
    \pcln\label{} \t \rsuPubKey{\rsuID} \assign \voucherDblSpendToken/\dyPRF[\prfKey[2]](0)^{\coinRand}\\
    \pcln\label{} \t \rsuPubKey{\rsuIDAlt} \assign \voucherDblSpendToken'/\dyPRF[\prfKey[2]](0)^{\coinRandAlt}\\
\pcln \t \text{If} ~\rsuPubKey{\rsuIDAlt} \neq \rsuPubKey{\rsuID}, \text{then} ~\abort \\
\pcln\label{dbl:atm_double_issue_output} \t \pcreturn \rsuPubKey{\rsuID} ~\qquad\pclinecomment{coin has been double-issued by \genRSU} 
}
}}
\caption{\noncompactProtocol: Verifying transactions \label{fig:dbl}}
\vspace{-0.3cm}
\end{figure}
Upon being reported a transaction, the bank invokes $\bank.\serverVerifyTx(\genericCoin, \genericVoucher, \genericTx, 
\genCar{\vIDDep})$ and verifies that the coin is valid, the voucher is linked to the coin, and the transaction is valid. If the transaction is valid, the bank checks if the coin is already recorded in the log \coinStore, in which case it invokes the double spending detection algorithm. Specifically, the bank checks for two cases: i) the user has double-spent the same coin, or ii) an ATM has double-issued the coin. The process for the former is nearly identical to the compact \ecash design, and computes $\dyPRF[\userPRFKey](\coinID)$ from the double spending tokens $\txPRFEval$ and $\txPRFEval'$ in the transactions that spend the same coin (\stepsref{dbl:vehicle_double_spend}{dbl:vehicle_double_spend_output}). Then, the bank is able to obtain the public key of the offending user. 
The latter also follows the same process but uses the double-issuing tokens $\voucherDblSpendToken$ and $\voucherDblSpendToken'$ for obtaining the public key of the offending ATM. We show in \appref{appx:proof_dbl} that the algorithm fails if an adversary produces an incorrectly computed voucher/transaction that the bank verifies as valid, which happens with negligible probability.


\section{Description of Protocol with Compactness}
\label{appx_compact_protocol}

\figref{fig:withdrawal_compact} describes the \ecash protocol with compactness. The main differences are that the an ATM can reuse the same PRF key pair $\prfKey[1], \prfKey[2]$ at most $\numkeys$ times. For this, the bank blindly signs these keys with a scheme that supports proof of knowledge of a signature. When the ATM issues a coin to a user it first commits to the keys, and proves that it knows a signature on the openings of the commitments (\stepsref{compact:commit_key1}{compact:prove_key2}). To show that the key hasn't been used more than $\numkeys$ times, the ATM uses the value of a counter \coinSN as a serial number, and evaluates the PRF on the serial number (\stepsref{compact:computeID}{compact:computeVoucherDblToken}). The ATM also proves in zero knowledge that the value of the counter does not exceed $\numkeys$ (\stepref{compact:zkp_range}). Note that the value of the counter is not provided in the clear. 

We remark here the notable differences for the non-compact protocol described in \figref{fig:withdrawal}. First, the two PRF evaluations \voucherPRFEval and \voucherDblSpendToken are on the \emph{same} value, namely the serial number \coinSN. This is because we need to ensure that an adversary cannot reuse the values of \voucherPRFEval and \voucherDblSpendToken as part of different coins. Thus, evaluating the PRF on \coinSN, whose value is not known to the adversary, and tying the evaluations using the NIZK proof \zkp{\coinID} ensures that each pair of evaluations \voucherPRFEval and \voucherDblSpendToken can be used only once and as part of the same coin. This is not necessary in the non-compact version because the commitments to the keys used for evaluating \voucherPRFEval and \voucherDblSpendToken are tied together in a coin using the hash function $\crHash(\keyCommitment[1], \keyCommitment[2], \rsuPrivKeyCommitment)$. Thus, if the adversary provides $\voucherPRFEval$ evaluated under the opening of $\keyCommitment[1]$ in a voucher, then the same voucher must also include $\voucherDblSpendToken$ evaluated under the opening of $\keyCommitment[2]$, unless there is a collision in \crHash. 

Second, in the compact version  $\voucherPRFEval = \dyPRF[\prfKey[1]](\coinSN)$, \emph{independent of the user identity}. In the non-compact version (\figref{fig:withdrawal}), the evaluation $\voucherPRFEval = \dyPRF[\prfKey[1]](\privKeyCommitment)$ is on the commitment to the user's private key; this ensures that coins cannot be transferred between users without sharing keys, or an adversary computing $\voucherPRFEval = \dyPRF[\prfKey[1]](\privKeyCommitment')$ without knowing $\prfKey[1]$. In the compact version, to transfer a coin + voucher, the double spending token  $\voucherDblSpendToken = \dyPRF[\prfKey[2]](\coinSN)^{\coinRand}$ would need to be recomputed by an adversary, which also entails computing the NIZK proof $\zkp{\coinID}$ without the secrets $\prfKey[1], \prfKey[2], \rsuPrivKeyCommitment, \coinSN$. Also note that, similar to the non-compact version, the user still needs to tie the coin to its identity by using the token $\coinID \assign \crHash(\privKeyCommitment, \coinSNCommitment)$ for using the PRF, and proving that its knows the opening of $\privKeyCommitment$ in $\zkp{\genericTx}$.

\begin{figure}[h!]
\resizebox{0.75\textwidth}{!}{\vbox{%
    \pcb{
    \textbf{\genCar.\carGetCoin} \hspace{7cm} \textbf{\genRSU.\rsuAssignCoin} \\
    (\pubKey{\vID}, \privKey{\vID}, \userCLSig{\privKey{\vID}}) \hspace{7cm} (\sPubKeyBlind, \sPubKeyBlindZKP, \genericSig[\rsuPrivKey{\rsuID}], \prfKey[1], \prfKey[2], \genericSig[\prfKey[1]], \genericSig[\prfKey[2]]) \\
    \\[0.05\baselineskip] \\ [-0.3\baselineskip]
    \pcln\label{compact:getcoin:comm} \privKeyCommitment \assign \commFn(\vPrivKey{\vID}, \userPRFKey)\\
    \pcln\label{compact:getcoin:zkp}\zkp{\privKey{\vID}} \assign \zprove{\sPubKeyBlindZKP}(\privKeyCommitment, \userCLSig{\privKey{\vID}}) \\
    \hspace{2cm} \sendmessageright*[2cm]{\vPubKey{\vID}, \privKeyCommitment, \zkp{\privKey{\vID}}} \\
    \hspace{4cm} \pclinecomment{Check user balance before \stepsref{compact:vrfyCLSig}{compact:computeZKP}} \\
    \hspace{4cm} \pcln\label{compact:vrfyCLSig} \text{If}~\zvrfy{\sPubKeyBlindZKP}(\privKeyCommitment, \zkp{\privKey{\vID}}) \neq 0, \text{then} ~\abort\\
	\hspace{4cm} \pcln\label{compact:commit_key1} \textcolor{red}{\keyCommitment[1] \assign \commFn(\prfKey[1])}\\
	\hspace{4cm} \pcln\label{compact:prove_key1}\textcolor{red}{\zkp{\keyCommitment[1]} \assign \zprove{\sPrivKeyBlindZKP}(\prfKey[1], \genericSig[\prfKey[1]])} \quad\pclinecomment{prove key $\prfKey[1]$ has been signed by bank}\\
	\hspace{4cm} \pcln\label{compact:commit_key2} \textcolor{red}{\keyCommitment[2] \assign \commFn(\prfKey[2])}\\
	\hspace{4cm} \pcln\label{compact:prove_key2}\textcolor{red}{\zkp{\keyCommitment[2]} \assign \zprove{\sPrivKeyBlindZKP}(\prfKey[2], \genericSig[\prfKey[2]])} \quad\pclinecomment{prove key $\prfKey[2]$ has been signed by bank}\\
	\hspace{4cm} \pcln \rsuPrivKeyCommitment \assign \commFn(\rsuPrivKey{\rsuID})\\
	\hspace{4cm} \pcln \zkp{\rsuPrivKey{\rsuID}} \assign \zprove{\sPubKeyBlindZKPAlt}(\rsuPrivKeyCommitment, \genericSig[\rsuPrivKey{\rsuID}]) \quad\pclinecomment{prove ATM's private key has been signed by bank}\\
	\hspace{4cm}\pcln \coinRand \assign \crHash(\privKeyCommitment) \\
    \hspace{4cm}\pcln\label{compact:computeID} \coinSN \assign \coinSN + 1 \quad\pclinecomment{track how many coins have been issued by ATM} \\
	\hspace{4cm} \pcln\label{compact:comm_coinSN} \coinSNCommitment \assign \commFn(\coinSN) \\
 \hspace{4cm}\pcln\label{compact:computeVoucherPRF} \voucherPRFEval \assign \dyPRF[\prfKey[1]](\textcolor{red}{\coinSN}) ~\quad\pclinecomment{Unique serial number for the coin to track double spending} \\
\hspace{4cm}\pcln\label{compact:computeVoucherDblToken} \voucherDblSpendToken \assign \rsuPubKey{\rsuID} \cdot \dyPRF[\prfKey[2]](\textcolor{red}{\coinSN})^{\coinRand} \quad\pclinecomment{double issuing token} \\
\hspace{4cm}\pclinecomment{range proof that $\le \numkeys$ coins issued}\\
\hspace{4cm} \pcln\label{compact:zkp_range} \textcolor{red}{\zkp{\coinSN} \assign \zkpGen[\alpha, \beta](\coinSNCommitment)\{ \coinSNCommitment = \commFn(\alpha; \beta) \land \alpha \in [1, \numkeys] \}}  \\ 
\hspace{4cm} \pclinecomment{Tie the secrets in coin with the voucher} \\
    \hspace{4cm}\pcln\label{compact:computeZKP} \zkp{\coinID} \assign \zkpGen[\alpha_1, \beta_1, \alpha_2, \beta_2, \alpha_3, \beta_3, \alpha_4, \beta_4](\keyCommitment[1], \keyCommitment[2],\rsuPrivKeyCommitment, \coinSNCommitment, \voucherPRFEval, \voucherDblSpendToken,\coinRand) \\
    \hspace{4cm}~\quad\{\keyCommitment[1] = \commFn(\alpha_1; \beta_1) ~\land~ \keyCommitment[2] = \commFn(\alpha_2; \beta_2) ~\land~\rsuPrivKeyCommitment = \commFn(\alpha_3, \beta_3) \\ 
\hspace{4cm}~\quad \t \coinSNCommitment = \commFn(\alpha_4; \beta_4) ~\land~ \voucherPRFEval  = \dyPRF[\alpha_1](\alpha_4) 
~\land \voucherDblSpendToken = \grpGen^{\gamma}\cdot\dyPRF[\alpha_2](\alpha_4)^{\coinRand} \} \\
\hspace{4cm} \pcln \genericCoin \assign \langle \keyCommitment[1], \keyCommitment[2], \rsuPrivKeyCommitment,  \zkp{\prfKey[1]}, \zkp{\prfKey[2]}, \zkp{\rsuPrivKey{\rsuID}}, \rangle \\
\hspace{4cm} \pcln \genericVoucher \assign \langle \coinSNCommitment, \privKeyCommitment, \voucherPRFEval, \voucherDblSpendToken, \zkp{\coinSN},\zkp{\coinID} \rangle  \\
    \hspace{2.5cm}\sendmessageleft*[2cm]{\genericCoin, \genericVoucher} \\
    \pcln\text{store} ~\genericCoin, \genericVoucher \\
    \pcln\userReceipt \assign \genCar.\carIssueReceipt \\
    \hspace{2.5cm}\sendmessageright*[2.5cm]{\userReceipt}\\
    \hspace{7cm} \pcln\server.\serverUpdateWithdrawal(\vPubKey{\vID}, \userReceipt) 
    }
}}
\caption{\compactProtocol: Withdrawing a coin\label{fig:withdrawal_compact}}
\end{figure}

\begin{figure}[h!]
\centering
\resizebox{0.8\textwidth}{!}{\vbox{%
    \pcb{
  \textbf{\genCar.\carDepositCoin} \hspace{7cm} \textbf{\genDepRSU.\rsuVerifyCoin}\\
(\pubKey{\vID}, \privKey{\vID}, \userCLSig{\privKey{\vID}}, \userPRFKey, \genericCoin, \genericVoucher) \hspace{6.5cm} (\sPubKeyBlind, \sPubKeyBlindZKP, \sPubKeyBlindZKPAlt) 
       \\[0.05\baselineskip] \\ [-0.3\baselineskip]
    \hspace{9cm}\pcln\label{} \coinRandVehicle \sample \bin^{\secParam} \\
\hspace{3cm} \sendmessageleft*[3cm]{\coinRandVehicle}\\
    \pcln\label{} \txRand \assign \crHash(\rsuPubKey{\vIDDep}, \coinRandVehicle) \\
	\pcln\label{} \coinID \assign \crHash(\privKeyCommitment, \coinSNCommitment) \\
    \pcln \txPRFEval \assign \grpGen^{\privKey{\vID}} \dyPRF[\userPRFKey](\coinID)^{\txRand} \quad\pclinecomment{double spending token} \\
    \pcln \zkp{\genericTx} \assign \zkpGen[\alpha, \beta, \gamma](\privKeyCommitment, \txPRFEval, \coinID, \txRand) \\
  \hspace{2cm} \t \{ \commFn(\alpha, \gamma; \beta) = \privKeyCommitment ~\land \\ 
   \hspace{2cm} \t \t  \txPRFEval = \grpGen^{\alpha} \dyPRF[\gamma](\coinID)^{\txRand} \}\\
  \pcln  \genericTx \assign \langle \txPRFEval, \zkp{\genericTx}, \coinRandVehicle, \txRand \rangle \\
    \hspace{2cm} \sendmessageright*[3cm]{\genericCoin, \genericVoucher, \genericTx} \\
\hspace{5cm}\pclinecomment{verify random elements} \\
\hspace{5cm}\pcln \coinRand =\crHash(\privKeyCommitment)\\
\hspace{5cm}\pcln  \text{If} ~\txRand \neq \crHash(\rsuPubKey{\genDepRSUID}, 	\coinRandVehicle), \text{then} ~\abort \\
\hspace{5cm}\pcln  \text{If} ~\coinID \neq \crHash(\privKeyCommitment,\coinSNCommitment), \text{then} ~\abort \\
\hspace{5cm}\pclinecomment{verify identity of user $\genCar$} \\
\hspace{5cm} \pcln \text{If} ~\zvrfy{\sPubKeyBlindZKP}(\privKeyCommitment, \zkp{\privKey{\vID}}) \neq 0, \text{then} ~\abort\\
\hspace{5cm}\pcln \text{If} ~\zvrfy{\sPubKeyBlindZKP}(\userPRFKeyCommitment, \zkp{\userPRFKey}) \neq 0, \text{then} ~\abort \\
\hspace{5cm}\pclinecomment{verify identity of coin-issuing ATM $\genRSU$} \\
\hspace{5cm}\pcln \text{If} ~\zvrfy{\sPubKeyBlindZKPAlt}(\rsuPrivKeyCommitment, \zkp{\rsuPrivKey{\rsuID}}) \neq 0, \text{then} ~\abort\\
\hspace{5cm}\pclinecomment{verify coin is valid} \\
\hspace{5cm}\pcln \text{If} ~\zvrfy{\sPubKeyBlind}(\keyCommitment[1], \zkp{\prfKey[1]}) \neq 1, \text{then} ~\abort\\
\hspace{5cm}\pcln \text{If} ~\zvrfy{\sPubKeyBlind}(\keyCommitment[2], \zkp{\prfKey[2]}) \neq 1, \text{then} ~\abort\\
\hspace{5cm}\pcln \text{If} ~\zkpVerify(\zkp{\coinSN}, \coinSNCommitment) \neq 1, \text{then} ~\abort\\
\hspace{5cm}\text{If} ~\bvrfy{\sPubKeyBlind}(\crHash(\keyCommitment[1], \keyCommitment[2], \rsuPrivKeyCommitment), \signedCoin) \neq 0, \text{then} ~\abort  \\
\hspace{5cm}\pclinecomment{verify coin is linked to voucher for $\genCar$ issued by $\genRSU$ } \\
\hspace{5cm}\pcln \text{If} ~ \zkpVerify(\zkp{\coinID}, \keyCommitment[1], \keyCommitment[2], \rsuPrivKeyCommitment, \coinSNCommitment, \voucherPRFEval, \voucherDblSpendToken, \coinRand) \neq 0,\\
\hspace{5cm}\hspace{1cm} \t \t \text{then} ~\abort 	 \\
\hspace{5cm}\pclinecomment{verify coin + voucher is linked to the transaction by $\genCar$} \\
\hspace{5cm}\pcln \text{If} ~ \zkpVerify(\zkp{\genericTx}, \privKeyCommitment, \txPRFEval, \coinID, \txRand) \neq 0, \\
\hspace{5cm}\hspace{1cm} \t \t \text{then} ~\abort \\
\hspace{5cm}  \pcln \server.\serverVerifyTx(\genericCoin, \genericVoucher, \genericTx, \genDepRSU)
    }
}}
\caption{\compactProtocol: Spending a coin\label{fig:spending_compact}}
\end{figure}

\section{Fully Offline Withdrawals}
\label{sec:offline_withdraw}

We have assumed in our discussions that the \rsus will always check
(and update) balances from the bank before dispensing coins to a user,
although the coins that are dispensed are withdrawn offline.  We can allow an \rsu to defer checking a user's account balance until later, i.e., off the critical path of dispensing coins to that user,
provided that (i) the user had a sufficiently large balance when the
\rsu last checked, and (ii) users are globally rate-limited in their
abilities to withdraw coins.  

The first condition can be fulfilled if the bank 
broadcasts a set of accounts that do not have the minimum withdrawal 
balance. To make this list concise, the bank will represent this set as a 
Bloom filter. When a user requests cash, the \rsu will first check whether
the user's account number is in the Bloom filter;  if not, it 
will dispense the required cash. The Bloom filter can be tuned for negligible 
false positives~\cite{larisch2017crlite}.

For the second condition, we need to prevent a user from withdrawing 
simultaneously from 
multiple ATMs.  If the ATMs are edge devices and users are 
required to present some unforgeable credentials to access them, then 
implicitly the user may not be able to visit multiple \rsus to overdraw their 
account, unless they are able to forge their credentials. In such cases, standard authentication 
systems that measure physical property, like biometrics can be used to ensure that a user is 
physically present during withdrawals. We will ensure that the time period between two subsequent 
updates from an ATM must be bounded so that a user cannot visit a large number of ATMs 
in that time period.

Alternatively, if the ATMs are digital entities that users connect to, then  
globally rate-limiting a user's coin withdrawals might be accomplished
using a public random beacon (e.g., as implemented by
NIST~\cite{kelsey2019reference}).  If each user holds a private key
for a verifiable random function family
(VRF)~\cite{micali1999verifiable}, for which the bank (and its \rsus)
 holds the corresponding public
key, then an \rsu can permit coin withdrawal by this user only if the
user's VRF, applied to the current beacon value, designates this \rsu
as the one at which the user is currently eligible to withdraw.  In
this way, the user will be rate-limited to withdraw at only one \rsu
per beacon pulse and, then, only the number of coins that \rsu allows.

\section{Microbenchmark}
\label{sec:eval_compute}

\begin{figure}[t]
    \centering
    \includegraphics[width=0.9\textwidth]{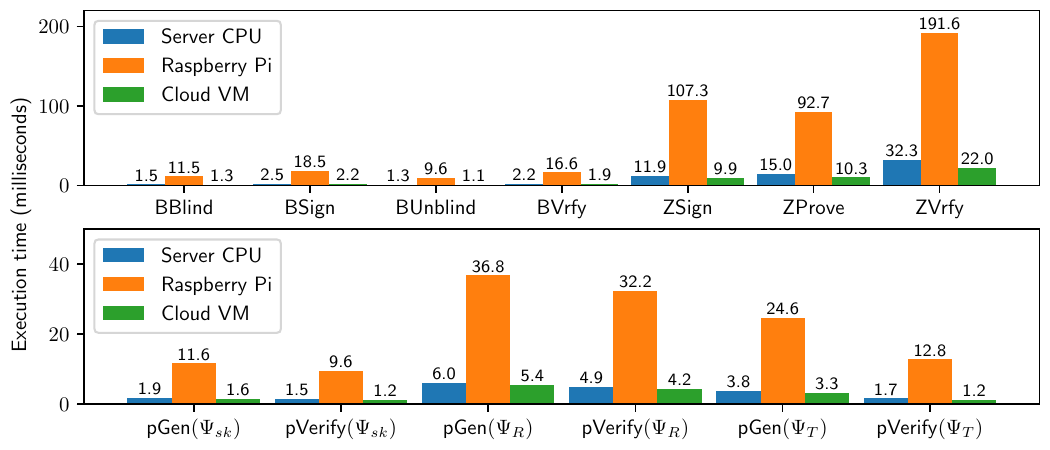}
    \caption{\small Computation overhead of cryptographic functions.}
    \label{fig:computation_crypto}
\end{figure}

We implemented our \ecash scheme with blind signatures on a 2048-bit RSA group, blind signature with proofs of knowledge on the elliptic curve BLS12\_381 (CL signature~\cite{camenisch2004signature}), and non-interactive zero-knowledge proofs of knowledge on a 2048-bit RSA groups.
BLS12\_381 is a pairing-friendly elliptic curve providing 128-bit security, widely adopted in blockchain applications such as ZCash~\cite{Zcash2024} and Ethereum~\cite{Ethereum2024}.
The identities of users and ATMs are managed using a 2048-bit RSA public key system. 
All cryptographic functions utilized are standardized and recognized for their security.

In terms of implementation details, the scheme is implemented in 1,694 lines of C++. The CL signature~\cite{camenisch2004signature} and pseudorandom functions (PRFs)~\cite{dodis2005verifiable} are implemented using the MIRACL library~\cite{zhang2022making}, the NIZK is implemented on ZKPDL~\cite{meiklejohn2010zkpdl}, while the other cryptographic functions are implemented on OpenSSL~\cite{openssl}.

We benchmarked the performance of \sys on a single core of three CPU platforms:
\begin{itemize}[nosep,leftmargin=1em,labelwidth=*,align=left]
    \item Intel Xeon Silver 4110 8-core CPU @ 2.10GHz of around $\$1{,}400$ representing the typical computation power on servers.
    \item Raspberry Pi 4 with quad-core Cortex-A72 64-bit SoC @ 1.5GHz costing $\$40$, as the most affordable choice.
    \item Amazon AWS virtual machine (t2.micro, $\$0.28$ per day) for benchmarking the performance on a low-resource cloud system.
\end{itemize}

\myparagraph{Cryptographic functions}
\figref{fig:computation_crypto} illustrates the execution time of cryptographic functions. Across various hardware platforms, the server CPU and cloud VM exhibit similar performance, whereas the Raspberry Pi is approximately 8$\times$ slower. The most expensive function is the verification of CL signature, while it can be finished in 200 ms on Raspberry Pi 4. Overall, the benchmarking results demonstrate that our \ecash scheme is well-suited for real-time applications, even on resource-constrained hardware.

\begin{table}[t]
\small
\setlength{\tabcolsep}{4pt}
\caption{Size and execution time of each protocol step in the ATM \ecash scheme. 
}
  \label{tab:computation_atm}
  \centering

{\centering \textit{(a) Protocol object sizes}\par}


\vspace{2px}
\begin{tabular}{|l|l|r|}
\noalign{\global\arrayrulewidth1pt}\hline\noalign{\global\arrayrulewidth0.4pt}
Object & \kilobytes \\
\noalign{\global\arrayrulewidth1pt}\hline\noalign{\global\arrayrulewidth0.4pt}
Coin \genericCoin                                  & 1  \\
Voucher \genericVoucher    & 54 \\
Transaction \genericTx                                  & 2  \\
\noalign{\global\arrayrulewidth1pt}\hline\noalign{\global\arrayrulewidth0.4pt}
\end{tabular}

\vspace{6px}
\centering
{\centering \textit{(b) Communication cost per protocol step}\par}
\vspace{2px}
\begin{tabular}{|l|l|}
\noalign{\global\arrayrulewidth1pt}\hline\noalign{\global\arrayrulewidth0.4pt}
Messgage direction & \kilobytes \\
\noalign{\global\arrayrulewidth1pt}\hline\noalign{\global\arrayrulewidth0.4pt}
\multicolumn{2}{|l|}{\textit{Coin request from bank} (\rsuWithdrawCoin/\serverIssueCoin)} \\
\hline
ATM $\to$ bank ($\hashedCoin$) & $0.3$ \\
Bank $\to$ ATM ($\signedCoin'$) & $0.3$ \\
\hline
\multicolumn{2}{|l|}{\textit{Coin withdrawal} (\carGetCoin / \rsuAssignCoin)} \\
\hline
User $\to$ \rsu ~~(\vPubKey{\vID}, \privKeyCommitment, \zkp{\privKey{\vID}}) & 25 \\
\rsu $\to$ User ~~(\genericCoin, \genericVoucher)  & 55 \\
\hline
\multicolumn{2}{|l|}{\textit{Coin spending} (\carDepositCoin / \rsuVerifyCoin)} \\
\hline
User $\to$ Merchant ~~(\genericCoin, \genericVoucher, \genericTx) & 57 \\
\noalign{\global\arrayrulewidth1pt}\hline\noalign{\global\arrayrulewidth0.4pt}
\end{tabular}

\vspace{6px}
\centering
{\centering \textit{(c) Execution time (ms)}\par}
\vspace{2px}
\begin{tabular}{|c|r|r|r|r|}
\noalign{\global\arrayrulewidth1pt}\hline\noalign{\global\arrayrulewidth0.4pt}
\multirow{2}{*}{CPU} & \multicolumn{2}{c|}{\carGetCoin / \rsuAssignCoin} & \multicolumn{2}{c|}{\carDepositCoin / \rsuVerifyCoin} \\
\cline{2-5}
& User & \rsu & User & Merchant \\
\noalign{\global\arrayrulewidth1pt}\hline\noalign{\global\arrayrulewidth0.4pt}
Server CPU     & 15.52 & 58.57  & 4.33  & 67.77  \\
Raspberry Pi   & 94.56 & 414.13 & 26.46 & 397.62 \\
Cloud VM       & 10.80 & 46.86  & 3.80  & 64.55  \\
\noalign{\global\arrayrulewidth1pt}\hline\noalign{\global\arrayrulewidth0.4pt}
\end{tabular}
\vspace{2px}
{\\ \centering \scriptsize Execution time in milliseconds.}
\end{table}

\myparagraph{Object sizes \& communication}
\tblref{tab:computation_atm} presents the concrete sizes of the protocol objects alongside the execution time of each protocol step across the three hardware platforms for the non-compact version of the protocol. 
The coin \genericCoin consists of three Pedersen commitments and one blind signature, giving a size of 1\kilobytes.
The voucher \genericVoucher is dominated by two CL signature proofs ($\zkp{\privKey{\vID}}$ and $\zkp{\rsuPrivKey{\rsuID}}$), which together account for $\approx$93\% of its 55\kilobytes size.
The transaction \genericTx is roughly 2\kilobytes, containing the double-spending token, a NIZK proof with 3 witnesses, and two random elements.
Importantly, the coin size is \emph{constant} regardless of the number of users, ATMs, or prior transactions --- only the voucher carries the zero-knowledge proofs needed for verification.

We note that 
introducing compactness increases the latency for withdrawal since the voucher now includes an additional zero-knowledge range proof. The spending time also slightly increases due to the additional proof verification time. While the compact scheme can be used when storage resources are limited, we believe that for most deployment scenarios, ATMs will have enough storage to stock coins without compactness, due to the small size of the coins in our protocol. Without compactness, the ATMs will also need more interaction with the bank to fetch coins (invoking \rsuWithdrawCoin / \serverIssueCoin), but the total communication per coin is $\approx$0.6\kilobytes. Thus, fetching around 100,000 coins in a batch requires roughly 66\megabytes of total communication, and storing them requires roughly 100MB of storage.

\myparagraph{End-to-end compute latency}
Coin withdrawal from an ATM (which invokes $\carGetCoin$ and  $\rsuAssignCoin$) is the most expensive online step for the user, as its proves in zero-knowledge the possession of a valid CL signature on its private identity key. During the withdrawal, the ATM computes the voucher for the user. On the server CPU, user-side computation takes 15.52\,ms and ATM-side computation takes 58.57\,ms. On the Raspberry Pi, these increase to 94.56\,ms and 414.13\,ms respectively, approximately 7$\times$ slower, consistent with the per-function overhead observed in \figref{fig:computation_crypto}.

Spending a coin (\carDepositCoin) is lightweight on the user side: 4.33\,ms on the server CPU and 26.46\,ms even on the Raspberry Pi, as it only requires computing the double-spending token and the transaction proof $\zkp{\genericTx}$.
Verification at the merchant (\rsuVerifyCoin) takes 67.77\,ms on the server CPU and 397.62\,ms on the Raspberry Pi, dominated by the CL signature verifications of user and ATM credentials, the blind signature check (\bvrfy{}) on the coin, and the NIZK proof verifications. The cloud VM performs comparably to the server CPU.
Overall, the total latency for the online withdrawal path (user + ATM) is 74.09\,ms on the server CPU and 508.69\,ms on the Raspberry Pi, while the spending path (user + merchant) completes in 72.10\,ms and 424.08\,ms respectively. These results confirm that our ATM \ecash protocol is practical for real-time applications across all tested hardware platforms.



\mysection{Conclusion}
We presented a new cryptographic bearer token design for offline \ecash systems, enabling anonymous, unforgeable, and untraceable withdrawals from decentralized ATMs. Our construction leverages an unforgeable and \emph{doubly-anonymous} voucher that allows a one-time transfer of coins between an ATM and a user, while hiding their identities from parties not involved in the transaction.

\section{Acknowledgements}

This work was supported in part through NSF award 2112562, 2425891 and 2425892.We thank our shepherd and the anonymous ACNS reviewers for their excellent suggestions and feedback. The views and conclusions in this document are those of
the authors and should not be interpreted as representing the
official policies, either expressed or implied, of the National
Science Foundation.


\bibliographystyle{IEEEtranSN}
{\footnotesize
\bibliography{bib.bib}
}

\appendix

\iffullversion
\section{Cryptographic definitions}
\label{app:defn}

\subsection{Discrete log problem}
Given a multiplicative group \grp of order \grpOrd and group generator \grpGen, the discrete log 
problem can be formally defined with the following experiment

\subsection{IND-CPA security}
The security assumption for encryption to which we reduce is IND-CPA
security~\cite{bellare1998relations}, defined as follows.

 Consider
the following experiment, parameterized by a bit $\indcpaBit \in
\{0,1\}$, where \indcpaAdversary is an adversary algorithm with
knowledge of the encryption scheme $\langle \keyGen, \encrypt,
\decrypt\rangle$.
\begin{center}
	\begin{minipage}{0.5\columnwidth}
		\begin{tabbing}
			*** \= \kill
			Experiment $\Experiment{\indcpaLabel[\indcpaBit]}{\langle \keyGen, \encrypt, 
			\decrypt\rangle}(\indcpaAdversary)$ \\
			\> $\symmKey \gets \keyGen()$ \\
			\> $\indcpaGuess \gets 
			\indcpaAdversary^{\indcpaOracle{\symmKey}{\cdot}{\cdot}{\indcpaBit}}()$ \\
			\> return \indcpaGuess
		\end{tabbing}
	\end{minipage}
\end{center}

Here, \indcpaOracle is an oracle to which the bit \indcpaBit and the
key \symmKey are fixed inputs.  In each invocation to it, the
adversary \indcpaAdversary provides two inputs $\plaintext[0],
\plaintext[1] \in \plaintextSpace$ of the same length, to which the
oracle invocation
\indcpaOracle{\symmKey}{\plaintext[0]}{\plaintext[1]}{\indcpaBit}
returns $\encrypt{\symmKey}(\plaintext[\indcpaBit])$.  The adversary's
advantage is defined as

\begin{align*}
	\Advantage{\indcpaLabel}{\langle \keyGen, \encrypt, \decrypt\rangle}(\indcpaAdversary)
	& = \prob{\Experiment{\indcpaLabel[1]}{\langle \keyGen, \encrypt, 
	\decrypt\rangle}(\indcpaAdversary) = 1} \\
	& \qquad -~\prob{\Experiment{\indcpaLabel[0]}{\langle \keyGen, \encrypt, 
	\decrypt\rangle}(\indcpaAdversary) = 1} \\
	\Advantage{\indcpaLabel}{\langle \keyGen, \encrypt, \decrypt\rangle}(\timeBound, \oracleQueries) 
	&
	= \max_{\indcpaAdversary} \Advantage{\indcpaLabel}{\langle \keyGen, \encrypt, 
	\decrypt\rangle}(\indcpaAdversary)
\end{align*}

\noindent
where the maximum is taken over all adversaries \indcpaAdversary running
in time at most \timeBound and making at most \oracleQueries queries. 

\subsection{Pseudorandom functions}
We recall the
standard security definition of a PRF below

\begin{defn}[Pseudorandom function]
	\label{dfn:prp}
	Let $\genericFnFamily: \genericFnKeyspace \times \genericFnDomain
	\rightarrow \genericFnDomain$ be a function family,
	$\allFuncs{\genericFnDomain}$ be the set of all functions with domain
	and range \genericFnDomain, and \prfAdversary be an algorithm that
	takes an oracle of type $\genericFnDomain \rightarrow
	\genericFnDomain$ and returns a bit.  Let
	
	\smallskip
	\begin{minipage}[t]{0.4\columnwidth}
		\begin{tabbing}
			***\=***\=\kill
			Experiment $\Experiment{\prfLabel[0]}{\genericFnFamily}(\prfAdversary)$ \\
			\> $\genericFn \getsr \allFuncs{\genericFnDomain}$ \\
			\> $\prfGuess \leftarrow \prfAdversary^{\genericFn(\cdot)}()$ \\
			\> return $\prfGuess$
		\end{tabbing}
	\end{minipage}
	\hfill
	\begin{minipage}[t]{0.4\columnwidth}
		\begin{tabbing}
			***\=***\=\kill
			Experiment $\Experiment{\prfLabel[1]}{\genericFnFamily}(\prfAdversary)$ \\
			\> $\genericFnKey \getsr \genericFnKeyspace$, $\genericFn \gets 
			\genericFnFamily[\genericFnKey]$ \\
			\> $\prfGuess \leftarrow \prfAdversary^{\genericFn(\cdot)}()$ \\
			\> return $\prfGuess$
		\end{tabbing}
	\end{minipage}
	
	\smallskip
	\noindent Then,
	
	\begin{align*}
		\Advantage{\prfLabel}{\genericFnFamily}(\prfAdversary)
		& = \left|\prob{\Experiment{\prfLabel[1]}{\genericFnFamily}(\prfAdversary) = 1} - 
		\prob{\Experiment{\prfLabel[0]}{\genericFnFamily}(\prfAdversary) = 1} \right| \\
		\Advantage{\prfLabel}{\genericFnFamily}(\timeBound, \oracleQueries)
		& = \max_{\prfAdversary} \Advantage{\prfLabel}{\genericFnFamily}(\prfAdversary)
	\end{align*}
	
	\noindent
	where the maximum is taken over all adversaries \prfAdversary that run in
	time at most \timeBound and make at most \forwardOracleQueries queries
	to \genericFn.
\end{defn}

\subsection{Commitment scheme}
We will require
for security that the commitment scheme provides computational hiding and binding
properties. These properties can be defined by the experiments below.

\myparagraph{Hiding}
Intuitively, the property implies that a commitment to a message
\plaintext hides the message against an adversary.

\begin{center}
	\begin{minipage}{0.5\columnwidth}
		\begin{tabbing}
			*** \= \kill
			Experiment $\Experiment{\hidingLabel[\hidingBit]}{\commFn}(\hidingAdversary)$ \\
			\> $\hidingGuess \gets \hidingAdversary^{\commFn(\cdot, \cdot, \hidingBit)}()$ \\
			\> return \hidingGuess
		\end{tabbing}
	\end{minipage}
\end{center}

Here, the adversary is given oracle access to $\commFn$ to which the bit \hidingBit is a
fixed input.  In each invocation to it, the adversary \hidingAdversary
provides two inputs $\plaintext[0], \plaintext[1] \in \plaintextSpace$
of the same length, to which the oracle invocation
\commFn(\plaintext[0], \plaintext[1], \hidingBit) returns
$\commFn(\plaintext[\hidingBit])$.  We will define the adversary's
advantage as

\begin{align*}
	\Advantage{\hidingLabel}{\commFn}(\hidingAdversary)
	& = \left| \prob{\Experiment{\hidingLabel[0]}{\commFn}(\hidingAdversary) = 1} \right.\\
	& \left. \qquad -~\prob{\Experiment{\hidingLabel[1]}{\commFn}(\hidingAdversary) = 1} \right| \\
	\Advantage{\hidingLabel}{\commFn}(\timeBound, \oracleQueries) &
	= \max_{\hidingAdversary} \Advantage{\hidingLabel}{\commFn}(\hidingAdversary)
\end{align*}

\noindent
where the maximum is taken over all adversaries \hidingAdversary
running in time at most \timeBound and making at most \oracleQueries
queries to its oracle.

\myparagraph{Binding} Intuitively, this property implies that a
commitment uniquely identifies the message for which it is a
commitment.  In other words, an adversary cannot produce two messages
that produce the same commitment.

\begin{center}
	\begin{minipage}{0.5\columnwidth}
		\begin{tabbing}
			*** \= *** \= \kill
			Experiment $\Experiment{\bindingLabel}{\commFn}(\bindingAdversary)$ \\
			\> $\langle \plaintext[0], \plaintext[1], \genericRand{0}, \genericRand{1} \rangle \gets 
			\bindingAdversary^{\commFn(\cdot)}()$ \\
			\> If $\commFn(\plaintext[0]; \genericRand{0}) = \commFn(\plaintext[1]; \genericRand{1})$ \\
			\> \> then output 1 \\
			\> \> else output 0
		\end{tabbing}
	\end{minipage}
\end{center}

We define the advantage of the adversary as 

\begin{align*}
	\Advantage{\bindingLabel}{\commFn}(\bindingAdversary)
	& = \prob{\Experiment{\bindingLabel}{\commFn}(\bindingAdversary) = 1} \\
	\Advantage{\bindingLabel}{\commFn}(\timeBound, \oracleQueries) 
	& = \max_{\bindingAdversary} \Advantage{\bindingLabel}{\commFn}(\bindingAdversary)
\end{align*}

\noindent
where the maximum is taken over all adversaries \bindingAdversary running in
time at most \timeBound.

\subsection{Digital signatures}

For security we require unforgeability under a chosen message
attack for the digital signature scheme. We recall the definition here.

\begin{center}
	\begin{minipage}{0.95\columnwidth}
		\begin{tabbing}
			*** \= *** \= \kill
			Experiment $\Experiment{\ufcmaLabel}{\sigScheme}(\ufcmaAdversary)$ \\
			\> $\langle \pubKey, \privKey \rangle \gets \signKeyGen()$ \\
			\> $\langle \msg, \genericSig \rangle \gets 
			\ufcmaAdversary^{\sign{\privKey}(\cdot)}(\pubKey)$  \\
			\> if $\vrfy{\pubKey}(\msg, \genericSig) = 1 \land \mbox{\msg was not queried to 
			$\sign{\privKey}(\cdot)$}$ \\
			\> \> then output 1 \\
			\> \> else output 0 
		\end{tabbing}
	\end{minipage}
\end{center}

We define

\begin{align*}
	\Advantage{\ufcmaLabel}{\sigScheme}(\ufcmaAdversary) & = 
	\prob{\Experiment{\ufcmaLabel}{\sigScheme}(\ufcmaAdversary) = 1} \\
	\Advantage{\ufcmaLabel}{\sigScheme}(\timeBound, \oracleQueries) & = \max_{\ufcmaAdversary} 
	\Advantage{\ufcmaLabel}{\sigScheme}(\ufcmaAdversary) 
\end{align*}

\noindent
where the maximum is taken over all adversary algorithms
\ufcmaAdversary that runs in total time \timeBound and makes
\oracleQueries oracle queries.

\subsection{Non-interactive zero-knowledge proofs of knowledge}
We
reduce the security of our protocol to the following adversary
advantages against \zkpName.  \zkpSim offers two interfaces, denoted
\zkpSimHash and \zkpSimProve, that share state between them.
Similarly, \zkpExtract offers two interfaces, denoted \zkpExtractHash
and \zkpExtractGen, that share state between them.  In addition,
\zkpExtractGen has oracle access to an oracle (which, in turn, can
query \zkpExtractHash) that \zkpExtractGen can reset and have it
re-execute on the same random tape.

\subsubsection{Witness extraction}
\label{sec:background:zkp:soundness}

We adopt the standard soundness definition for noninteractive proofs
of knowledge.  We define a proving adversary \zkpProverAdversary that
participates in the experiment defined below.  \zkpProverAdversary is
any algorithm that can produce proofs that \zkpVerify will accept with
some probability.

\smallskip
\noindent
\begin{minipage}[t]{0.4\columnwidth}
	\begin{tabbing}
		** \= ** \= \kill
		Experiment $\Experiment{\soundnessLabel[0]}{\zkpName}(\zkpProverAdversary, \langle 
		\statement, \witness \rangle)$ \\
		\> $\randomOracle \getsr \randomOracles$ \\
		\> $\zkp \gets \zkpProverAdversary[\witness]^{\randomOracle}(\statement)$ \\
		\> if $\zkpVerify(\statement, \zkp)$ \\
		\> \> then return 1 \\
		\> \> else return 0
	\end{tabbing}
\end{minipage}
\hfill
\begin{minipage}[t]{0.45\columnwidth}
	\begin{tabbing}
		** \= ** \= \kill
		Experiment $\Experiment{\soundnessLabel[1]}{\zkpName}(\zkpProverAdversary, \statement)$ \\
		\> $\witness \gets 
		\zkpExtractGen^{\zkpProverAdversary^{\zkpExtractHash}(\statement)}(\statement)$ \\
		\> if $(\statement, \witness) \in \relation$\\
		\> \> then return 1 \\
		\> \> else return 0
	\end{tabbing}
\end{minipage}

\smallskip
We define

\begin{align*}
	\Advantage{\soundnessLabel}{\zkpName}(\zkpProverAdversary)
	& = \max_{\langle \statement,\witness\rangle \in \relation}
	\left(\begin{array}[m]{@{}r@{}}
		\prob{\Experiment{\soundnessLabel[0]}{\zkpName}(\zkpProverAdversary, \langle \statement, 
		\witness\rangle) = 1} \\
		-~\prob{\Experiment{\soundnessLabel[1]}{\zkpName}(\zkpProverAdversary, \statement) = 1}
	\end{array}\right) \\
	\Advantage{\soundnessLabel}{\zkpName}(\timeBound, \oracleQueries)
	& = \max_{\zkpProverAdversary} \Advantage{\soundnessLabel}{\zkpName}(\zkpProverAdversary)
\end{align*}

\noindent
where the maximum is taken over all adversaries \zkpProverAdversary
running in total time \timeBound and making at most \oracleQueries
oracle queries. 

\subsubsection{Zero knowledge}
\label{sec:background:zkp:zeroknowledge}
We define a distinguishing adversary to be an algorithm
\zkpVerifierAdversary that can participate in either of the
experiments \Experiment{\zkLabel[\zkpVerifierAdversaryBit]}{\zkpName}
defined below:

\smallskip
\noindent
\begin{minipage}[t]{0.5\columnwidth}
	\begin{tabbing}
		\hspace{1em} \= \hspace{1em} \= \kill
		Experiment $\Experiment{\zkLabel[0]}{\zkpName}(\zkpVerifierAdversary)$ \\
		\> $(\statement, \witness) \getsr \relation$ \\
		\> $\randomOracle \getsr \randomOracles$ \\
		\> $\zkp \gets \zkpGen[\witness]^{\randomOracle}(\statement)$ \\
		\> $\zkpVerifierAdversaryBit \gets \zkpVerifierAdversary^{\randomOracle}(\statement, \zkp)$ \\
		\> return \zkpVerifierAdversaryBit
	\end{tabbing}
\end{minipage}
\hfill
\begin{minipage}[t]{0.5\columnwidth}
	\begin{tabbing}
		\hspace{1em} \= \hspace{1em} \= \kill
		Experiment $\Experiment{\zkLabel[1]}{\zkpName}(\zkpVerifierAdversary)$ \\
		\> $(\statement, \witness) \getsr \relation$ \\
		\> $\zkp \gets \zkpSimProve(\statement)$ \\
		\> $\zkpVerifierAdversaryGuess \gets \zkpVerifierAdversary^{\zkpSimHash}(\statement, \zkp)$ \\
		\> return \zkpVerifierAdversaryBit
	\end{tabbing}
\end{minipage}

\smallskip

In words, the adversary \zkpVerifierAdversary must distinguish between
a real proof output from
$\zkpGen[\witness]^{\randomOracle}(\statement)$ and a proof output
from the simulator $\zkpSimProve(\statement)$ without knowledge of the
witness \witness but with the ability to implement the hash function
\zkpSimHash.  This permits \zkpSim to leverage the standard technique
of ``backpatching'' the random oracle outputs on inputs that
\zkpVerifierAdversary has not yet queried.  We define

\begin{align*}
	\Advantage{\zkLabel}{\zkpName}(\zkpVerifierAdversary)
	& = \prob{\Experiment{\zkLabel[1]}{\zkpName}(\zkpVerifierAdversary) = 1} - 
	\prob{\Experiment{\zkLabel[0]}{\zkpName}(\zkpVerifierAdversary) = 1} \\
	\Advantage{\zkLabel}{\zkpName}(\timeBound, \oracleQueries)
	& = \max_{\zkpVerifierAdversary} \Advantage{\zkLabel}{\zkpName}(\zkpVerifierAdversary)
\end{align*}

\noindent
where the maximum is taken over all adversaries \zkpVerifierAdversary
making \oracleQueries random oracle queries and running in time
\timeBound.

\subsection{Blind signatures with proofs of knowledge}

The security property
for blinding is defined in the following way, where \blindSigAdversary
is a signing adversary.

\begin{center}
	\begin{minipage}{0.5\columnwidth}
		\begin{tabbing}
			*** \= *** \kill
			Experiment 
			$\Experiment{\blindSigLabel[\blindSigBit]}{\blindSigScheme}(\blindSigAdversary)$ \\
			\> $\langle \pubKey, \privKey\rangle \gets \bkeygen()$ \\
			\> $\blindSigGuess \gets 
			\blindSigAdversary^{\blindSigOracle{\pubKey}{\cdot}{\cdot}{\blindSigBit}}(\pubKey, 
			\privKey)$ \\
			\> return \blindSigGuess
		\end{tabbing}
	\end{minipage}
\end{center}

Here, \blindSigOracle is an oracle to which the public key \pubKey and
the bit \blindSigBit are fixed inputs.  In each invocation to it, the
adversary \blindSigAdversary provides two inputs $\msg[0], \msg[1]$,
to which the oracle invocation
\blindSigOracle{\pubKey}{\msg[0]}{\msg[1]}{\blindSigBit} returns
$\bblind[\pubKey](\msg[\blindSigBit])$.  We will define the
adversary's advantage as

\begin{align*}
	\Advantage{\blindSigLabel}{\blindSigScheme}(\blindSigAdversary)
	& = \left| \prob{\Experiment{\blindSigLabel[0]}{\blindSigScheme}(\blindSigAdversary) = 1} \right.\\
	& \left. \qquad -~\prob{\Experiment{\blindSigLabel[1]}{\blindSigScheme}(\blindSigAdversary) = 1} 
	\right| \\
	\Advantage{\blindSigLabel}{\blindSigScheme}(\timeBound, \oracleQueries) &
	= \max_{\blindSigAdversary} 
	\Advantage{\blindSigLabel}{\blindSigScheme}(\blindSigAdversary)
\end{align*}

\noindent
where the maximum is taken over all adversaries \blindSigAdversary
running in time at most \timeBound and making at most \oracleQueries
queries to its oracle.

The algorithm
\bvrfy{} simply verifies \zkp, returning 1 if it is valid and 0
otherwise. The security property we require is defined by the following experiment 

\begin{center}
	\begin{minipage}{0.95\columnwidth}
		\begin{tabbing}
			*** \= *** \=*** \= \kill
			Experiment $\Experiment{\ufcmaBlindLabel}{\blindSigScheme}(\ufcmaBlindAdversary)$ \\
			\> $\langle \pubKey, \privKey \rangle \gets \bkeygen()$ \\
			\> $\langle \genericCommitment, \zkp \rangle \gets 
			\ufcmaBlindAdversary^{\bsign{\privKey}(\cdot)}(\pubKey)$  \\
			\> if $\vrfy{\pubKey}(\genericCommitment, \zkp) = 1 \land \genericCommitment = 
			\commScheme(\msg, \dots)$\\
			\>\> $\land ~\mbox{\msg was not queried to $\sign{\privKey}(\cdot)$}$ \\
			\> \> \>then output 1 \\
			\> \>\> else output 0 
		\end{tabbing}
	\end{minipage}
\end{center}

We define

\begin{align*}
	\Advantage{\ufcmaBlindLabel}{\blindSigScheme}(\ufcmaBlindAdversary) & = 
	\prob{\Experiment{\ufcmaBlindLabel}{\blindSigScheme}(\ufcmaBlindAdversary) = 1} \\
	\Advantage{\ufcmaBlindLabel}{\blindSigScheme}(\timeBound, \oracleQueries) & = 
	\max_{\ufcmaBlindAdversary} 
	\Advantage{\ufcmaBlindLabel}{\blindSigScheme}(\ufcmaBlindAdversary) 
\end{align*}

\noindent
where the maximum is taken over all adversary algorithms
\ufcmaAdversary that runs in total time \timeBound and makes
\oracleQueries oracle queries.

\subsection{Hash function}

\begin{defn}[Collision resistance]
	
	Let $\genericFnFamily: \genericFnKeyspace \times
	\genericFnDomain \rightarrow \genericFnRange$ be a family of
	functions, and let \Adv{} be an algorithm that returns a pair
	of elements in \genericFnDomain. Let
	
	\indent\begin{minipage}[t]{0.5\textwidth}
		\begin{tabbing}
			****\=****\=\kill
			Experiment $\Experiment{\CRLabel}{\genericFn}(\crAdversary)$ \\
			\> $\genericFnKey \getsr \genericFnKeyspace$ \\
			\> $\plaintext[1], \plaintext[2] \leftarrow \crAdversary(\genericFnKey)$ \\
			\> if $\plaintext[1] \neq \plaintext[2] \wedge \genericFnFamily[\genericFnKey](\plaintext[1]) = 
			\genericFnFamily[\genericFnKey](\plaintext[2])$ \\
			\> \> then return $1$ \\
			\> \> else return $0$
		\end{tabbing}
	\end{minipage}
	
	\noindent Then,
	\begin{align*}
		& \Advantage{\CRLabel}{\genericFnFamily}{\crAdversary} = 
		\prob{ \Experiment{\CRLabel}{\genericFn}(\crAdversary) = 1} \\
		& \Advantage{\CRLabel}{\genericFnFamily}(\timeBound) = 
		\max_{\Adv{}} \Advantage{\CRLabel}{\genericFnFamily}{\crAdversary}
	\end{align*}
	where the maximum is taken over all algorithms \Adv{} that run in time
	at most \timeBound.
\end{defn}

\section{Security Definitions}
\label{app:sec_defn}

For our definitions, we will assume that
$\numHonestRSU$ \rsus are honest out of a total of $\numRSUs$ \rsus,
and $\numHonestVehicles$ users are honest out of a total of
$\numVehicles$ users.

\subsection{Unforgeability of coins}
\label{app:defn:unforgeCoin}
We define the unforgeability property with the following
experiment, where \unforgeAdversary is an adversary representing the
malicious vehicles, able to invoke both the honest RSUs and honest
vehicles as oracles.

\begin{center}
	\begin{minipage}{0.75\columnwidth}
		\begin{tabbing}
			*** \= *** \= *** \= \kill
			Experiment $\Experiment{\unforgeLabel}{\repCoinScheme}(\unforgeAdversary)$ 
			\\
			\> $\langle \sPubKey, \sPrivKey \rangle \gets \bank.\serverInit()$ \\
			\> $\langle \genericCoin, \genericTx \rangle~\gets$ \\
			\> \> $\unforgeAdversary^{\langle \genRSU \rangle_{\rsuID \in 
			\nats[\numHonestRSU]}, \langle 
				\genCar{\vID} \rangle_{\vID \in \nats[\numHonestVehicles]}, \bank} (\langle 
				\vPubKey{\vID'}, 
			\vPrivKey{\vID'} \rangle_{\vID' \in [\numHonestVehicles+1, \numVehicles]})$\\
			\> if $\mbox{\genericCoin is not an output of \rsuWithdrawCoin()}$\\ 
			\>\>  $\land~\bank.\serverVerifyTx(\genericTx, \genericCoin) = 1$ \\
			\>\>\> then return 1 \\
			\>\>\> else return 0  
		\end{tabbing}
	\end{minipage}
\end{center}

We define
\begin{align*}
	\Advantage{\unforgeLabel}{\repCoinScheme}(\unforgeAdversary) & = 
	\prob{\Experiment{\unforgeLabel}{\repCoinScheme}(\unforgeAdversary) = 1} \\
	\Advantage{\unforgeLabel}{\repCoinScheme}(\timeBound, \oracleQueries) & = 
	\max_{\unforgeAdversary}	
	\Advantage{\unforgeLabel}{\repCoinScheme}(\unforgeAdversary) 
\end{align*}
where the maximum is taken over all adversary algorithms
\unforgeAdversary that run in total time \timeBound and such that
\unforgeAdversary makes at most \oracleQueries oracle queries.


Intuitively, the adversary \unforgeAdversary can make queries to the bank, and each 
honest \rsu as 
an oracle with inputs of its choice to obtain coins and vouchers, and it can also invoke 
honest users 
with inputs of its choice to
deposit vouchers.  It eventually produces a coin \genericCoin and a transaction 
\genericTx.  The adversary wins if the bank accepts a 
transaction that spends a coin that was not issued by the bank.

\newcommand{\unforgeVoucherExp}{\Experiment{\unforgeVoucherLabel}{\repCoinScheme}(\unforgeAdversary[1],
 \unforgeAdversary[2])\xspace}
 \newcommand{\unforgeVoucherAdv}{	
 \Advantage{\unforgeVoucherLabel}{\repCoinScheme}(\timeBound, 
 	\oracleQueries)}
 
\subsection{Unforgeability of vouchers}
\label{app:defn:unforgeVoucher}
We define the unforgeability property with the following
experiment, where \unforgeAdversary is an adversary representing the
malicious vehicles and \rsus, able to invoke both the honest RSUs and honest
vehicles as oracles.

\begin{center}
	\begin{minipage}{0.75\columnwidth}
		\begin{tabbing}
			*** \= *** \= *** \= \kill
			Experiment 
			$\Experiment{\unforgeVoucherLabel}{\repCoinScheme}(\unforgeAdversary)$
			 
			\\
			\> $\langle \sPubKey, \sPrivKey \rangle \gets \bank.\serverInit()$ \\
			\> for $\rsuID \in \nats[\numRSUs]$ $\langle \genericSig[\rsuID], 
			\genericPRFKey[\rsuID] 
			\rangle \gets \genRSU.\rsuInit()$  \\
			\> for $\vID \in \nats[\numVehicles]$: $\langle \vPubKey{\vID}, 
			\vPrivKey{\vID} 
			\rangle \gets 
			\genCar.\carInit()$ \\
			\> $\langle \genRSUID, \vID, \vIDDep, \advState \rangle \gets 
			\unforgeAdversary[1](\langle 
			\vPubKey{\vID'}, 
			\vPrivKey{\vID'} \rangle_{\vID' \in [\numHonestVehicles+1, \numVehicles]})$ \\
			\> If $\genRSUID \notin \nats[\numHonestRSU] \lor \vID \notin 
			\nats[\numVehicles] \lor \vIDDep \notin \nats[\numHonestVehicles]$ \\
			\> \> then return 0 \\
			\> $\langle \genericCoin, \genericVoucher \rangle\gets 
			\genRSU.\rsuAssignCoin(\pubKey{\vID, \dots})$\\
			\> $\langle \genericVoucherAlt, \txAlt \rangle~\gets 
			\unforgeAdversary[2]^{\langle 
			\genRSU \rangle_{\rsuID \in 
					\nats[\numHonestRSU]}, \langle 
				\genCar{\vID} \rangle_{\vID \in \nats[\numHonestVehicles]}, \bank} 
				(\advState, \langle 
			\vPubKey{\vID'}, 
			\vPrivKey{\vID'} \rangle_{\vID' \in [\numHonestVehicles+1, \numVehicles]})$\\
			\> If $\genCar{\vIDDep}.\rsuVerifyCoin(\pubKey{\vIDAlt}, \genericCoin, 
			\genericVoucherAlt, \txAlt,  \dots ) = 1 \land ~\pubKey{\vID} \neq 
			\pubKey{\vIDAlt}$\\
			\>\>\> then return 1 \\
			\>\>\>else return 0  
		\end{tabbing}
	\end{minipage}
\end{center}

We define
\begin{align*}
	\Advantage{\unforgeVoucherLabel}{\repCoinScheme}(\unforgeAdversary) & = 
	\prob{\Experiment{\unforgeVoucherLabel}{\repCoinScheme}(\unforgeAdversary) 
	= 1} 
	\\
	\Advantage{\unforgeVoucherLabel}{\repCoinScheme}(\timeBound, 
	\oracleQueries) & 
	= 
	\max_{\unforgeAdversary}	
	\Advantage{\unforgeVoucherLabel}{\repCoinScheme}(\unforgeAdversary) 
\end{align*}
where the maximum is taken over all adversary algorithms
\unforgeAdversary that run in total time \timeBound and such that
\unforgeAdversary makes at most \oracleQueries oracle queries.


Intuitively, the adversary \unforgeAdversary[1] given the public keys of all the 
honest 
\rsus, produces the identity of an honest \rsu \genRSUID. This \rsu produces a coin 
and 
a voucher that is linked to that coin.  \unforgeAdversary[2] can make queries to the 
bank, and each honest \rsu as an oracle with inputs of its choice to obtain coins and 
vouchers, and it can also invoke honest users with inputs of its choice to
deposit vouchers.  It eventually produces a coin \genericCoin, and a voucher 
$\genericVoucher' \neq \genericVoucher$. The adversary wins if the voucher is 
linked to the coin and is issued to a user different to the one the coin was issued 
to by the honest \rsu. 

\subsection{Untraceability of coins and vouchers}
\label{app:defn:loc_privacy}

We define the untraceability of \repCoins as an indistinguishability
experiment between oracle calls to honest RSUs, and a simulator
simulating the RSUs.

\begin{minipage}[t]{0.4\columnwidth}
	\begin{tabbing}
		*** \= *** \= \kill
		Experiment 
		$\Experiment{\locPrivacyLabel[\locPrivacyBit]}{\repCoinScheme}(\locPrivacyAdversary[1],
		 \locPrivacyAdversary[2])$ \\
		\> for $\rsuID \in \nats[\numRSUs]$: $\langle \genericSig[\rsuID], 
		\genericPRFKey[\rsuID] \rangle 
		\gets \genRSU.\rsuInit()$  \\
		\> for $\vID \in \nats[\numVehicles]$: $\langle \vPubKey{\vID}, \vPrivKey{\vID} 
		\rangle \gets 
		\genCar.\carInit()$ \\
		\> $\langle \genRSUID[0], \genRSUID[1], \vID, \locPrivacyAdvState \rangle \gets$ \\
		\> \> $\locPrivacyAdversary[1]^{\langle \genRSU 
		\rangle_{\rsuID\in\nats[\numHonestRSU]}, 
			\langle \genCar{\vIDAlt} \rangle_{\vIDAlt \in \nats[\numHonestVehicles]}} 
		\left(\begin{array}{@{}l@{}}\langle \genericSig[\rsuID], \genericPRFKey[\rsuID] 
		\rangle_{\rsuID 
				\in \nats[\numHonestRSU+1][\numRSUs]},\\ \langle \vPubKey{\vIDAlt}, 
				\vPrivKey{\vIDAlt} 
			\rangle_{\vIDAlt \in [\numHonestVehicles+1, \numVehicles]}\end{array}\right)$ \\
		\> if $\genRSUID[0] \not\in \nats[\numHonestRSU] \vee \genRSUID[1] \not\in 
		\nats[\numHonestRSU] \vee \genRSUID[0] = \genRSUID[1] \vee \vID 
		\not\in\nats[\numHonestVehicles]$ \\
		\> \> return $0$ \\
		\> $\langle \genericCoin, \genericVoucher \rangle \gets 
		\genRSU{\genRSUID[\locPrivacyBit]}.\rsuAssignCoin(\pubKey{\vID}, 
		\ldots)$ \\
		\> $\locPrivacyGuess \gets \locPrivacyAdversary[2]^{\langle \genRSU 
			\rangle_{\rsuID\in\nats[\numHonestRSU]}, \genCar}(\genericCoin, 
			\genericVoucher, 
			\locPrivacyAdvState)$ \\
		\> return \locPrivacyGuess
	\end{tabbing}
\end{minipage}

We then define

\begin{align*}
	\Advantage{\locPrivacyLabel}{\repCoinScheme}(\locPrivacyAdversary[1], 
	\locPrivacyAdversary[2])
	& = 
	\prob{\Experiment{\locPrivacyLabel[1]}{\repCoinScheme}(\locPrivacyAdversary[1], 
		\locPrivacyAdversary[2]) = 1} \\
	& \hspace{1em} 
	-~\prob{\Experiment{\locPrivacyLabel[0]}{\repCoinScheme}(\locPrivacyAdversary[1], 
		\locPrivacyAdversary[2]) = 1} \\
	\Advantage{\locPrivacyLabel}{\repCoinScheme}(\timeBound, \oracleQueries)
	& = \max_{\locPrivacyAdversary[1], \locPrivacyAdversary[2]} 
	\Advantage{\locPrivacyLabel}{\repCoinScheme}(\locPrivacyAdversary[1], 
	\locPrivacyAdversary[2])
\end{align*}
where the maximum is taken over all adversaries
\locPrivacyAdversary[1], \locPrivacyAdversary[2] that run in time at
most \timeBound and make at most \oracleQueries queries to their
oracles. Intuitively, a small advantage implies that the adversary
\locPrivacyAdversary[1], \locPrivacyAdversary[2] is not able to
distinguish between a coin+voucher issued by one of two honest RSUs
of its own choosing, and for an honest user of its choice.

\subsection{Unforgeability of transactions}
\label{app:defn:frame}

\begin{minipage}{0.8\columnwidth}
	\begin{tabbing}
		*** \= *** \= \kill
		Experiment $\Experiment{\frameLabel}{\repCoinScheme}(\frameAdversary[1], 
		\frameAdversary[2])$ \\
		\> $\langle \sPubKey, \sPrivKey \rangle \gets \server.\serverInit()$ \\
		\> for $\rsuID \in \nats[\numRSUs]$: $\langle \genericSig[\rsuID], 
		\genericPRFKey[\rsuID] \rangle 
		\gets \genRSU.\rsuInit()$  \\
		\> for $\vID \in \nats[\numVehicles]$: $\langle  \pubKey{\vID}, \privKey{\vID} \rangle 
		\gets 
		\genCar.\carInit()$ \\
		\> $\langle \vID, \vIDAlt, \genericCoin, \genericVoucher, \advState \rangle~\gets$ \\ 
		\> \> $\frameAdversary[1]^{\langle \genRSU{\rsuIDAlt} \rangle_{\rsuIDAlt \in 
				\nats[\numHonestRSU]}, \langle \genCar{\vIDAltAlt} \rangle_{\vIDAltAlt \in 
				\nats[\numHonestVehicles]}, \bank} \left(\begin{array}{@{}l@{}}\langle 
				\genericSig[\rsuID], 
			\genericPRFKey[\rsuID] \rangle_{\rsuID \in 
			\nats[\numHonestRSU+1][\numRSUs]},\\ \langle 
			\vPubKey{\vID}, \vPrivKey{\vID} \rangle_{\vID \in [\numHonestVehicles+1, 
				\numVehicles]} \end{array}\right)$ \\
		\> if $\vID \not\in \nats[\numHonestVehicles] \vee \vIDAlt \not\in 
		\nats[\numVehicles] \vee \rsuID 
		\not\in \nats[\numRSUs] \vee \genericCoin \not \coinVoucherLink \genericVoucher$ 
		\\
		\> \> return $0$ \\
		\> $\genericTx \gets \genCar{\vID}.\carDepositCoin(\genericCoin, \genericVoucher, 
		\genCar{\vIDAlt}, \dots)$ \\
		\> \bank.\serverVerifyTx(\genericCoin, \genericTx, \genCar{\vIDAlt}) \\
		\> $\genericTx' \gets \frameAdversary[2](\genericCoin, \genericVoucher, 
		\genericTx, 
		\advState)$ \\
		\> if $\genericTx \neq \genericTx' \wedge \bank.\serverVerifyTx(\genericTx', \ldots) 
		= 1 \land 
		\genericCoin \coinVoucherLink \genericTx'$ \\
		\>\> then output 1 \\
		\>\> else return 0
	\end{tabbing}
\end{minipage}

We define

\begin{align*}
	\Advantage{\frameLabel}{\repCoinScheme}(\frameAdversary[1], \frameAdversary[2])
	& = \prob{\Experiment{\frameLabel}{\repCoinScheme}(\frameAdversary[1], 
	\frameAdversary[2]) = 
		1} \\
	\Advantage{\frameLabel}{\repCoinScheme}(\timeBound, \oracleQueries)
	& = \max_{\frameAdversary[1], \frameAdversary[2]} 
	\Advantage{\frameLabel}{\repCoinScheme}(\frameAdversary[1], \frameAdversary[2]) 
\end{align*}

\noindent
where the maximum is taken over all adversary algorithms
\frameAdversary[1], \frameAdversary[2] that run in total time
\timeBound and such that \frameAdversary[1] makes at most
\oracleQueries oracle queries.

Intuitively, the adversary \Adv{1} outputs the identities of two users $\genCar{\vID}, 
\genCar{\vIDAlt}$ and a voucher \genericVoucher that is linked to a coin \genericCoin. 
The coin is 
used in a transaction \genericTx by $\genCar{\vID}$ with $\genCar{\vIDAlt}$ that acts as 
the 
merchant. Then, the adversary \Adv{2} wins if given the coin, voucher and the 
transaction 
\genericTx, produces another transaction $\genericTx' \neq \genericTx$ that uses the 
same coin.

\subsection{Anonymity against bank}
\label{app:defn:anon}

\begin{center}
	\begin{minipage}{0.5\columnwidth}
		\begin{tabbing}
			*** \= *** \= \kill
			Experiment 
			$\Experiment{\anonLabel[\anonBit]}{\repCoinScheme}(\anonAdversary[1], 
			\anonAdversary[2])$ \\
			\> $\langle \sPubKey, \sPrivKey \rangle \gets \bank.\serverInit()$ \\
			\> for $\rsuID \in \nats[\numRSUs]$: $\langle \genericSig[\rsuID], 
			\genericPRFKey[\rsuID] 
			\rangle \gets \genRSU.\rsuInit()$  \\
			\> for $\vID \in \nats[\numVehicles]$: $\langle \vPubKey{\vID}, \vPrivKey{\vID} 
			\rangle \gets 
			\genCar.\carInit()$ \\
			\>  $\langle \vID[0], \vID[1], \rsuID, \rsuIDAlt, \anonAdversaryState \rangle~\gets$ 
			\\
			\> \> $\anonAdversary[1]^{\langle \genRSU{\rsuIDAltAlt} \rangle_{\rsuIDAltAlt \in 
					\nats[\numHonestRSU]}, \langle \genCar{\vID} \rangle_{\vID \in 
					\nats[\numHonestVehicles]}, \bank} 
			\left(\begin{array}{@{}l@{}}\langle \genericSig[\rsuIDAltAlt], 
				\genericPRFKey[\rsuIDAltAlt] 
				\rangle_{\rsuIDAltAlt \in \nats[\numHonestRSU+1][\numRSUs]},\\ \langle 
				\vPubKey{\vID}, 
				\vPrivKey{\vID} \rangle_{\vID \in [\numHonestVehicles+1, 
					\numVehicles]}\end{array}\right)$ 
			\\
			\> if $\vID[0] \not\in \nats[\numHonestVehicles] \vee \vID[1] \not\in 
			\nats[\numHonestVehicles] 
			\vee \vID[0] = \vID[1] \vee \rsuID \not\in \nats[\numHonestRSU] \vee
			\rsuIDAlt \not\in \nats[\numHonestRSU]$ \\
			\> \> return $0$ \\
			\> $\langle \genericCoin, \genericVoucher \rangle \gets 
			\genRSU.\rsuAssignCoin(\pubKey{\vID[\anonBit]}, \ldots)$ \\
			\> If $\genericCoin \coinVoucherLink \genericVoucher$ \\
			\> \> return $0$ \\
			\> $\genericTx \gets \genCar{\vID[\anonBit]}.\carDepositCoin(\genericCoin, 
			\genericVoucher, 
			\genRSU{\rsuIDAlt}, 
			\dots)$ \\
			\> $\anonGuess \gets \anonAdversary[2](\genericCoin, \genericTx, \sPubKey, 
			\sPrivKey, 
			\anonAdversaryState)$ \\
			\> return \anonGuess
		\end{tabbing}
	\end{minipage}
\end{center}

The adversary's advantage is defined as

\begin{align*}
	\Advantage{\anonLabel}{\repCoinScheme}(\anonAdversary[1], \anonAdversary[2])
	& = \prob{\Experiment{\anonLabel[0]}{\repCoinScheme}(\anonAdversary[1], 
		\anonAdversary[2]) = 
		1} \\
	& \qquad -~\prob{\Experiment{\anonLabel[1]}{\repCoinScheme}(\anonAdversary[1], 
		\anonAdversary[2])= 1} \\
	\Advantage{\anonLabel}{\repCoinScheme}(\timeBound, \oracleQueries)  
	& = \max_{\anonAdversary{1}, \anonAdversary{2}} 
	\Advantage{\anonLabel}{\repCoinScheme}(\anonAdversary[1], \anonAdversary[2])
\end{align*}

\noindent
where the maximum is taken over all adversaries $(\anonAdversary[1],
\anonAdversary[2])$ running in time at most \timeBound and making at
most \oracleQueries oracle queries.

In the experiment, the adversary \anonAdversary{1} is given oracle access
to the bank, the honest \rsus and users, and produces the identity of two
distinct, honest users \vID[0], \vID[1] along with an \rsu \genRSU
from which to receive a coin and a voucher \genericVoucher and an honest \rsu
\rsuIDAlt at which to deposit it. The adversary \anonAdversary{2} then
guesses \anonBit. Intuitively, if the adversary is unable to
distinguish between a transaction transcript produced by
$\genCar{\vID[0]}$ and $\genCar{\vID[1]}$, then this implies the
transaction transcript does not reveal the identity of the user.

\subsection{Double spending detection}
\label{app:defn:dbl}
%

\noindent
\begin{minipage}{\columnwidth}
	\begin{tabbing}
		*** \= *** \= \kill
		Experiment $\Experiment{\dblSpendLabel}{\repCoinScheme}(\dblSpendAdversary[1], 
		\dblSpendAdversary[2])$ \\
		\> $\langle \sPubKey, \sPrivKey \rangle \gets \server.\serverInit()$ \\
		\> for $\rsuID \in \nats[\numRSUs]$: $\langle \genericSig[\rsuID], \genericPRFKey[\rsuID] \rangle 
		\gets \genRSU.\rsuInit()$  \\
		\> for $\vID \in \nats[\numVehicles]$: $\langle \vPubKey{\vID}, \vPrivKey{\vID} \rangle \gets 
		\genCar.\carInit()$ \\
		\> $\langle \vID, \rsuID, \dblSpendAdversaryState \rangle~\gets$ \\
		\>\> $\dblSpendAdversary[1]^{\langle \genRSU{\rsuIDAlt} \rangle_{\rsuIDAlt \in 
		\nats[\numHonestRSU]}, \langle \genCar{\vIDAlt} \rangle_{\vIDAlt \in 
		\nats[\numHonestVehicles]}} \left(\begin{array}[m]{@{}l@{}} \langle \genericSig[\rsuIDAlt], 
		\genericPRFKey[\rsuIDAlt] \rangle_{\rsuIDAlt \in [\numHonestRSU+1, \numRSUs]}, \\ \langle 
		\vPubKey{\vIDAlt}, \vPrivKey{\vIDAlt} \rangle_{\vIDAlt \in [\numHonestVehicles+1, 
		\numVehicles]}, \sPubKey\end{array}\right)$ \\
		\> if $\vID \not\in \nats[\numVehicles] \vee \rsuID \not\in\nats[\numHonestRSU]$ \\
		\>\> return $0$ \\  
		\> $\langle \genericCoin, \genericVoucher \rangle \gets 
		\genRSU{\rsuID}.\rsuAssignCoin(\genCar{\vID}, \dots)$\\
		\> $\langle \tx, \txAlt, \rsuIDAlt , \rsuIDAltAlt \rangle \gets \dblSpendAdversary[2]^{\langle 
		\genRSU{\rsuIDAltAlt} \rangle_{\rsuIDAltAlt \in \nats[\numHonestRSU]}, \langle \genCar{\vIDAlt} 
		\rangle_{\vIDAlt \in \nats[\numHonestVehicles]}} (\genericCoin, 
		\dblSpendAdversaryState)$ \\
		\> if \= $\rsuIDAlt \in \nats[\numRSUs] \wedge \rsuIDAltAlt \in \nats[\numRSUs] \wedge 
		\tx.\txCoinLabel = \txAlt.\txCoinLabel = \genericCoin  \wedge \tx \neq \txAlt$ \\
		\>\> $\land~\server.\serverVerifyTx(\tx, \genRSU{\rsuIDAlt}, \dots) = 1$ \\
		\>\> $\land~\server.\serverVerifyTx(\txAlt, \genRSU{\rsuIDAltAlt}, \dots) = 1$ \\
		\>\> $\land~\server.\serverDblSpend(\tx, \txAlt) = 0$ \\
		\> *** \= \kill
		\>\> then return 1 \\
		\>\> else return 0
	\end{tabbing}
\end{minipage}

We define

\begin{align*}
	\Advantage{\dblSpendLabel}{\repCoinScheme}(\dblSpendAdversary[1], \dblSpendAdversary[2]) & 
	= \prob{\Experiment{\dblSpendLabel}{\repCoinScheme}(\dblSpendAdversary[1], 
	\dblSpendAdversary[2]) = 1} \\
	\Advantage{\dblSpendLabel}{\repCoinScheme}(\timeBound, \oracleQueries) & = 
	\max_{\dblSpendAdversary[1], \dblSpendAdversary[2]} 
	\Advantage{\dblSpendLabel}{\repCoinScheme}(\dblSpendAdversary[1], \dblSpendAdversary[2]) 
\end{align*}

\noindent
where the maximum is taken over all adversary algorithms
(\dblSpendAdversary[1], \dblSpendAdversary[2]) that run in total time \timeBound and such that
(\dblSpendAdversary[1], \dblSpendAdversary[2]) make at most \oracleQueries oracle queries.

\input{app_security_proof}
\section{Optimizing efficient of withdrawal in \sys}
\label{app:optimal_withdrawal}

\begin{mdframed}[frametitle={Coin Issuing Mechanism in \sys}, font=\small]
	
	\noindent
	\textbf{Parameters:} Commitment scheme \commSchemeDefn;
        non-interactive zero-knowledge proof scheme $\zkpName =
        \langle \zkpGen$, \zkpVerify, \zkpSim, $\zkpExtract\rangle$
        with random oracle \randomOracle, a collision resistant hash function $\roHash: \{0,1\}^{\ast} \rightarrow \grp$

	
	\smallskip\noindent\underline{$\genRSU.\rsuAssignCoin(\genCar,\coinRandVehicle, \vCert{\vID})$}
	
  \textit{/* Let $\signedComm[1][\rsuID]$ and $\signedComm[2][\rsuID]$ be the signature by the server on the PRF keys $\prfKey[1][\rsuID]$ and $\prfKey[2][\rsuID]$ that have been signed during initialization. */}

	\begin{enumerate}
		\item $\coinRand \getsr \integersModArg{\grpOrd}$
		\item \label{assign:commitPRF1} $\keyCommitment[1][\rsuID] \gets \commScheme(\prfKey[1][\rsuID]; \genericRand{1})$
		\item \label{assign:commitPRF2} $ \keyCommitment[2][\rsuID] \gets \commScheme(\prfKey[2][\rsuID]; \genericRand{2})$
		\item \label{assign:zkp1} $\zkpSig[1][\rsuID] \gets \zprove(\keyCommitment[1][\rsuID],  \sPubKeyBlind, \signedComm[1][\rsuID], \genericRand{1}, \prfKey[1][\rsuID])$
		\item \label{assign:zkp2} $\zkpSig[2][\rsuID] \gets \zprove(\keyCommitment[2][\rsuID], \sPubKeyBlind, \signedComm[2][\rsuID], \genericRand{2}, \prfKey[2][\rsuID])$
		\item $\coinID{} \gets \roHash(\pubKey{\vID}, \coinRandVehicle, \coinRand, \auxInfo)$
		\item $\prfEval[1][\coinID{}] \gets \dyPRF[\prfKey[1][\rsuID]](\coinID{})$
		
		\item $ \prfEval[2][\coinID{}] \gets   \genCarID^{\coinRand} \cdot \dyPRF[\prfKey[2][\rsuID]](\coinID{})$ \\
		
		\item $\zkpPRF[\coinID{}] \gets \zkpGen[\langle \alpha_1, \alpha_2, \beta_1, \beta_2 \rangle]^{\randomOracle}(\keyCommitment[1][\rsuID], \keyCommitment[2][\rsuID], \prfEval[1][\coinID{}], \prfEval[2][\coinID{}], \coinID{}) \\ 
		\{\keyCommitment[1][\rsuID]= \commScheme(\alpha_1; \beta_1)  \land \keyCommitment[2][\rsuID] = \commScheme(\alpha_2; \beta_2)  
		\\ \land \prfEval[1][\coinID{}] = \dyPRF[\alpha_1](\coinID{})  \land \prfEval[2][\coinID{}] = \genCarID^{\coinRand} \cdot \dyPRF[\alpha_2](\coinID{})  \}$
    \item $\genericCoin \assign \langle \zkpSig[1][\rsuID], \zkpSig[2][\rsuID], \signedComm[1][\rsuID], \signedComm[2][\rsuID] \rangle$ 
    \item $\genericVoucher \assign \genericVoucherDefn$ 
    \item Output $\repCoinAbbrv \assign \langle \genericCoin, \genericVoucher \rangle$
	\end{enumerate}
\end{mdframed}

The rest of the interfaces are almost identical in \sys to our \ecash scheme, and therefore we omit the formal descriptions. 
\fi

\mysubsection{Security}
\label{sec:model:security}
We will describe the security guarantees provided by
our \ecash scheme. We note that the following properties of unforgeability, anonymity, authenticity of balance and identification of double-spenders is provided by all \ecash schemes. Our definitions adapt the model to multiple issuers. In all our definitions we assume that there are \numRSUs ATMs out which \numHonestRSU are honest, there are \numVehicles users out of which \numHonestVehicles are honest. For notational simplicity, we assume that $\rsu_1, \dots, \rsu_{\numHonestRSU}$ are honest, and $\rsu_{\numHonestRSU + 1}, \dots, \rsu_{\numRSUs}$ are dishonest. Similarly, we assume that $\user{1}, \dots, \user{\numHonestVehicles}$ are honest, and  $\user{\numHonestVehicles + 1}, \dots, \user{\numVehicles}$ are dishonest. 

\subsubsection{Unforgeability}
In our scheme, only the bank \bank
can mint new coins. Formally, we define this property with a security 
experiment $\Experiment{\unforgeLabel}{\repCoinScheme}$, 
where the adversary represents the malicious users and \rsus, and 
is able to invoke the bank, the honest \rsus and the honest users as oracles.
It eventually produces a coin \genericCoin, a voucher \genericVoucher, possibly linked to the coin, and a transaction \genericTx that spends the coin.  The adversary wins if the bank accepts the transaction, and the coin was not issued by the bank.

\begin{center}
    \begin{tabbing}
        ** \= ** \= ** \= \kill
        Experiment $\Experiment{\unforgeLabel}{\repCoinScheme}(\unforgeAdversary)$ 
        \\
        \> $\langle \sPubKey, \sPrivKey \rangle \assign \bank.\serverInit()$ \\
 \> for $\rsuID \in \nats[\numRSUs]$: $\langle \rsuPubKey{\rsuID}, \rsuPrivKey{\rsuID} 
    \rangle \assign \genRSU.\rsuInit()$  \\
    \> for $\vID \in \nats[\numVehicles]$: $\langle \vPubKey{\vID}, \vPrivKey{\vID} 
    \rangle \assign \genCar.\carInit()$ \\
        \> $\langle \genericCoin, \genericVoucher, \genericTx \rangle \assign
        \unforgeAdversary^{\oracle(\cdot)}\left(\adversaryInputKeyPairs\right)$\\
        \> Let $\setOfWithdrawnCoins$ be the set of outputs from $\server.\serverIssueCoin()$\\
        \> If $\bank.\serverVerifyTx(\genericCoin, \genericVoucher, \genericTx) = 0 \land \genericCoin \notin \setOfWithdrawnCoins$ \\
        \>\> then return 1 \\
        \>\> else return 0  
    \end{tabbing}
\end{center}

,  where \oracle is an oracle answering queries to the honest ATMs $\langle \genRSU \rangle_{\rsuID \in \nats[\numHonestRSU]}$, the honest vehicles $\langle 
 \genCar \rangle_{\vID \in \nats[\numHonestVehicles]}$, and the bank $\bank$

\medskip
\noindent
The adversary's advantage is defined as
\begin{align*}
	\Advantage{\unforgeLabel}{\repCoinScheme}(\unforgeAdversary) & = 
	\prob{\Experiment{\unforgeLabel}{\repCoinScheme}(\unforgeAdversary) = 1} \\
	\Advantage{\unforgeLabel}{\repCoinScheme}(\timeBound, \oracleQueries) & = 
	\max_{\unforgeAdversary}	
	\Advantage{\unforgeLabel}{\repCoinScheme}(\unforgeAdversary) 
\end{align*}
where the maximum is taken over all adversary algorithms
\unforgeAdversary that run in total time \timeBound and such that
\unforgeAdversary makes at most \oracleQueries oracle queries.

\subsubsection{Anonymity}
We will ensure that the coins spent at a merchant and the transactions
submitted to the bank do not reveal any information about the spender,
unless the coin has been double spent. We define this property with an
indistinguishability experiment $\Experiment{\anonLabel[\anonBit]}{\repCoinScheme}(\anonAdversary[1], 
    \anonAdversary[2])$, 
where a distinguisher -- signifying the
bank, a merchant or an ATM -- receives transactions deposited by two
honest users and attempts to distinguish between them.

Intuitively, \anonAdversary[1] selects two honest users \genCar{\vID[0]} and \genCar{\vID[1]} with public keys $\vPubKey{\vID[0]}$ and $\vPubKey{\vID[1]}$ respectively, and two \rsus $\genRSU{\rsuID}$ and $\genRSU{\rsuIDAlt}$. $\genRSU{\rsuID}$ issues a coin and a voucher for $\genCar{\vID[\anonBit]}$. Then, \genCar{\vID[\anonBit]} produces a transaction spending the coin at $\genRSU{\rsuIDAlt}$; $\genRSU{\rsuIDAlt}$ can be malicious as per the experiment but $\genRSU{\rsuID}$ must be honest because otherwise they may collude and trivially win the experiment. \anonAdversary[2] is provided the coin, 
the voucher and the transaction and guesses the bit \anonBit.

\begin{tabbing}
    ** \= ** \= \kill
    Experiment 
    $\Experiment{\anonLabel[\anonBit]}{\repCoinScheme}(\anonAdversary[1], 
    \anonAdversary[2])$ \\
    \> $\langle \sPubKey, \sPrivKey \rangle \assign \bank.\serverInit()$ \\
    \> for $\rsuID \in \nats[\numRSUs]$: $\langle \rsuPubKey{\rsuID}, \rsuPrivKey{\rsuID} 
    \rangle \assign \genRSU.\rsuInit()$  \\
    \> for $\vID \in \nats[\numVehicles]$: $\langle \vPubKey{\vID}, \vPrivKey{\vID} 
    \rangle \assign \genCar.\carInit()$ \\
    \>  $\langle \vID[0], \vID[1], \rsuID, \genDepRSUID, \anonAdversaryState \rangle \assign 
    \anonAdversary[1]^{\oracle} 
    \mathopen{}\left(\adversaryInputKeyPairs \right)\mathclose{}$ 
    \\
    \> if $\vID[0] \not\in \nats[\numHonestVehicles] \vee \vID[1] \not\in 
    \nats[\numHonestVehicles] 
    \vee \vID[0] = \vID[1] \vee \rsuID \not\in \nats[\numHonestRSU] \vee
    \genDepRSUID \not\in \nats[\numVehicles]$ \\
    \> \> return $0$ \\
    \> $\langle \genericCoin, \genericVoucher_{\anonBit} \rangle \assign 
    \genRSU.\rsuAssignCoin(\pubKey{\vID[\anonBit]}, \ldots)$ \\
    \> If $\genericCoin \not \coinVoucherLink \genericVoucher_{\anonBit}$\\
    \> \> return $0$ \\
    \> $\genericTx_{\anonBit} \assign \genCar{\vID[\anonBit]}.\carDepositCoin(\genericCoin, 
    \genericVoucher_\anonBit, \vPubKey{\genDepRSUID}, 
    \dots)$ \\
    \> $\anonGuess \assign \anonAdversary[2]^{\oracle(\cdot)}(\genericCoin, \genericVoucher_{\anonBit}, \genericTx_{\anonBit}, \sPubKey, 
    \sPrivKey, 
    \anonAdversaryState)$ \\
    \> return \anonGuess
\end{tabbing}

, where \oracle is an oracle answering queries to the honest ATMs $\langle \genRSU \rangle_{\rsuID \in \nats[\numHonestRSU]}$, the honest vehicles $\langle 
 \genCar \rangle_{\vID \in \nats[\numHonestVehicles]}$, and the bank $\bank$

\noindent
The adversary's advantage is defined as
\begin{align*}
	\Advantage{\anonLabel}{\repCoinScheme}(\anonAdversary[1], \anonAdversary[2])
	& = \prob{\Experiment{\anonLabel[0]}{\repCoinScheme}(\anonAdversary[1], 
		\anonAdversary[2]) = 1}
             - \prob{\Experiment{\anonLabel[1]}{\repCoinScheme}(\anonAdversary[1], 
		\anonAdversary[2])= 1} \\
	\Advantage{\anonLabel}{\repCoinScheme}(\timeBound, \oracleQueries)  
	& = \max_{\anonAdversary{1}, \anonAdversary{2}} 
	\Advantage{\anonLabel}{\repCoinScheme}(\anonAdversary[1], \anonAdversary[2])
\end{align*}
where the maximum is taken over all adversaries $(\anonAdversary[1],
\anonAdversary[2])$ running in time at most \timeBound and making at
most \oracleQueries oracle queries.

\subsubsection{Authenticity of balance}
We will ensure that an adversary, representing dishonest users and ATMs, cannot dishonestly withdraw from an honest user's account. They also cannot forge unauthorized transactions to credit their own accounts using coins that have been withdrawn by honest users. (This also ensures that an honest user cannot be falsely implicated for double spending). We formally define this property with the experiment $\Experiment{\frameLabel}{\repCoinScheme}(\frameAdversary)$.

Intuitively, the adversary $\frameAdversary[1]$ is given oracle access to the honest parties, and outputs the identity of an honest user $\genCar{\vID}$, the identity of an ATM $\genRSU$ which may be malicious, and a receipt of a withdrawal \userReceipt. Then, \genRSU produces a coin and voucher intended for $\genCar{\vID}$. The adversary $\frameAdversary[2]$ produces a transaction and wins the experiments if the bank accepts the withdrawal using the receipt, but $\genCar{\vID}$ has not performed the withdrawal, or if the bank accepts the transaction 
$\genericTx$ spending the coin \genericCoin but no honest user has performed the transaction.

\noindent
\begin{tabbing}
    ** \= ** \= ** \= \kill	
    Experiment $\Experiment{\frameLabel}{\repCoinScheme}(\frameAdversary)$ \\
\> $\langle \sPubKey, \sPrivKey \rangle \assign \bank.\serverInit()$ \\
    \> for $\rsuID \in \nats[\numRSUs]$: $\langle \rsuPubKey{\rsuID}, \rsuPrivKey{\rsuID} \rangle 
    \gets \genRSU.\rsuInit()$  \\
    \> for $\vID \in \nats[\numVehicles]$: $\langle  \pubKey{\vID}, \privKey{\vID} \rangle 
    \gets 
    \genCar.\carInit()$ \\
    \> $\langle  \vID, \genDepRSUID, \rsuID, \frameAdversaryState \rangle \gets 
       \frameAdversary[1]^{\oracle} \mathopen{}\left(\adversaryInputKeyPairs \right)\mathclose{}$ \\
\> If $\vID \notin \nats[\numHonestVehicles]  \vee \rsuID \notin \nats[\numRSUs]$ \\
\> \> then output 0\\
\> $\langle \genericCoin, \genericVoucher \rangle \assign \genRSU.\rsuAssignCoin(\vPubKey{\vID}, \dots)$\\
\> $\langle \userReceipt, \genericVoucherAlt, \genericTx \rangle \assign \frameAdversary[2]^{\oracle(\cdot)}(\genericCoin, \genericVoucher, \frameAdversaryState)$ \\ 
\> Let $\setOfValidReceipts$ be the set of valid receipts \\
\> \> returned by invocations to 
$\langle \genCar{\vID}.\carGetCoin \rangle_{\vID \in 
\nats[\numHonestVehicles]}$\\
\> If $\server.\serverUpdateWithdrawal(\vPubKey{\vID}, \userReceipt, \rsuPubKey{\rsuID})~\land~\userReceipt \notin \setOfValidReceipts$ \\ 
\> \> $~\vee \bank.\serverVerifyTx(\genericCoin, \genericVoucherAlt, \genericTx) = 0$\\
    \>\>\> then output 1 \\
    \>\>\> else return 0
\end{tabbing}

, where \oracle is an oracle answering queries to the honest ATMs $\langle \genRSU \rangle_{\rsuID \in \nats[\numHonestRSU]}$, the honest vehicles $\langle 
 \genCar \rangle_{\vID \in \nats[\numHonestVehicles]}$, and the $\bank$

\noindent
The adversary's advantage is defined as
\begin{align*}
	\Advantage{\frameLabel}{\repCoinScheme}(\frameAdversary[1], \frameAdversary[2])
	& = \prob{\Experiment{\frameLabel}{\repCoinScheme}(\frameAdversary[1], \frameAdversary[2]) = 
		1} \\
	\Advantage{\frameLabel}{\repCoinScheme}(\timeBound, \oracleQueries)
	& = \max_{\frameAdversary[1], \frameAdversary[2]} 
	\Advantage{\frameLabel}{\repCoinScheme}(\frameAdversary[1], \frameAdversary[2]) 
\end{align*}
where the maximum is taken over all adversary algorithms $(\frameAdversary[1], \frameAdversary[2])$, that run in total time
\timeBound and such that $\frameAdversary[1]$ makes at most
\oracleQueries oracle queries.

\subsubsection{Fair exchange}
We will ensure that a user that receives a coin during a withdrawal must provide a receipt to the ATM otherwise it is unable to spend the coin without being detected by the bank. Combined with authenticity of balance where a user does not provide a receipt until it receives a coin from the ATM, this property implies a fair exchange enabled through the bank acting as the trusted third party\footnote{It is known that a fair exchange without a TTP is impossible}. We define this property using the experiment $\Experiment{\fairexchangeLabel}{\repCoinScheme}(\fairexchangeComb)$. 

Intuitively, in the experiment the adversary \fairexchangeAdversary[1] is given oracle access to all the honest parties, and outputs the identity a user and an honest ATM. The honest ATM issues a coin and a corresponding voucher for the user $\genCar$. The adversary \fairexchangeAdversary[2] wins if given the coin and the voucher it can produce a transaction that is accepted by the bank, but the ATM 
\genRSU did not receive a valid receipt for \genericCoin.

\noindent
\begin{tabbing}
    ** \= ** \= ** \= \kill	
    Experiment $\Experiment{\fairexchangeLabel}{\repCoinScheme}(\fairexchangeComb)$ \\
\> $\langle \sPubKey, \sPrivKey \rangle \assign \bank.\serverInit()$ \\
    \> for $\rsuID \in \nats[\numRSUs]$: $\langle \rsuPubKey{\rsuID}, \rsuPrivKey{\rsuID} \rangle 
    \gets \genRSU.\rsuInit()$  \\
    \> for $\vID \in \nats[\numVehicles]$: $\langle  \pubKey{\vID}, \privKey{\vID} \rangle 
    \gets 
    \genCar.\carInit()$ \\
    \> $\langle \vID, \rsuID, \frameAdversaryState \rangle \gets 
       \frameAdversary[1]^{\oracle} \mathopen{}\left(\adversaryInputKeyPairs \right)\mathclose{}$ \\
\> If $\vID \notin \nats[\numVehicles]  \vee \rsuID \notin \nats[\numHonestRSU]$ \\
\> \> then output 0\\
\> $\langle \genericCoin, \genericVoucher \rangle \assign \genRSU.\rsuAssignCoin(\vPubKey{\vID}, \dots)$\\
\> $\genericTx \assign \frameAdversary[2]^{\oracle(\cdot)}(\genericCoin, \genericVoucher, \frameAdversaryState)$ \\ 
\> Let $\setOfValidReceipts$ be the set of valid receipts \\
\> \> returned by invocations to 
$\langle \genCar{\vID}.\carGetCoin \rangle_{\vID \in 
\nats[\numHonestVehicles]}$\\
\> If $(\bank.\serverVerifyTx(\genericCoin, \genericVoucher, \genericTx) = 0 ~\land ~\text{\genRSU does not have a valid receipt for \genericCoin}) $ \\ 
    \>\> then output 1 \\
    \>\> else return 0
\end{tabbing}

, where \oracle is an oracle answering queries to the honest ATMs $\langle \genRSU \rangle_{\rsuID \in \nats[\numHonestRSU]}$, the honest vehicles $\langle 
 \genCar \rangle_{\vID \in \nats[\numHonestVehicles]}$, and the $\bank$

\noindent
The adversary's advantage is defined as
\begin{align*}
	\Advantage{\fairexchangeLabel}{\repCoinScheme}(\fairexchangeComb)
	& = \prob{\Experiment{\fairexchangeLabel}{\repCoinScheme}(\fairexchangeComb) = 
		1} \\
	\Advantage{\fairexchangeLabel}{\repCoinScheme}(\timeBound, \oracleQueries)
	& = \max_{\fairexchangeComb} 
	\Advantage{\fairexchangeLabel}{\repCoinScheme}(\fairexchangeComb) 
\end{align*}
where the maximum is taken over all adversary algorithms $(\fairexchangeComb)$, that run in total time
\timeBound and such that $\fairexchangeAdversary[1]$ makes at most
\oracleQueries oracle queries.

\subsubsection{Identification of double spenders}
We will ensure that two valid transactions double spending the same coin will identify the user responsible. Alternatively, if the same coin is issued to two different users by a malicious ATM, then the identity of the offending ATM will be disclosed to the bank. We formally define this with the experiment $\Experiment{\dblSpendLabel}{\repCoinScheme}(\dblSpendAdversary[1], \dblSpendAdversary[2])$.

Intuitively, \dblSpendAdversary[1] selects a user $\genCar{\vID}$ and an honest \rsu \genRSU. The ATM issues a coin and a voucher for $\genCar{\vID}$. \dblSpendAdversary[2] then either double issues the coin by creating another voucher \genericVoucherAlt, or double-spends the coin by creating two transactions \genericTx and \genericTx'. If the bank is unable to identify (by outputting the public key) of the double issuing ATM \genRSU or the double spending user  \genCar{\vID} given the transaction, then the adversary wins the experiment. Note that the adversary also wins the experiment if the bank outputs the public key of $\genCar{\vIDAlt} \neq \genCar{\vID}$ or $\rsu_\genRSUIDAlt \neq \genRSU$

\label{appx:proof_dbl}
\noindent
\begin{tabbing}
    ** \= ** \= ** \= \kill
    Experiment $\Experiment{\dblSpendLabel}{\repCoinScheme}(\dblSpendAdversary[1], 
    \dblSpendAdversary[2])$ \\
\> $\langle \sPubKey, \sPrivKey \rangle \assign \bank.\serverInit()$ \\
        \> for $\rsuID \in \nats[\numRSUs]$: $\langle \rsuPubKey{\rsuID}, \rsuPrivKey{\rsuID} 
        \rangle \assign \genRSU.\rsuInit()$  \\
        \> for $\vID \in \nats[\numVehicles]$: $\langle \vPubKey{\vID}, \vPrivKey{\vID} 
        \rangle \assign \genCar.\carInit()$ \\
    \> $\langle \vID, \rsuID, \dblSpendAdversaryState \rangle \assign
    \dblSpendAdversary[1]^{\oracle} \mathopen{}\left(\adversaryInputKeyPairs \right)\mathclose{}$ \\
    \> if $\vID \not\in \nats[\numVehicles] \vee \rsuID \not\in\nats[\numRSUs]$ \\
    \>\> return $0$ \\
	\> $\genericCoin \assign \genRSU.\rsuWithdrawCoin(\dots)$\\
    \> $\langle \genericCoin, \genericVoucher \rangle \assign 
    \genRSU{\rsuID}.\rsuAssignCoin(\genCar{\vID}, \dots)$\\
    \> $\langle \genericVoucherAlt, \genericTx, \genericTxAlt \rangle \assign \dblSpendAdversary[2]^{\oracle} (\genericCoin, \genericVoucher,  
    \dblSpendAdversaryState)$ \\
    \> $\returnVar \gets \server.\serverVerifyTx(\genericCoin, \genericVoucher, \genericTx)$ \\
            \> $\returnVarAlt \gets \server.\serverVerifyTx(\genericCoin, \genericVoucherAlt, \genericTxAlt)$ \\
\> if $\returnVar \neq 0$\\ 
\> \> then return $0$ \\
 \> if $\returnVarAlt \neq \pubKey{\vID} \land \returnVarAlt \neq \rsuPubKey{\rsuID}$ \\
            \> \> then return 1 \\
            \> \> else return 0
\end{tabbing}

, where \oracle is an oracle answering queries to the honest ATMs $\langle \genRSU \rangle_{\rsuID \in \nats[\numHonestRSU]}$, the honest vehicles $\langle 
 \genCar \rangle_{\vID \in \nats[\numHonestVehicles]}$, and the bank $\bank$

\medskip\noindent
We define
\begin{align*}
	\Advantage{\dblSpendLabel}{\repCoinScheme}(\dblSpendAdversary[1], \dblSpendAdversary[2]) & 
	= \prob{\Experiment{\dblSpendLabel}{\repCoinScheme}(\dblSpendAdversary[1], 
	\dblSpendAdversary[2]) = 1} \\
	\Advantage{\dblSpendLabel}{\repCoinScheme}(\timeBound, \oracleQueries) & = 
	\max_{\dblSpendAdversary[1], \dblSpendAdversary[2]} 
	\Advantage{\dblSpendLabel}{\repCoinScheme}(\dblSpendAdversary[1], \dblSpendAdversary[2]) 
\end{align*}

\noindent
where the maximum is taken over all adversary algorithms
(\dblSpendAdversary[1], \dblSpendAdversary[2]) that run in total time \timeBound and such that
(\dblSpendAdversary[1], \dblSpendAdversary[2]) make at most \oracleQueries oracle queries.


\subsubsection{Untraceability}
We will introduce a new security property that is important for schemes such as ours that can dispense cash through intermediaries. If a voucher includes information about the \rsu that issued it, e.g., in the form of a signature, then the voucher will divulge information regarding the whereabouts of the user to the merchant and the bank. \ecash systems typically do not need to protect this information since they only allow withdrawals directly from the bank. However, as we are assuming that \rsus can be physical entities, protecting the user's locations from untrusted merchants becomes critical.

We will ensure that information regarding the honest \rsu where an honest 
user withdrew a coin is not available to the merchant where it deposits coins and to the bank. We 
formalize this using an indistinguishability experiment  $\Experiment{\locPrivacyLabel[\locPrivacyBit]}{\repCoinScheme}(\locPrivacyAdversary[1],
     \locPrivacyAdversary[2])$ as an inability to distinguish
a coin + voucher dispensed by one of two honest \rsus to an 
honest user of the adversary's choice

\begin{tabbing}
    *** \= *** \= \kill
    Experiment 
    $\Experiment{\locPrivacyLabel[\locPrivacyBit]}{\repCoinScheme}(\locPrivacyAdversary[1],
     \locPrivacyAdversary[2])$ \\
	\> $\langle \sPubKey, \sPrivKey \rangle \assign \bank.\serverInit()$ \\
    \> for $\rsuID \in \nats[\numRSUs]$: $\langle \rsuPubKey{\rsuID}, \rsuPrivKey{\rsuID} 
        \rangle \assign \genRSU.\rsuInit()$  \\
        \> for $\vID \in \nats[\numVehicles]$: $\langle \vPubKey{\vID}, \vPrivKey{\vID} 
        \rangle \assign \genCar.\carInit()$ \\
        \> $\langle \genRSUID[0], \genRSUID[1], \vID, \locPrivacyAdvState \rangle \assign \locPrivacyAdversary[1]^{\oracle} 
    \mathopen{}\left(\adversaryInputKeyPairs \right)\mathclose{}$ \\
    \> if $\genRSUID[0] \not\in \nats[\numHonestRSU] \vee \genRSUID[1] \not\in 
    \nats[\numHonestRSU] \vee \genRSUID[0] = \genRSUID[1] \vee \vID 
    \not\in\nats[\numHonestVehicles]$ \\
    \> \> return $0$ \\
    \> $\langle \genericCoin, \genericVoucher \rangle \assign 
    \genRSU{\genRSUID[\locPrivacyBit]}.\rsuAssignCoin(\pubKey{\vID}, 
    \ldots)$ \\
    \> $\locPrivacyGuess \assign \locPrivacyAdversary[2]^{\oracle}(\genericCoin, 
        \genericVoucher, 
        \locPrivacyAdvState)$ \\
    \> return \locPrivacyGuess
\end{tabbing}

, where \oracle is an oracle answering queries to the honest ATMs $\langle \genRSU \rangle_{\rsuID \in \nats[\numHonestRSU]}$, the honest vehicles $\langle 
 \genCar \rangle_{\vID \in \nats[\numHonestVehicles]}$, and the bank $\bank$

\medskip
\noindent
We then define
\begin{align*}
	\Advantage{\locPrivacyLabel}{\repCoinScheme}(\locPrivacyAdversary[1], 
	\locPrivacyAdversary[2])
	& = 
	\prob{\Experiment{\locPrivacyLabel[1]}{\repCoinScheme}(\locPrivacyAdversary[1], 
		\locPrivacyAdversary[2]) = 1}
	- \prob{\Experiment{\locPrivacyLabel[0]}{\repCoinScheme}(\locPrivacyAdversary[1], 
		\locPrivacyAdversary[2]) = 1} \\
	\Advantage{\locPrivacyLabel}{\repCoinScheme}(\timeBound, \oracleQueries)
	& = \max_{\locPrivacyAdversary[1], \locPrivacyAdversary[2]} 
	\Advantage{\locPrivacyLabel}{\repCoinScheme}(\locPrivacyAdversary[1], 
	\locPrivacyAdversary[2])
\end{align*}
where the maximum is taken over all adversaries
\locPrivacyAdversary[1], \locPrivacyAdversary[2] that run in time at	
most \timeBound and make at most \oracleQueries queries to their
oracles. Intuitively, a small advantage implies that the adversaries is not able to
distinguish between a coin+voucher issued by one of two honest ATMs
of its own choosing, and for an honest user of its choice.

\section{Security Analysis}

We formally prove that our scheme satisfies all the properties outlined in \secref{sec:model:security}

\subsubsection{Proof of unforgeability}

\begin{theorem}
	
	The \ecash scheme described in \secref{sec:scheme} ensures that coins are 
	unforgeable
	
\[	
\unforgeAdv \le \ufcmaBlindAdv + \crAdv
\]

where \unforgeAdversary makes \oracleQueries to the signature oracle. 
	
\end{theorem}

\begin{proof}
Let $\unforgeAdversary$ produce the coin $\genericCoin = \langle \keyCommitment[1], \keyCommitment[2], \rsuPrivKeyCommitment, \signedCoin \rangle$. Then, 
it must have either i) produced a valid signature on $\hashFn(\keyCommitment[1], \keyCommitment[2], \rsuPrivKeyCommitment)$ without the bank's private key, or ii) produced to collision on the hash function $\hashFn$ such that $\hashFn(\keyCommitment[1], \keyCommitment[2], \rsuPrivKeyCommitment) =\hashFn(\keyCommitment[1]', \keyCommitment[2]', \rsuPrivKeyCommitment')$. In the latter case, $\unforgeAdv$ can use the coin $\langle \keyCommitment[1]', \keyCommitment[2]', \rsuPrivKeyCommitment', \signedCoin \rangle$ to produce $\genericCoin$. This happens with probability 
at most $\ufcmaBlindAdv + \crAdv$. 
\end{proof}

\subsubsection{Proof of anonymity}

\begin{theorem}
The \ecash scheme described in \secref{sec:scheme} $\repCoinScheme$ provides anonymity assuming that the zero-knowledge proof scheme 
satisfies the zero-knowledge property, the commitment scheme produces computationally hiding commitments and the PRF scheme produces outputs indistinguishable from random, i.e., 
	
	\begin{multline*}
		\probDiff{\anonLabel[0]}{\anonLabel[1]}  \\
		\le 4 \cdot \zkpHidingAdv + 2 \cdot \hidingAdv + 2 \cdot \PRFAdv
	\end{multline*}
	

\end{theorem}

\setcounter{ctr}{-1}
\begin{proof}

Consider $\anonAdversary[2]$'s view in the experiment $\anonExp[0]$. The coin $\langle \keyCommitment[1], \keyCommitment[2], \rsuPrivKeyCommitment, \signedComm[1] \rangle$ is identically distributed in \anonExp[0] and \anonExp[1], so we do not include it in further discussion. Similarly, the random coins $\coinRandVehicle \sample \{0,1\}^{\secParam}$, and the value $\txRand \assign \crHash(\rsuPubKey{\vIDDep}, \coinRandVehicle)$ can be excluded since they are independent of the user's identity. 

The other components of the view include

	\begin{enumerate}
  	\item $\privKeyCommitment \assign \commFn(\vPrivKey{\vID[0]}, \userPRFKey[0])$

    \item $\zkp{\privKey{\vID}} \assign \zprove{\sPubKeyBlindZKP}(\privKeyCommitment, \userCLSig{\privKey{\vID[0]}})$ 
	\item $\coinRand \assign \crHash(\privKeyCommitment)$
    \item $\coinID \assign \coinRand + 1 \bmod \grpOrd$
    \item $\voucherPRFEval \assign \dyPRF[\prfKey[1]](\coinID)$  
\item $\voucherDblSpendToken \assign \dyPRF[\prfKey[2]](0)^{\coinRand}$
    \item $\zkp{\coinID} \assign \zkpGen[\alpha_1, \beta_1, \alpha_2, \beta_2, \gamma, \delta](\keyCommitment[1], \keyCommitment[2], \voucherPRFEval, \voucherDblSpendToken, \rsuPrivKeyCommitment, \coinID, \coinRand)$
    \item $\txPRFEval \assign \pubKey{\vID[0]} \dyPRF[\userPRFKey[\vID[0]]](\coinID)^{\txRand}$
\item $\zkp{\genericTx} \assign \zkpGen[\alpha, \beta, \gamma](\privKeyCommitment, \txPRFEval, \coinID, \txRand)$
	\end{enumerate}

	Also \anonAdversary[2]'s view in \anonExp[1], excluding the random coins are the following
	
	\begin{enumerate}
	\item $\privKeyCommitment' \assign \commFn(\vPrivKey{\vID[1]}, \userPRFKey[1])$
    \item $\zkp{\privKey{\vID}}' \assign \zprove{\sPubKeyBlindZKP}(\privKeyCommitment', \userCLSig{\privKey{\vID[1]}})$
	\item $\coinRand' \assign \crHash(\privKeyCommitment')$
    \item $\coinID' \assign \coinRand' + 1 \bmod \grpOrd$
    \item $\voucherPRFEval' \assign \dyPRF[\prfKey[1]](\coinID')$  
\item $\voucherDblSpendToken' \assign \dyPRF[\prfKey[2]](0)^{\coinRand'}$
    \item $\zkp{\coinID'} \assign \zkpGen[\alpha_1, \beta_1, \alpha_2, \beta_2, \gamma, \delta](\keyCommitment[1], \keyCommitment[2], \voucherPRFEval', \voucherDblSpendToken', \rsuPrivKeyCommitment, \coinID', \coinRand')$   
    \item $\txPRFEval' \assign \pubKey{\vID[1]} \dyPRF[\userPRFKey[\vID[1]]](\coinID')^{\txRand}$
\item $\zkp{\genericTx} \assign \zkpGen[\alpha, \beta, \gamma](\privKeyCommitment, \txPRFEval, \coinID, \txRand)$
	\end{enumerate}

	Now, consider the following sequence of hybrid experiments

	\medskip\nextHyb: This hybrid is identical to \anonExp[0] except that the commitment $\privKeyCommitment$, is replaced by a commitment to zeroes and the NIZK proof of possession of a signature is generated by invoking the simulation \zkpSim. \anonAdversary[2]'s view in \thisHyb is indistinguishable from its view in \anonExp[0] due to the hiding property of the commitment scheme, and the 
	zero-knowledge property of the NIZK scheme.

\medskip\nextHyb: This hybrid is identical to \prevHyb except that the PRF evaluation $\txPRFEval$ is replaced by a random element in \grp. \anonAdversary[2]'s view in \prevHyb is indistinguishable from its view in \thisHyb because if \anonAdversary[2] can distinguish between the two experiments, then it can distinguish between $\dyPRF[\userPRFKey[0]](\simOp{\coinID}) = (\txPRFEval/\pubKey{\vID[0]})^{\txRand^{-1}}$ and a random element in \grp. Note that \grp has a prime order, and thus $\dyPRF[\userPRFKey[0]](\simOp{\coinID})$ is a generator with overwhelming probability. Consequently, 
\anonAdversary[2] can be invoked as a subroutine by \prfAdversary to win  
the PRF experiment with identical advantage. 


Summing up the advantages across all the experiments gives 

\begin{eqnarray*}
\probDiff{\anonExp[0]}{\thisHyb} \\ 
\le 2 \cdot \zkpHidingAdv + \hidingAdv + \PRFAdv
\end{eqnarray*}

A similar series of hybrid experiments show 

\begin{eqnarray*}
\probDiff{\anonExp[1]}{\thisHyb} \\ 
\le 2 \cdot \zkpHidingAdv + \hidingAdv + \PRFAdv
\end{eqnarray*}

Thus, from the above arguments, we have the result.

\end{proof}

\subsubsection{Authenticity of Balance}

\begin{theorem}
	The \ecash scheme described in \secref{sec:scheme} $\repCoinScheme$ ensures authenticity of balance assuming that the signature scheme \sigScheme is unforgeable under chosen message attack, the commitment scheme produces binding commitments, and the NIZK proof scheme has a polynomial time knowledge extractor 
	
	\begin{multline*}
		\frameAdv \le \ufcmaAdv + \bindingAdv + \\
 \dlogAdv + \zkpSoundnessAdv + \dblSpendAdv
	\end{multline*}
	
\end{theorem}

\begin{proof}

The adversary wins the experiment if one of the following cases takes place: i) forges a withdrawal receipt on behalf of an honest user $\genCar$, ii) does not provide a coin to the user after receiving the receipt, or iii) spends a coin that has been issued to an honest user. For case (i),  \frameAdversary[2] must forge a signature on the receipt under $\genCar$'s private key. This happens with probability at most \ufcmaAdv. 

For case (ii), the honest uses sends an abort withdrawal message $\langle \abortWithdrawal$ $\Sigma$, $\intentMsg$,$\genericVoucher$, $\nonce$, $\vPubKey{\vID}$, $\rsuPubKey{\rsuID} \rangle$. Consequently, the bank adds $\intentMsg$ and all other information to the invalid list \invalidList. If \frameAdversary[2] produces \genericTx spending \genericCoin, then there are two cases i) $\genericVoucher \neq \genericVoucherAlt$, and iii) $\genericVoucher = \genericVoucherAlt$. For the former case, when invoking \serverVerifyTx, the bank detects the double-issuing and outputs $\rsuPubKey{\rsuID}$, failing with probability at most $\dblSpendAdv$. The latter case is equivalent to 
\frameAdversary[2] spending a coin withdrawn by an honest user, which is described next.

For case (iii), note that $\genericTx$ includes the NIZK proof \zkp{\genericTx} proving that $\txPRFEval$ is computed using \pubKey{\vID}, and the spender knows the private key $\privKey{\vID}$ corresponding to $\pubKey{\vID}$. Thus, if $\frameAdversary[2]$ produces \genericTx using the same voucher, then one of the following conditions is true: i) $\frameAdversary[2]$ can de-commit $\privKeyCommitment$ to multiple distinct values, ii) $\frameAdversary[2]$ learns \privKey{\vID} from \pubKey{\vID}, or iii) $\frameAdversary[2]$ violates the soundness guarantees of the NIZK proof $\zkp{\genericTx}$. The probability of this happening is at most $\bindingAdv + \dlogAdv + \zkpSoundnessAdv$.

\end{proof}

\subsubsection{Proof of double spending detection}

\begin{theorem}
	
	The \ecash scheme described in \secref{sec:scheme} enables 
	double spending detection assuming that the commitment scheme produces binding commitments, the NIZK proof scheme has a polynomial time knowledge extractor and \hashFn is a collision resistance hash function, i.e., 
	
	\begin{multline*}
	\dblSpendAdv \le 3 \cdot \bindingAdv + 
	 2 \cdot \zkpSoundnessAdv + \crAdv
	\end{multline*}
	
\end{theorem}
\setcounter{ctr}{-1}
\begin{proof}

The experiment returns 1 if \dblSpendAdversary[2] has produced two valid transaction $\genericTx$ and $\genericTxAlt$, such that either a coin is double spent or double issued, and the algorithm does not return either $\pubKey{\vID}$ or $\rsuPubKey{\rsuID}$. We consider the cases when this happens. Assume that $\privKeyCommitment$ is a commitment to $\privKey{\vID}$, $\keyCommitment[1]$ is a commitment to $\prfKey[1]$ and $\rsuPrivKeyCommitment$ is a commitment to $\rsuPrivKey{\rsuID}$.

\medskip\noindent\underline{Case 1: $\txPRFEval \neq \grpGen^{\privKey{\vID}} \dyPRF[\userPRFKey](\coinRand)^{\txRand}$:}
Assume that $\txPRFEval = \grpGen^{\privKey{\vIDAlt}} \dyPRF[\userPRFKey'](\coinID)^{\txRand}$. For the transaction to be valid, either of the following conditions must have happened: i) \dblSpendAdversary[2] can de-commit \privKeyCommitment to both \privKey{\vID} and \privKey{\vIDAlt}, ii) $\dblSpendAdversary[2]$ produces a collision in 
the hash function such that $\crHash(\privKeyCommitment) = \crHash(\privKeyCommitment')$, or iii) \dblSpendAdversary[2] can violate the soundness guarantees of the NIZK proof $\zkp{\genericTx}$. This happens with probability at most $\bindingAdv + \crAdv + \zkpSoundnessAdv$. 
An identical reasoning exists for the case when $\txPRFEval' \neq \grpGen^{\privKey{\vIDAlt}} \dyPRF[\userPRFKey'](\coinID)^{\txRandAlt} $

\medskip\noindent\underline{Case 2: $\voucherDblSpendToken \neq \grpGen^{\rsuPrivKey{\rsuID}} \dyPRF[\prfKey[1]](0)^{\coinRand}$:}
Let $\voucherDblSpendToken = \grpGen^{\rsuPrivKey{\rsuIDAlt}} \dyPRF[\prfKey[1]'](0)^{\coinRandAlt}$ which means one of the following cases are true i) $\dblSpendAdversary[2]$ must have  de-commited \rsuPrivKeyCommitment to both $\rsuPrivKey{\rsuID}$ and $\rsuPrivKey{\rsuIDAlt}$, ii) \dblSpendAdversary[2] has de-committed $\keyCommitment[1]$ to both $\prfKey[1]$ and $\prfKey[1]'$, or iii) $\dblSpendAdversary[2]$ must have violated the soundness property. of the NIZK proof $\zkp{\coinID}$. This happens with probability at most  $2 \cdot \bindingAdv + \zkpSoundnessAdv$. An identical reasoning exists for the case 
$\voucherDblSpendToken' \neq \grpGen^{\rsuPrivKey{\rsuID}} \dyPRF[\prfKey[1]](0)^{\coinRandAlt}$

\end{proof}
\subsection{Proof of fair exchange}

\begin{theorem}
The Ecash scheme satisfies the fair exchange property defined by experiment $\Experiment{\fairexchangeLabel}{\repCoinScheme}(\fairexchangeComb)$ assuming that $\ro: \bin^{\ast} \rightarrow \bin^{\secParam}$ is a random oracle, \sigScheme satisfies unforgeability under chosen message attacks, and the PRF scheme produces outputs indistinguishable from random, i.e.,

\begin{eqnarray*}
\fairexchangeAdv \le 2 \PRFAdv + 2 \zkpHidingAdv + \ufcmaBlindAdv + \neglProb
\end{eqnarray*}

\end{theorem}

\begin{proof}
The adversary (\fairexchangeComb) wins the experiment if it can spend a coin \genericCoin without producing a valid receipt. Since in the invocation of \carGetCoin, the ATM provides the coin \genericCoin only after receiving a receipt from the user, \fairexchangeAdv can obtain the coin without receiving it from the \genRSU only if it can produce the coin from \intentMsg or \genericVoucher. Now, note that the components of $\genericVoucher \gets  \langle \privKeyCommitment, \zkp{\privKey{\vID}}, \voucherPRFEval, \voucherDblSpendToken, \zkp{\coinID}, \zkp{\rsuPrivKey{\rsuID}}, \coinRand \rangle $ that are not already known to \fairexchangeAdv can be simulated without any information of the coin: i) \voucherPRFEval and \voucherDblSpendToken are PRF evaluations under randomly chosen keys and so indistinguishable from random and ii) \zkp{\coinID} and  \zkp{\rsuPrivKey{\rsuID}} are zero-knowledge proofs and can be simulated assuming the zero-knowledge property. Therefore, \fairexchangeAdversary[2] obtains \genericCoin from \genericVoucher with advantage at most $\le 2 \PRFAdv + 2 \zkpHidingAdv$ 

Furthermore, since $\intentMsg = \ro(\genericCoin, \vPubKey{\vID}, \rsuPubKey{\rsuID})$, \fairexchangeAdversary[1] obtains \genericCoin from \intentMsg either if it has randomly queried $\ro$ on $\langle \genericCoin, \vPubKey{\vID}, \rsuPubKey{\rsuID} \rangle$ which happens with probability at most negligible probability in $\secParam$, or \fairexchangeAdversary[1] has produced a valid coin \genericCoin without interacting with the bank, which happens with probability at most \ufcmaBlindAdv. 
\end{proof}

\subsubsection{Proof of untraceability}

\begin{theorem}
	\rsuAssignCoin ensures untraceability of \repCoins assuming that the commitment scheme produces hiding commitments, the PRF scheme produces outputs indistinguishable from random, the NIZK scheme has a polynomial time knowledge extractor and the blind signature scheme produces indistinguishable signatures i.e., 
	
	\begin{multline*}
		\untraceAdv 
		\le 6 \cdot \hidingAdv + 4 \cdot \PRFAdv + \\ 4 \cdot \zkpSoundnessAdv + 2 \cdot \Advantage{\blindSigLabel}{\blindSigScheme}(\blindSigAdversary)
	\end{multline*}
	
\end{theorem}

\setcounter{ctr}{-1}
\begin{proof}

Consider the view of $\locPrivacyAdversary[2]$ in the 
experiment $\Experiment{\locPrivacyLabel[0]}{\repCoinScheme}(\locPrivacyAdversary[1],
     \locPrivacyAdversary[2])$

\begin{enumerate}
\item $\keyCommitment[1] \assign  \commFn(\prfKey[1][0])$
\item $\keyCommitment[2] \assign  \commFn(\prfKey[2][0])$
\item $\rsuPrivKeyCommitment \assign \commFn(\rsuPrivKey{\genRSUID[0]})$
\item $\signedCoin \assign \bsign(\hashFn(\keyCommitment[1], \keyCommitment[2], \rsuPrivKeyCommitment))$
	\item $\privKeyCommitment \assign \commFn(\vPrivKey{\vID}, \userPRFKey)$
    \item $\zkp{\privKey{\vID}} \assign \zprove{\sPubKeyBlindZKP}(\privKeyCommitment, \userCLSig{\privKey{\vID}})$ 
	\item $\coinRand \assign \crHash(\privKeyCommitment)$
    \item $\coinID \assign \coinRand + 1 \bmod \grpOrd$
    \item $\voucherPRFEval \assign \dyPRF[\prfKey[1][0]](\coinID)$  
\item $\voucherDblSpendToken \assign \dyPRF[\prfKey[2][0]](0)^{\coinRand}$
\item $\zkp{\rsuPrivKey{\genRSUID}} \assign \zprove(\rsuPrivKeyCommitment, \genericSig[\rsuPrivKey{\genRSUID[0]}])$
    \item $\zkp{\coinID} \assign \zkpGen[\alpha_1, \beta_1, \alpha_2, \beta_2, \gamma, \delta](\keyCommitment[1], \keyCommitment[2], \voucherPRFEval, \voucherDblSpendToken, \rsuPrivKeyCommitment, \coinID, \coinRand)$
\end{enumerate}

Also consider the view of $\locPrivacyAdversary[2]$ in the experiment $\Experiment{\locPrivacyLabel[1]}{\repCoinScheme}(\locPrivacyAdversary[1],
     \locPrivacyAdversary[2])$

\begin{enumerate}
\item $\keyCommitment[1]' \assign  \commFn(\prfKey[1][1])$
\item $\keyCommitment[2]' \assign  \commFn(\prfKey[2][1])$
\item $\rsuPrivKeyCommitment' \assign \commFn(\rsuPrivKey{\genRSUID[1]})$
\item $\signedCoin' \assign \bsign(\hashFn(\keyCommitment[1]', \keyCommitment[2]', \rsuPrivKeyCommitment'))$
	\item $\privKeyCommitment \assign \commFn(\vPrivKey{\vID}, \userPRFKey)$
    \item $\zkp{\privKey{\vID}} \assign \zprove{\sPubKeyBlindZKP}(\privKeyCommitment, \userCLSig{\privKey{\vID}})$ 
	\item $\coinRand \assign \crHash(\privKeyCommitment)$
    \item $\coinID \assign \coinRand + 1 \bmod \grpOrd$
    \item $\voucherPRFEval' \assign \dyPRF[\prfKey[1][1]](\coinID)$  
\item $\voucherDblSpendToken' \assign \dyPRF[\prfKey[2][1]](0)^{\coinRand}$
\item $\zkp{\rsuPrivKey{\genRSUID}}' \assign \zprove(\rsuPrivKeyCommitment, \genericSig[\rsuPrivKey{\genRSUID[1]}])$
    \item $\zkp{\coinID}' \assign \zkpGen[\alpha_1, \beta_1, \alpha_2, \beta_2, \gamma, \delta](\keyCommitment[1]', \keyCommitment[2]', \voucherPRFEval', \voucherDblSpendToken', \rsuPrivKeyCommitment', \coinID, \coinRand)$
\end{enumerate}

Note that $\privKeyCommitment, \zkp{\privKey{\vID}}, \coinRand, \coinID$ are identically distributed in both views and they do not provide any advantage to $\locPrivacyAdversary[1], \locPrivacyAdversary[2]$.  

Now consider the following series of experiments.

\medskip\nextHyb: This hybrid is identical to $\Experiment{\locPrivacyLabel[0]}{\repCoinScheme}(\locPrivacyAdversary[1],
     \locPrivacyAdversary[2])$ except that the commitments $\keyCommitment[1], \keyCommitment[2], \rsuPrivKeyCommitment$ are all commitments to 0. \locPrivacyAdversary[2]'s view in \thisHyb is indistinguishable to its view in $\Experiment{\locPrivacyLabel[0]}{\repCoinScheme}(\locPrivacyAdversary[1],
     \locPrivacyAdversary[2])$ assuming that the commitment scheme satisfies the hiding property, the NIZK proof ensures system ensures the zero-knowledge property and the blind signing scheme ensures the blinding property defined in $\Experiment{\blindSigLabel[\blindSigBit]}{\blindSigScheme}(\blindSigAdversary)$

\medskip\nextHyb: This hybrid is identical to \prevHyb, except that the PRF evaluations $\voucherPRFEval, \voucherDblSpendToken$ are replaced by random elements from \grp. \locPrivacyAdversary[2]'s view in \thisHyb is indistinguishable to its view in \prevHyb assuming that \dyPRF produces outputs indistinguishable from random. 

Summing up the advantages across all the hybrids gives 

\begin{eqnarray*}
\probDiff{\Experiment{\locPrivacyLabel[0]}{\repCoinScheme}(\locPrivacyAdversary[1],
     \locPrivacyAdversary[2])}{\thisHyb} \\
\le 3 \cdot \hidingAdv + 2 \cdot \PRFAdv + 2 \cdot \zkpSoundnessAdv + \Advantage{\blindSigLabel}{\blindSigScheme}(\blindSigAdversary)
\end{eqnarray*}

A similar series of experiments can be constructed to show that

\begin{eqnarray*}
\probDiff{\Experiment{\locPrivacyLabel[1]}{\repCoinScheme}(\locPrivacyAdversary[1],
     \locPrivacyAdversary[2])}{\thisHyb} \\
\le 3 \cdot \hidingAdv + 2 \cdot \PRFAdv + 2 \cdot \zkpSoundnessAdv + \Advantage{\blindSigLabel}{\blindSigScheme}(\blindSigAdversary)
\end{eqnarray*}

\end{proof}


\end{document}